\documentclass[
reprint,
superscriptaddress,
 amsmath,amssymb,
 prl,
 longbibliography
]{revtex4-1}

\usepackage{mathtools}

\DeclarePairedDelimiter\floor{\lfloor}{\rfloor}

\usepackage[usenames]{color}
\usepackage{amssymb}
\usepackage{amsmath} 
\usepackage{MnSymbol}

\usepackage{esdiff}
\usepackage{amsthm}
\usepackage{placeins}
\usepackage{makeidx}
\usepackage{url}
\usepackage{float}
\usepackage[caption = false]{subfig}
\usepackage[normalem]{ulem}
\usepackage{tabto}
\usepackage{multirow}
\usepackage{xcolor}
\usepackage{tikz-network}

\usepackage{graphicx}
\usepackage{dcolumn}
\usepackage{bm}
\usepackage[mathlines]{lineno}

\usepackage{etoolbox}
\makeatletter
\patchcmd\linenumberpar{\@LN@parpgbrk}{\penalty\@LN@parpgpen\relax}{}{}
\makeatother

\DeclareMathOperator*{\argmin}{arg\,min}

\newtheorem{theorem}{Theorem}[]
\newtheorem{corollary}{Corollary}[]

\newtheorem{remark}{Remark}[]

\begin{document}


\title{
The shape of memory in temporal networks\\
}

\author{Oliver E. Williams}
\affiliation{School of Mathematical Sciences, Queen Mary University of London, London, E1 4NS, United Kingdom}

\author{Lucas Lacasa}
\affiliation{School of Mathematical Sciences, Queen Mary University of London, London, E1 4NS, United Kingdom}

\author{Ana P. Mill\'an}
\affiliation{Amsterdam UMC, Vrije Universiteit Amsterdam, Department of Clinical Neurophysiology and MEG Center, Amsterdam Neuroscience, De Boelelaan 1117, Amsterdam, The Netherlands}

\affiliation{Institute \textit{Carlos I} for Theoretical and Computational Physics, University of Granada, Spain}

\author{Vito Latora}
\affiliation{School of Mathematical Sciences, Queen Mary University of London, London, E1 4NS, United Kingdom}

\affiliation{Dipartimento di Fisica ed Astronomia, Universit\`a di Catania and INFN, I-95123 Catania, Italy}

\affiliation{The Alan Turing Institute, The British Library, London NW1 2DB, United Kingdom}




\date{\today}

\maketitle
    {\bf Temporal networks \cite{Holme_rev12,masuda_guide_temp_net,holme2013temporal,holme2019temporal}
are widely used models for describing the architecture of complex systems 
\cite{Starnini:2013,Szell:2012aa,Yoneki:2009,corsi2018measuring,mazzarisi2019dynamic,millan2018concurrence,Valencia08,zanin2009dynamics,tang2010sw,lambiotte2019networks}.  Network memory --that is the dependence of a temporal
network's structure on its past-- has been shown to play a prominent
role in diffusion \cite{delvenne2015diffusion,Lambiotte_jcn15,masuda2013temporal, Scholtes_natcomm14}, 
epidemics \cite{Hiraoka:2018aa, takaguchi2013bursty,lambiotte2013burstiness,karsai2011small,Williams_2019,van2013non}
and other processes \cite{Fallani08,Singer_PLOSONE14} {occurring} over the network, and even to alter
its community structure \cite{peixoto2017modelling,Rosvall_natcomm14}.
Recent works have proposed to estimate the length of memory in a
  temporal network by using high-order Markov models~\cite{Scholtes_natcomm14,scholtes2017network,peixoto2018change}.
Here we show that network memory is 
inherently multidimensional and cannot be meaningfully reduced to a single
scalar quantity. Accordingly, we  introduce a mathematical framework  
for defining and efficiently estimating the microscopic shape of memory, 
which fully characterises how the activity of each link intertwines with
the activities of all other links.  
We validate our methodology on a wide range of
synthetic models of temporal networks with tuneable memory,
and subsequently study the heterogeneous shapes of memory emerging in
various real-world networks. 
}


\medskip
A temporal network is a graph whose structure changes over time.   
A temporal network $\mathcal{G}$ over
$N$ nodes can be 
{formalised} as a set of $L$ discrete-time stochastic processes
$\mathcal{G} = \{ {\mathcal{E} }^{\alpha} \}^{\alpha=1,2,\ldots,L}$, 
where $L\leq N(N-1)/2$ is the number of different pairs of nodes that
can be connected by links over time. Each 
${\mathcal{E} }^{\alpha}= \{ E^{\alpha}_t \}_{t=1,2,\ldots} $
is the stochastic process governing the dynamics of link $\alpha$,
with the random variable $E_t^\alpha$ taking the value 1 if link 
$\alpha$ is present at time $t$, and 0 otherwise.
Note that, in general, these stochastic processes are not independent. 
Indeed, the properties of a temporal network not only depend  
on the patterns of activities of each of its links, but also on the ways in
which these patterns influence each other across the network. 
Since the set of every possible graph with $N$ nodes is finite, it is
in principle possible to enumerate all the configurations of a
temporal network, build an alphabet accordingly, and transform the
temporal network into a time series of symbols from this alphabet. A
straightforward way to define a {\em scalar memory}
$\Omega(\mathcal{G})$ of $\mathcal{G}$ is then, by direct analogy to
the case of a scalar time series, as the order $p$ of the lowest-order
Markov chain that is able to reproduce the sequence of symbols
generated by $\mathcal{G}$ (see SI Section I-A and B and II-A for
details). This approach can only work in practice for very small numbers
of nodes $N$, as the size of the alphabet grows extremely rapidly
($\sim 2^{\frac{1}{2} (N^2 - N)}$) and very long time series would
be required for an accurate estimate.

\begin{figure*}
\includegraphics[width=0.9\textwidth]{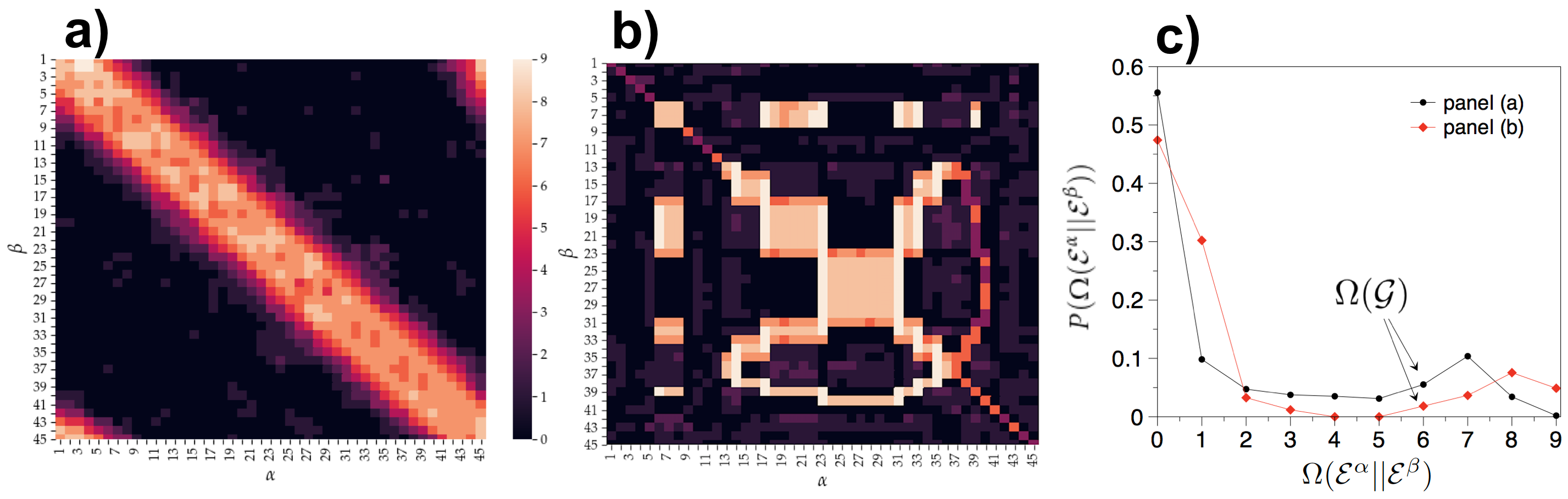}
\caption{{\bf Shape of memory and emergence of virtual loops in
    temporal networks with correlated link dynamics} A temporal fully
  connected network $\mathcal{G}$ with $N=10$ nodes and $L=45$ links
  whose dynamics are both autocorrelated and heterogeneously
  cross-correlated, generated from an eCDARN($p$) model with
  parameters $q=0.9$, $y=0.5$, $c=0.7$ and a set of memory lengths $p$
  randomly sampled with uniform (panel (a)) or a bimodal (panel (b))
  probability from $\{0, 1, \ldots, 6 \}$ 
  (see SI section IV-A for details).  {\bf (a)} The $45 \times 45$ entries of
  the co-memory matrix $\mathbb{M}$ (shown with a color code) display
  the shape of the network memory at the microscopic scale of pairs of
  links. In this specific case the eCDARN($p$) model is chosen such
  that the causal structure of link dependencies is restricted in a
  Bayesian ring topology of $L=45$ nodes, so that when link $\alpha$
  samples its activity from the past history of other links, it
  randomly samples from $\alpha \pm 1$.  The scalar memory of the
  network is
$\Omega(\mathcal{G})=6$. Pairs of neigbouring {or close} links in the Bayesian ring 
exhibit high memory co-order, often above $\Omega(\mathcal{G})$, due to the onset
of virtual loops (see the text), whereas distant links are seldom cross-correlated
and thus display low co-order memory. {\bf (b)} Similar to (a),
but where the link's causal structure is given by a different Bayesian graph
(see SI section IV-A for full details). A notably different memory shape emerges,
however the scalar memory of the network is still $\Omega(\mathcal{G})=6$.
 {\bf (c)} Distribution of memory co-orders in both examples, showing different  heterogeneous
profiles which in both cases are not well characterized
by $\Omega(\mathcal{G})$.  }
\label{fig:1}
\end{figure*}

\noindent
There is, however, a more fundamental problem with this approach.
Not only is the scalar memory order $\Omega(\mathcal{G})$ 
hard to estimate, but it also fails to capture fundamental microscopic
differences between temporal networks. As we will show below, each
temporal network is characterised by a precise pattern of memories at
a microscopic scale, that we name the {\em shape of the memory}. 
Links can heterogeneously influence the
activity of other links, and the entangled temporal dependencies among
these can even bring about {virtual memory resonances} in the activity of each link
which are systematically undetected by
$\Omega(\mathcal{G})$, yet have real and measurable physical effects on e.g. spreading dynamics. Overall, memory is indeed a
heterogeneous, multidimensional fingerprint which is not reducible to a scalar quantity.\\


\begin{figure*}
\includegraphics[width=0.9\textwidth]{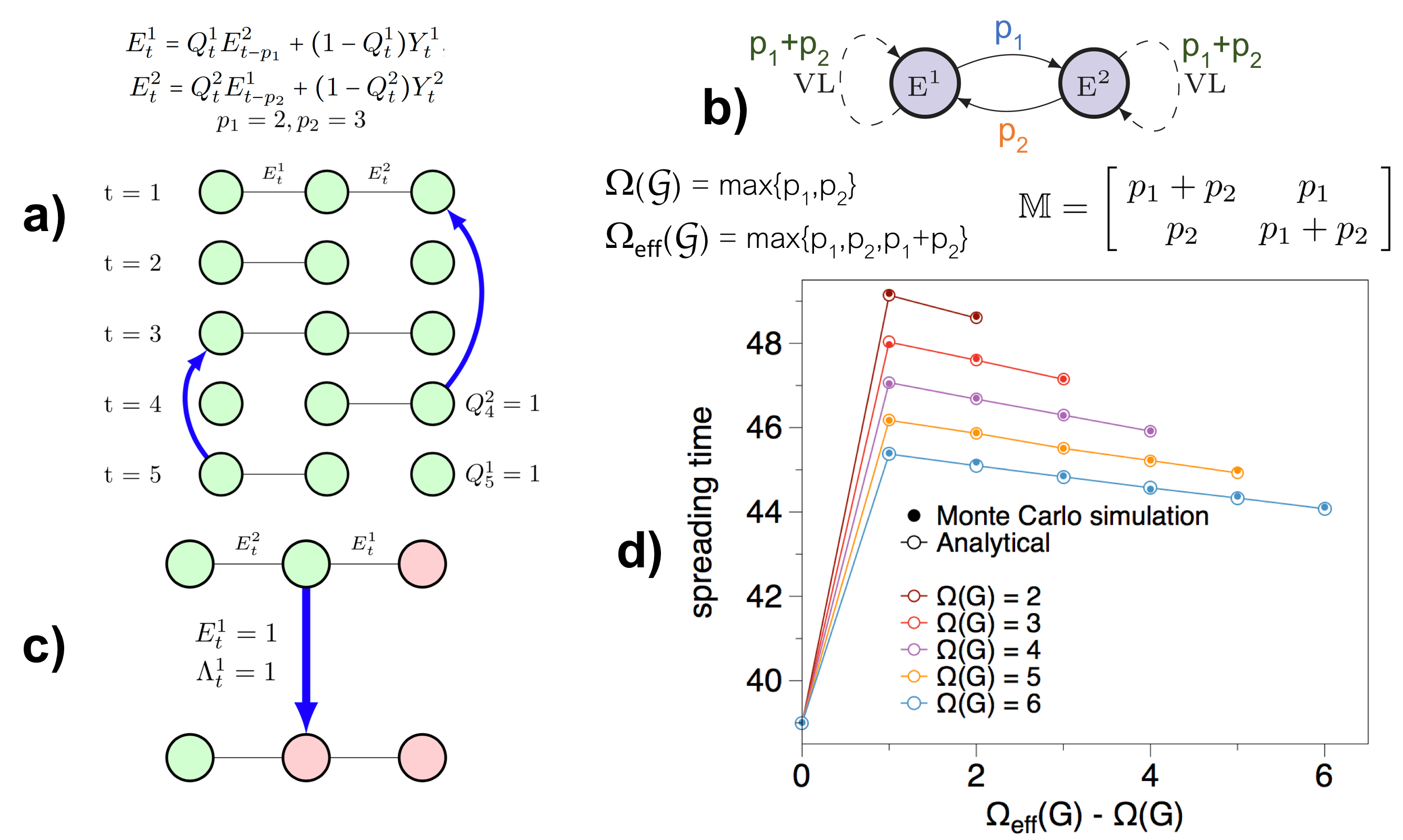}
\caption{{\bf Virtual loops and their effects on spreading processes 
    in temporal networks
  }  {\bf (a)} A sketch of five steps of the temporal evolution of a network
  with three nodes and two links evolving according to the
  displayed equation (for details see SI section IIIA). If 
  $Q^{\ell}_{t} = 0$ then link $\ell$ at time $t$ is generated randomly. 
  Conversely, if $Q^{\ell}_t = 1$ link $\ell$ copies a state from the past,
  namely link  $1$ copies the value of link 2 at time $t-2$, while link $2$
  copies the value of link 1 at time $t-3$.  {\bf (b)} Bayesian graph of
  the causal dependencies between the two links. Link 1 copies from
  the past of link 2 ($p_1$), whereas link 2 copies from the past
  of link 1 ($p_2$), thereby inducing first-order virtual loops in the
  memory of links 1 and 2, which virtually copy from their own past,
  $p_1+p_2$ steps back. The co-memory matrix $\mathbb{M}$ whose entries are 
  the co-memory order of each pair of links is also shown.  The
  scalar memory of the process can be proved to be $\Omega({\mathcal
    G})=\max\{p_1,p_2\}$, whereas the effective memory
  $\Omega_{\text{eff}}({\mathcal G})$, obtained as the largest entry of the
  co-memory matrix, is $p_1+p_2$, which differs from
  $\Omega({\mathcal G})$ due to the existence of the virtual loops. {\bf (c)} A
  SI (Susceptible-Infectious) epidemic spreading is defined over the temporal network. 
  Each node can either be in the infected (red) or susceptible (green) state. If at
  time $t$ there is a link $E^{\ell}_t = 1$ between an infected node
  and a susceptible one, then the infection will be passed with some
  probability (if random variable $\Lambda_t = 1$) and the
  susceptible node will become infected (see SI section IIIE for
  details).  {\bf (d)} Analytical and numerical results for the
  average time taken for every susceptible to become infected,  
  as a function of the difference between the scalar
  memory $\Omega(\mathcal{G})$ and the effective memory
  $\Omega_{\text{eff}}(\mathcal{G})$ (this latter being extracted from $\mathbb{M}$).
  For each curve, the value of
  $\Omega(\mathcal{G})$ is fixed to a given value $p_1$, while the 
  value of $\Omega_{\text{eff}}(\mathcal{G}) = p_1 + p_2 $ is varied
  by changing $p_2 \leq p_1$. We find that the spreading times depend on the value
  of the effective memory $\Omega_{\text{eff}}(\mathcal{G})$, which is then 
  a better descriptor of the effects of memory than $\Omega(\mathcal{G})$, this latter quantity
  being unable to detect any of these effects. Numerical results are
  obtained as averages over $10^7$ realisations of the network, and are in perfect agreement with the analytical prediction (see SI section IIIF for the full analysis).}
\label{fig:toy_si_spreading}
\end{figure*}


\noindent {\bf Theory --}
In order to fully characterise the shape of the memory of a temporal network, we
propose to define the {\em memory co-order} $\Omega ( {\mathcal
  E}^\alpha || {\mathcal E}^\beta) $ of a pair of links $\alpha$ and
$\beta$ 
as the furthest point in the past history of $\{ E^{\beta}_t \}$ 
which has influence on the current evolution of
$\{ E^{\alpha}_t \}$
(see SI section IIB for details).
Notice that for $\alpha=\beta$ we have 
  $\Omega (  {\mathcal E}^\alpha ||  {\mathcal E}^\beta) =
  \Omega (  {\mathcal E}^\alpha ||  {\mathcal E}^\alpha)  = \Omega( {\mathcal E}^\alpha)$. 
  The evaluation of the whole $L \times L $ co-memory matrix
$\mathbb{M}$,  
whose element $m_{\alpha  \beta}= 
\Omega (  {\mathcal E}^\alpha ||  {\mathcal E}^\beta)$ 
is the memory co-order of the pair of links $\alpha$ and $\beta$,  
allows us then to describe, at a microscopic level,  
the type of memory present in a network. 
As an example, Fig.\ref{fig:1}(a,b) displays the co-memory matrices
$\mathbb{M}$ we have extracted in the case of a synthetic temporal
network with $N=10$ nodes and $L=45$ links with two different types of correlated dynamics.
To compute the values of $m_{\alpha \beta}$ here (and throughout this work)
we have used a modified version of the 
Efficient Determination Criterion (EDC) \cite{zhao2001determination,dorea2014simulation} 
as this performs well as an estimator, is strongly consistent,
and allows for optimised implementations (for full details see SI sections I-C and II-C).
The network's memory exhibits a very peculiar {\em shape} induced by both the specific link dependence that we have planted and the pre-specified set of memory length parameters. This is further highlighted in the heterogeneous 
distribution of memory co-orders reported in Fig.\ref{fig:1}(c), where it is clear that
memory cannot indeed be characterised by
the value of the scalar memory $\Omega(\mathcal{G})$ alone, which in this case is equal to $6$ in both networks.
The long memory contributions (above order $6$) that we see in this distribution 
are a manifestation of what we call {\it virtual loops (VLs)}.
These emerge e.g. when link $\alpha$ depends on the past of link $\beta$, and
link $\beta$ in turn depends on the past of link $\alpha$, inducing a
long-memory loop in the activity of each link separately.
While being virtual in the sense that they are not pre-specified by
the model nor captured by $\Omega(\mathcal{G})$, they do indeed play an important role
in the dynamics of the network and affect processes occurring on it. These virtual loops typically emerge when causal dependencies of link
activities are described by a cyclic Bayesian network, and are indeed reminiscent of other forms of causal loops appearing in various fields of physics
and modern science (see SI section III-B and C for a discussion).
\\
\noindent 
To better illustrate this, Fig.\ref{fig:toy_si_spreading}(a) and (b)
show an example of a toy model of a temporal network with only three nodes
and two links. The model allows us to tune the shape of the memory, while the
scalar memory $\Omega(\mathcal{G})$ of the network is kept fixed. The adopted causal 
dependencies between the two links (each link can copy
from the past of the other link) induce virtual loops. 
These govern
the two diagonal terms of the co-order matrix $\mathbb{M}$ and have 
measurable and important effects on dynamical processes taking place over the network. 
Fig.\ref{fig:toy_si_spreading}(c) and (d) show that the time taken for an infection
to spread over the entire network can indeed be very different in 
networks with the same value of $\Omega(\mathcal{G})$, but with 
different memory shapes
 (see SI sections III and IV for other models, thorough mathematical analysis of virtual loops, and additional details). 

\noindent Furthermore, it is easy to prove (see theorem 1 in SI
section II-A and B) that $\Omega(\mathcal{G})\leq
\max_{\alpha,\beta}\{m_{\alpha
  \beta}\} :=\Omega_{\text{eff}}(\mathcal{G})$,
that is, the scalar memory is bounded from above by the maximum
co-order over all link pairs, which we term the {\em effective memory}
of the network. Fig.\ref{fig:toy_si_spreading}(d) shows that
$\Omega_{\text{eff}}(\mathcal{G})$ accounts for the virtual loops in
the toy network model and thus captures the measurable differences in the
spreading times. Of course, $\Omega_{\text{eff}}(\mathcal{G})$ is still not able to account for the rich memory heterogeneity of a temporal network (see panel (c) of Fig.\ref{fig:1}), but is (i) better conceptually defined than $\Omega(\mathcal{G})$ as it captures the effect of virtual loops, and (ii) can be 
computed efficiently from $\mathbb{M}$. In those cases where virtual loops are absent or they are {\it decoherent},   $\Omega(\mathcal{G})$ indeed approaches $\Omega_{\text{eff}}(\mathcal{G})$ (see SI section IV for a thorough exposition of virtual loop decoherence).
%
%

\begin{figure*}[htb]
\includegraphics[width=0.95\textwidth]{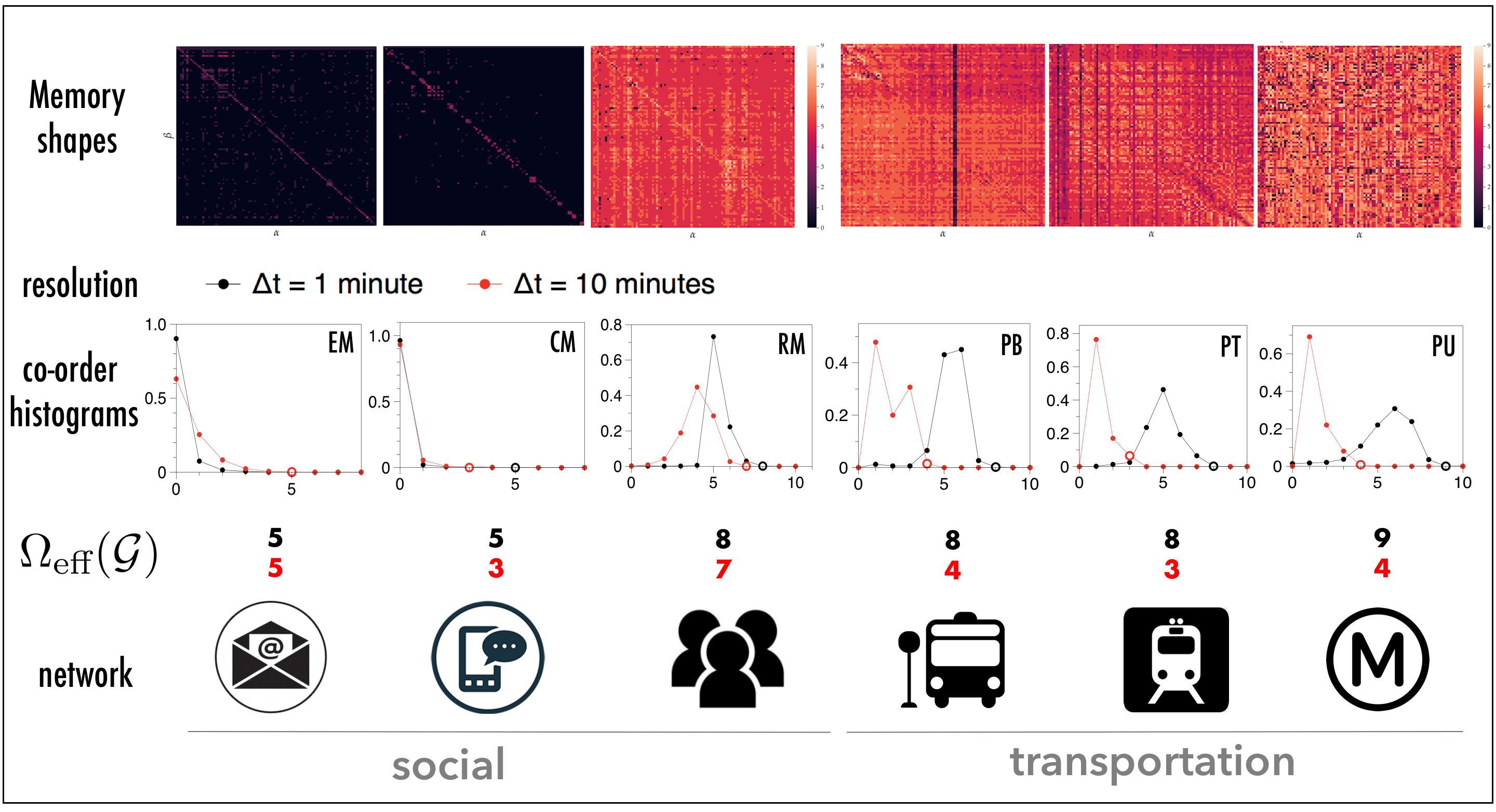}
\caption{{\bf The shape of memory of six real-world temporal
    networks:} 
  (EM) emails between the
  employees at a construction company \cite{email_bb}, (CM) text
  messages between college students \cite{panzarasa2009patterns}, (RM)
  social contacts at a US university from the Reality Mining
  experiment \cite{eagle2006reality}, routes taken by (PB) buses, (PT)
  overground trains and (PU) underground trains in the Paris public
  transport system \cite{kujala2018collection}. For each temporal
  network, we estimate the shape of the memory $\mathbb{M}$ restricted
  to the 100 most frequently active links, and plot the respective
  heat maps (lighter colour means higher memory). From these we
  extract the distributions of memory co-orders. The
  {effective} network memory $\Omega_{\text{eff}}(\mathcal{G})$ is
  also highlighted by hollow circles, and the actual values are reported 
  below the plots. Networks have been sampled at two different temporal
  resolutions $\Delta t$, namely every $1$ and $10$ minutes (heatmaps only show the $\Delta t=1$ min resolution). {In the two online
  social networks (EM and CM) the distributions of memory co-orders
  concentrate around zero and decay rapidly, indicating very short
  memory overall, except for a few pairs of links. In the offline university social network (RM) we find instead
  two clear peaks corresponding to the presence of memory at
  two time-scales of about 5 and 40-50 minutes (corresponding to interaction during lecture room changes 
  and during whole lectures) respectively.} 
{Two peaks are also observed in the three engineered networks. However,
both peaks are compatible with a time scale of 5-7 minutes in PT and PU,  
suggesting that such systems exhibit only one effective timescale, 
due to enforced planning and scheduling. The bus network (PB) in addition
to the 5-7 minutes also shows
a memory timescale of about 30 minutes, possibly due 
to external phenomena such as collective delays
induced by traffic jams.}
 }
\label{fig:pair_mem_dist}
\end{figure*}

\noindent {\bf Validation in synthetic networks -- } {We have tested
  the accuracy of our memory shape estimator in four generative
  temporal network models, each of varying complexity and with
  differing memory {shapes}. Whereas we lack analytical expressions
  for the memory shape, in all these models a ground truth for
  $\Omega(\mathcal{G})$ can be obtained analytically, so our analysis
  can examine the measurable effect of virtual loops. (i) First, we
  consider the DARN($p$) and eDARN($p_j$) models \cite{Williams_2019},
  where all links have independent yet autocorrelated dynamics. By
  design these network models are free from virtual loops, thus we
  expect $\Omega_{\text{eff}}(\mathcal{G}) \approx \Omega(\mathcal{G})
  $. (ii) Then, we consider the CDARN($p$) and eCDARN($p$) models
  \cite{williams2019_diff}, where link dynamics are not only
  autocorrelated but also cross-correlated, since links can sample
  their next state from either their own history or from the history
  of other chosen links. Virtual loops are expected to emerge in these
  cases, inducing
  $\Omega_{\text{eff}}(\mathcal{G})>\Omega(\mathcal{G}) $ (see
  Fig.\ref{fig:1} for an illustration of the eCDARN($p$) model and SI
  section IV-A for details of all four models and precise theorem
  statements and analytical derivations of the scalar memory). In
  every case we generate $10^3$ realisations of each temporal network
  model (fully connected backbone of $N=10$ nodes with randomly chosen
  ground truth scalar memory and a range of different parameter
  configurations) and count the hit rate (percentage of correct
  predictions) between the estimated $\Omega_{\text{eff}}(\mathcal{G})$
  and the analytical value of $\Omega(\mathcal{G})$. Results are reported 
  in SI Figure S8 and Section IV-B.  For long enough temporal series
  the hit rate is consistently 100\% in models which are free from
  virtual loops, suggesting that not only is our estimator accurate,
  but that in these cases 
  $\Omega_{\text{eff}}(\mathcal{G}) = \Omega(\mathcal{G})$. In those
  models where virtual loops emerge, the hit rate
  decays as expected. Interestingly, in a variety of cases a high hit
  rate is maintained due to the phenomenon of virtual loop decoherence
  (see SI Section IVB-D for full details and a more in depth analysis
  of the performance of our estimator).


\noindent {\bf Real-world temporal networks -- }
We have then studied the shape of memory in various real-world
temporal networks, including online and offline social interactions
and different types of transportation systems. Results are shown in
Fig.\ref{fig:pair_mem_dist} for the six networks. Only the
$10^2$ most active links for each network, i.e. those with largest
value of $\sum_t E_t^\alpha$, have been considered when constructing
the co-memory matrix $\mathbb{M}$.  The coloured heat maps in the top
row indicate that memory 
shapes vary across networks, overall being notably longer 
in offline
networks than in online ones. The middle panels show 
the distribution of memory co-orders. In order to
detect memory at different timescales, we have sampled the networks
at two temporal resolutions, namely $\Delta t=1$ and $\Delta t=10$
min (see SI Section V-A for details). 
The results should be interpreted accordingly: notice for instance,
that order 2 at the $\Delta t= 1$ resolution is equivalent to a memory
length of 2 minutes, whereas order 2 at the $\Delta t= 10$ resolution
is equivalent to a memory length of 20 minutes.
We have found that, while in the transportation networks with tight scheduling
only one memory timescale flags up, in the case of the bus
network, whose scheduling can be more affected by external factors such as
traffic jams, and in the case of the offline social network, at least 
at least two different memory timescales show up. The situation is particular clear
in the case of human contacts at university, which show memory lengths of 
5 and 40-50 minutes corresponding to different types of mechanisms of recurrent
social interactions during lectures and in between lecture room changes). 
\\
As a complement, a different projection of $\mathbb{M}$ into the so-called
$(\left< \Omega \right>_{\text{in}}^i, \left< \Omega
\right>_{\text{out}}^i)$ plane is considered in SI section V-C,
confirming that there is a notable difference between online and
offline temporal networks, with the former having on average weaker
and more homogeneous memory profiles than the latter, as well as a
difference between social and engineered ones.

\medskip
\noindent {\bf Conclusions -- } Our approach, based on the evaluation of
the co-memory matrix, not only provides a sound and efficient
approximation of the memory of a temporal network, but also offers a
comprehensible description of its microscopic shape.
Our results unveil previously hidden rich and heterogeneous memory shapes and indicate that fully considering this microscopic structure  
is of capital
importance when it comes to understanding how epidemics spread or
information diffuses in time-varying systems. We hope our work will
prompt further studies and will find useful applications 
in areas such as urban
mobility, epidemiology or information processing in neuroscience.

\bigskip
\noindent{\bf Data Availability.}
Data associated with this study can be found via the following links: 
manufacturing company e-mail communication - \url{www.ii.pwr.edu.pl/~michalski/},
text messages between college students - \url{snap.stanford.edu/data/CollegeMsg.html},
reality mining experiment - \url{http://realitycommons.media.mit.edu/realitymining.html},
public transport data - \url{www.nature.com/articles/sdata201889}

\medskip
\noindent {\bf Code Availability.} Complete implementations of our general method and all examples are available
in C++, Python 2.7, Python 3.6, Java, MatLab and Rust at \url{github.com/oewilliams/temp-net-memory} (see also SI section VI).

\medskip
\noindent {\bf Acknowledgments.} 
L.L. acknowledges support from EPSRC ECF EP/P01660X/1.
A.P.M. is supported by ZonMw and the Dutch Epilepsy Foundation, project number 95105006.
A.P.M. acknowledges support from the Spanish Ministry of Science and Technology and the ``Agencia Espa{\~n}ola de
Investigaci{\'o}n (AEI)''  under grant FIS2017-84256-P (FEDER funds), and from ``Obra Social La Caixa'' (ID 100010434 with code LCF/BQ/ES15/10360004).
V.L. acknowledges support from the EPSRC project
EP/N013492/1 and from the Leverhulme Trust Research Fellowship ``CREATE: the network components of creativity and success''.

\medskip
\noindent {\bf Author contributions.} OEW, LL and VL designed the research. OEW developed the formal analysis, implemented all algorithms and cleaned the data. APM contributed to the computational and formal analysis. All authors wrote the paper.

\medskip
\noindent {\bf Competing Interests statement} The authors declare no competing interests.
\vspace*{\fill}
\null










\pagebreak

\onecolumngrid

\begin{center}
  \textbf{\large Supplementary information for  ``The shape of memory in temporal networks"}\\[.2cm]
\end{center}

\setcounter{equation}{0}
\setcounter{figure}{0}
\setcounter{table}{0}
\setcounter{page}{1}
\renewcommand{\theequation}{S\arabic{equation}}
\renewcommand{\thefigure}{S\arabic{figure}}

\section{The memory of a time series}
\label{SI_ts}

\subsection{Random variables and entropies}

Let us consider a discrete random variable $X$ with sample space
$\mathcal{S}$ and probability mass function ${\mathbb P}(x)= {\rm Prob} \{ X=x \}$
with $x \in \mathcal{S}$. 
The entropy $H(X)$ of the random variable $X$ is defined in terms of the
probabilities ${\mathbb P}(x)$ of observing $x \in \mathcal{S}$ as: 
\begin{equation}
	H \left( X \right) = - \sum_{x} {\mathbb P} \left( x \right) \log {\mathbb P} \left( x \right).
	\label{shannon_ent}
\end{equation}
The definition of entropy can be extended to a pair or more discrete
random variables. 
Let $Y$ be a second discrete random variable with sample space
$\mathcal{S}'$. 
The joint entropy of the pair $X$, $Y$ is given by:
\begin{equation}
  H \left( X , Y \right) =  - \sum_{x,y} {\mathbb P} \left(x,y \right) \log {\mathbb P} \left( x ,y \right), 
\end{equation}
where ${\mathbb P}(x,y)= \rm{Prob} \left(X=x, Y=y \right)$  is the joint 
distribution of the two random variables. 
We can also define the entropy of the random variable $X$ when it is
conditioned on the second discrete random variable $Y$ as:   
\begin{equation}
  H \left( X | Y \right) =  \sum_{x} {\mathbb P}(x) H(X|Y=y) =
  - \sum_{x} {\mathbb P}(x)  \sum_{y} {\mathbb P}(y|x) \log p(y|x) =
  - \sum_{x,y} {\mathbb P} \left(x,y \right) \log {\mathbb P} \left( x | y \right).
\end{equation}

\subsection{Entropy and memory of a time series}

A time series ${\cal T}= \{X_t \}_{t=0,1,\ldots}$ or simply $\{X_t \}$ 
is a time-discrete stochastic 
process in which, at each time step $t$, with $t=0,1,2,\ldots$, 
the random variable $X_t$ takes values in state space ${\cal S}$.
We will indicate as $x_t$ the realization of random
variable $X_{t}$, i.e. the value taken at time $t$ by time series. 

The entropy rate $H$ of the time series $\mathcal{T}$ can be defined 
as:
\begin{equation}
	H \left( \mathcal{T} \right) = \lim_{n \to \infty} \frac{1}{n} H \left( X_0,...,X_n \right),
\end{equation}
and the conditional entropy of $\mathcal{T}$ as: 
\begin{equation}
	H' \left( \mathcal{T} \right) = \lim_{n \to \infty} H \left( X_n | X_{0},...,X_{n-1} \right) .
	\label{eqn:entropy}
\end{equation}
If the process $\mathcal{T}$ is strongly stationary, i.e. if its
joint probability distribution does not change when shifted
in time \cite{gagniuc2017markov}, 
then it can be proven that 
$H \left( \mathcal{T} \right) = H' \left( \mathcal{T} \right) $.  For
our purposes we will assume that this is always the case, allowing us
to study only conditional entropies. Here and in the following, for
the sake of simplicity, we introduce the following notation.  We
denote the sequence of random variables $X_0,X_1,.....,X_n$ as
$X_{0,n}$, and similarly for the realisations $x_0, x_1,..., x_n$ we
write $x_{0,n}$. Since $x_i \in {\mathcal S}~ \forall i,$ then we have $x_{0,n}\in
{\mathcal S}^{n+1}$. 

We then define the $n_{th}$ order block entropy $H_n$ of 
the process $\mathcal{T}$ as the 
entropy associated with the first $n+1$ random variables $X_0,\dots,X_n$: 
\begin{equation}
	H_n \left(\mathcal{T} \right) = - \sum_{x_{0,n}} {\mathbb P} \left( x_{0,n} \right) \log {\mathbb P}\left( x_{0,n} \right).
\end{equation}
Note that in particular $H_0$ coincides with the entropy of the marginal distribution of the first random variable $X_0$. Since $\mathcal{T}$ is stationary, $H_0$ is then the entropy associated with the marginal distribution of any of the random variables, in other words $H_0 = H(X_i), \forall i$.

Similarly,  for $n>0$, $H_n$ is the entropy of blocks of $n+1$ consecutive random variables (i.e. we are not required to consider the {\it first} $n+1$ random variables differently to any other set of $n+1$ consecutive random variables).
Of course, the entropy rate $H$ of the process $\mathcal{T}$ is then: 
\begin{equation}
 H \left( \mathcal{T} \right) = \lim_{n \to \infty} \frac{1}{n} 	H_n \left(\mathcal{T} \right) 
\end{equation}
Analogously, we can define the $n_{th}$ order conditional entropy $h_n$ as: 
\begin{eqnarray}
\begin{aligned}
	h_n \left(\mathcal{T}\right) =& H_n \left(\mathcal{T}\right) - H_{n-1} \left(\mathcal{T} \right)\\
	=& - \sum_{x_{0,n}} {\mathbb P} \left(x_{0,n} \right) \log {\mathbb P} \left(  x_n | x_{0,n-1} \right)\\
	=& - \sum_{x_{0,n}} {\mathbb P} \left(  x_n | x_{0,n-1} \right) {\mathbb P} \left(x_{0,n-1} \right) \log {\mathbb P} \left(  x_n | x_{0,n-1} \right).
\end{aligned}
\end{eqnarray}
so that: 
\begin{equation}
 H' \left( \mathcal{T} \right) = \lim_{n \to \infty} h_n \left(\mathcal{T} \right) 
\end{equation}

We are now ready to define the memory of the time series ${\mathcal
  T}$. Informally, the memory of ${\mathcal T}$ can be thought of as
the number of times steps into the past which have an influence on the
next observed value.  More formally, we can define the {\em memory length},
{\em memory order}, or simply {\em memory} $\Omega({\cal T})$ of stochastic
process $\mathcal{T} = \{ X_t \}$ as the order $p$ of the lowest-order
Markov chain that is able to reproduce the process, i.e. such that the
conditional probability mass functions satisfy:
\begin{equation}
  	\mathbb{P} ( x_t | x_0,x_1,\dots,x_{t-1}  ) = \mathbb{P}( x_t  | x_{t-p},\dots,x_{t-1})
	\label{gen_mem_def}
\end{equation}
for each $x_0,x_1,\dots,x_{t} \in {\cal S}^{t+1}$, or in compact form
$\mathbb{P} ( x_t | x_{0,{t-1}} ) = \mathbb{P}( x_t | x_{t-p,t-1})$
for each $x_{0,t} \in {\cal S}^{t+1}$.
This is equivalent to say that ${\mathcal T}$ can be identified as a $p$-th order
Markov chain, and we write:  $\Omega \left( \mathcal{T} \right) = p$: 
\begin{equation}
  \Omega ( {\mathcal T}  ) = \Omega ( \{ X_t  \}  ) 
  :=  \min_{p} \left [ p : \mathbb{P} ( x_t | x_{0,t-1} )
  = \mathbb{P} ( x_t | x_{t-p,t-1} ) \right] 
\end{equation}

\bigskip
It is now possible to relate $\Omega \left( \mathcal{T} \right)$ to 
the entropies we have introduced above. Observe first that $h_n$ is a 
monotonically non decreasing function of $n$, i.e. $h_n\geq h_{n-1} \forall n$.
Second, $h_n$ increases with $n$ until when $n$ 
is precisely equal to $p$, and remains constant thereafter,
i.e. $h_{p+d} ( \mathcal{T} ) = h_{p} (\mathcal{T} )$ for any positive
integer $d$. This can be easily demostrated, since  
for a $p_{th}$ order Markov chain $\mathcal{T}$, and for
$n = p + d$ with $d$ some positive integer, we can write:
\begin{eqnarray}
\begin{aligned}
	h_{p+d} \left( \mathcal{T} \right) =&  - \sum_{x_{0,p+d}} \mathbb{P} \left(  x_0 | x_{1,p+d} \right) \mathbb{P} \left(x_{1,p+d} \right) \log \mathbb{P} \left(  x_0 | x_{1,p+d} \right) \\
	=&  - \sum_{x_{0,p+d}} \mathbb{P} \left(  x_0 | x_{1,p} \right) \mathbb{P} \left(x_{1,p+d} \right) \log \mathbb{P} \left(  x_0 | x_{1,p} \right) \\
	=& -\sum_{x_{0,p}} \mathbb{P} \left(  x_0 | x_{1,p} \right)  \log \mathbb{P} \left(  x_0 | x_{1,p} \right) \sum_{x_{p+1,p+d}} \mathbb{P} \left( x_{1,p+d} \right) \\
	=& - \sum_{x_{0,p}} \mathbb{P} \left(  x_0 | x_{1,p} \right) \mathbb{P} \left(x_{1,p} \right) \log \mathbb{P} \left(  x_0 | x_{1,p} \right)\\
	=&  h_{p} \left( \mathcal{T} \right)
	\label{eqn:markov_ent}
\end{aligned}
\end{eqnarray}
These two conditions together mean that the $n_{th}$ order conditional entropy will be
at a maximum when we take $n=\Omega({\mathcal T})$, i.e. 
$\Omega({\mathcal T}) $ coincides with the minimum value of $n$ which maximises $h_n(\mathcal{T})$.

\noindent How should we compute the conditional entropy, and
  hence the memory, of a time series in practice? Ideally we would
need an infinitely long realisation of $\mathcal T$ to be able to
accurately estimate these {values}. For finite size time series, simply
maximising $h_n$ will not give a true estimate of the memory. Various
methods have been developed to overcome such a limitation, and thus
present consistent estimators for the order of a process when
observing a finite time series, which we will briefly discuss in
further sections.

\subsection{Estimating the memory}
We will now detail three examples of how to estimate the memory of a stochastic process from the information stored in finite time series.
These estimators all make use of the information theoretic framework that we have 
here detailed. The aim here is to overcome the problems associated with
only having a finite amount of data from which to estimate the memory via the
introduction of a ``penalty term".
Given an observed time series (i.e. a realisation of $\mathcal{T}$) with values
$(x_0,x_1,...,x_{T})$, where $x_k \in \mathcal{S}$ and the set has
$\left| \mathcal{S} \right| = m$ symbols, we start by to counting how many blocks of symbols of a given size are found in the time series. We label $n_{i_1,...i_s}$ as the total number of times a block of $s$ consecutive symbols with a specific arrangement for each of the $s$ entries appears in the time series, where $i_j \in {\mathcal S}, \ \forall j=1,\dots, s$. Specifically we have
\begin{equation}
	n_{i_1,...i_s} = \sum_{k=0}^{T-s+1} I(x_k = i_1 ,..., x_{k+s-1} = i_s), 
\end{equation}
where $I$ is the indicator function, so that $I(x_k = i_1 ,..., x_{k+s-1} = i_s) = 1$ 
when $ x_k = i_1 ,..., x_{k+s-1} = i_s$ is true, and zero otherwise.
We then define the following log likelihood function
\begin{equation}
	\log L_t = \sum_{i_1, ... ,i_{t+1}} n_{i_1, ... ,i_{t+1}} \log \frac{n_{i_1, ... ,i_{t+1}}}{n_{i_1, ... ,i_{t}}}.
\end{equation}
With this we can define the Akaike information criterion (AIC) \cite{tong1975determination}, Bayesian 
information criterion (BIC) \cite{schwarz1978estimating}, and 
the optimal form of the efficient determination criterion (EDC) \cite{dorea2014simulation,zhao2001determination}, 
as follows:
\begin{eqnarray}
	&&\text{AIC}(k) = -2 \log L_k + 2 m^k (m-1), \\
	&&\text{BIC}(k) = -2 \log L_k + m^k (m-1) \log T, \\
	&&\text{EDC}(k) = -2 \log L_k + 2 m^{k+1} \log \log T.
	\label{estimator_functions}
\end{eqnarray} 
Given some upper bound $K$ on the order of the time series
we define the corresponding estimators as
\begin{eqnarray}
	p_{\text{AIC}} = \argmin_{0 \leq k \leq K} \text{AIC}(k), \\
	p_{\text{BIC}} = \argmin_{0 \leq k \leq K} \text{BIC}(k), \\
	p_{\text{EDC}} = \argmin_{0 \leq k \leq K} \text{EDC}(k).
	\label{esitmator_values}
\end{eqnarray}
This gives us three alternatives for the estimation of the
true order of the observed process.
While the AIC and, to a lesser extent the BIC, are 
commonly used to estimate the order of a Markov chain
given observed data, the EDC has been shown to be, in some sense,
optimal. 
By optimal we mean that it is the strongly consistent estimator with the
fastest convergence to the true order.
It should be noted that the AIC, while the most popular,
is not a consistent estimator, but is included here for completeness.
In light of this, when estimating the order of a given stochastic process from the observation of a time series realisation, in this work
we will always use the EDC {(for completeness, we remind that there are also other possible alternatives, as the memory can be estimated in practice via different approaches, each with their own advantages}
\cite{tong1975determination,van1998testing,katz1981some,zhao2001determination,papapetrou2016markov,schwarz1978estimating}, but in this work we stick with the EDC for what said above).\\

\noindent So far all that is written above is well-known; in what follows we provide details of our proposed framework to extend the concept of memory to temporal networks.

\section{Defining and quantifying the memory of a temporal network}

\subsection{The scalar memory $\Omega(\mathcal{G})$ }

A temporal network $\mathcal{G}$ is formally defined by the
stochastic processes that generate the time evolution of its nodes and
links. For simplicity, we assume here that the set of nodes is fixed,
so that only the links of the network can change over time. We
indicate the number of nodes with $N$, and we label with the index
$\alpha \in \{1,2,\ldots L\}$ each of the $L$ node pairs that
can be connected over time.
In the most general case $L=N(N-1)/2$, however smaller
values of $L$ are enough for the adequate description of the system
if further topological restrictions are imposed
to the backbone of the temporal network. 
A temporal network $\mathcal{G}$ over the $N$ nodes can then be
  written as a set of $L$ {discrete-time} stochastic processes
  $\mathcal{G} = \{ {\mathcal{E} }^{\alpha} \}^{\alpha=1,2,\ldots,L} $. 
  For each link $\alpha$, with $\alpha=1,2,\ldots,L$,
${\mathcal{E} }^{\alpha}= \{ E^{\alpha}_t \}_{t=1,2,\ldots} $
  is the stochastic process which describes its dynamics. From now on, we
  will therefore indicate the network as {either
  $\mathcal{G} = \{ {\mathcal{E} }^{\alpha} \}^{\alpha=1,2,\ldots,L} $, }
  $\mathcal{G} = 
  \{ E_t^{\alpha} \}_{t=1,2,\ldots}^{\alpha=1,2,\ldots,L} $ or simply as
$\mathcal{G} = \{ {E}_t^{\alpha} \}$.\\  

\noindent We indicate as $e^{\alpha}_t$ the
value taken by the stochastic variable $E^{\alpha}_t$. 
For each $\alpha$ and each
$t$, $e^{\alpha}_t$ can only assume the value $1$, if link $\alpha$ is
present at time $t$, or $0$ otherwise. With  
${\bf e}_t  = e^1_t, e^2_t,\ldots, e^L_t $ we indicate the set of values taken by 
the $L$ stochastic variables at time $t$. In this way ${\bf e}_t$ completely
characterizes the state of the graph at time $t$. 

\noindent We can introduce a first definition of the memory of a temporal
network
$\mathcal{G} = \{ {E}_t^{\alpha} \}$
by direct
analogy to the case of a scalar time series discussed in Section \ref{SI_ts}. The
{\em scalar memory order}, or {\em scalar memory length}, or simply {\em scalar memory}
$\Omega(\mathcal{G})$ of the temporal network $\mathcal{G}$ can be defined
as the order $p$ of the lowest-order Markov chain able to reproduce
the process, i.e. the minimum value of $p$ such that: 
\begin{equation}
  \mathbb{P} ( {\bf e}_t | {\bf e}_0,{\bf e}_1,\dots,{\bf e}_{t-1}  ) = \mathbb{P}( {\bf e}_t  | {\bf e}_{t-p},\dots,{\bf e}_{t-1})
	\label{temp_net_mem_def}
\end{equation}
for each ${\bf e}_0,{\bf e}_1,\dots,{\bf e}_{t}$.
In compact form, this can be written as  
$\mathbb{P} ( {\bf e}_t | {\bf e}_{0,{t-1}} ) = \mathbb{P}( {\bf e}_t | {\bf e}_{t-p,t-1})$
for each ${\bf e}_{0,t}$, where by ${\bf e}_{0,t}$ we indicate
${\bf e}_0,{\bf e}_1,\dots,{\bf e}_{t}$. We can finally write:  
\begin{equation}
  \Omega ( \mathcal{G} )  =
   \Omega ( \{ {E}_t^{\alpha} \})
  :=  \min_{p} \left[ p :
\mathbb{P} ( {\bf e}_t | {\bf e}_{0,{t-1}} ) = \mathbb{P}( {\bf e}_t | {\bf e}_{t-p,t-1})
        \right]
\end{equation}
Since the space of every possible graph with $N$ nodes is finite, it
is in principle possible to enumerate all these graphs, build an alphabet 
accordingly, and transform a realization $\mathcal{G}$ into a time series of symbols
from this alphabet, from which the same methods to extract the memory 
order of a scalar time series can be used.

\subsection{The shape of the memory: the co-memory matrix $\mathbb{M}$
  and the effective memory $\Omega_{\text{\text{eff}}}(\mathcal{G})$}

However, since a temporal network is the set of stochastic processes
describing the dynamics of its $L$ links, together with the memory
order $\Omega(\mathcal{G})$ of the network as a whole, we may want
also to characterize the memory of each link separately, or even the
memory in the influence between two links.  Such a profile of the
memory at the microscopic levels of the links of the network is what
we will name the {\em shape of the memory} of the temporal network.
It is therefore convenient to introduce a novel concept to capture
the length of the memory of the mutual influence of two different
links. In order to do this, let us consider two stationary time series 
${\cal X} = \{X_t \}$ and ${\cal Y} = \{Y_t \}$. We can define 
the {\em memory co-order}, or simply {\em co-memory}
$\Omega( {\cal X}  || {\cal Y} )=\Omega(\{ X_t \} ||\{ Y_t \})$
of the random
process $\{ X_t \}$ with respect to $\{ Y_t \}$, as the
furthest point in the past history of $\{ Y_t \}$ which has influence
on the value taken by $X_t$ at time $t$, i.e. as the lowest value of $p$
such that $\mathbb{P} ( x_t | y_{t-1,t-p} ) = \mathbb{P} ( x_t | y_{t-1,-\infty} ) $
where $y_{0,t}$ is shorthand for the sequence $y_{0},y_{1},...,y_{t}$. We can
thus write:
\begin{equation}
  \Omega( {\cal X}  || {\cal Y} ) = \Omega (\{ X_t \}  ||\{ Y_t \}) :=  \min_{p} \left[ p : \mathbb{P} ( x_t | y_{0,t-1} )
  = \mathbb{P} ( x_t | y_{t-p,t-1} ) \right], 
\end{equation}
\noindent
for each $y_{0,t-1}$ and each $x_t$.  
Since we consider processes to be stationary, this value is invariant of $t$.
Note that if the two processes ${\cal X}$ and ${\cal Y}$ are the same, then
$\Omega ( {\cal X} || {\cal Y} ) = \Omega ( {\cal X} || {\cal X} )
= \Omega( {\cal X})$.

\bigskip
We can now define the {\em memory co-order} or simply
{\em co-memory} $\Omega ( {\mathcal E}^\alpha ||  {\mathcal E}^\beta)$ 
of the random process 
$  {\mathcal E}^\alpha =  \{ E^{\alpha}_t \}_{t=1,2,\ldots}$,
representing edge $\alpha$, with respect to a second
random process $  {\mathcal E}^\beta =  \{ E^{\beta}_t \}_{t=1,2,\ldots}$,
representing edge $\beta$, as the furthest point in
the past of $E_t^{\beta}$ which has influence on the future evolution of $E_t^{\alpha}$:
\begin{equation}
	\Omega ( {\mathcal E}^\alpha ||  {\mathcal E}^\beta)=
        \Omega ( \{ E^{\alpha}_t \}_{t=1,2,\ldots} ||  \{ E^{\beta}_t \}_{t=1,2,\ldots} )
    :=  \min_{p} \left[ p : \mathbb{P} ( e_t^{\alpha} | e_{t-1,t-p}^{\beta} ) = \mathbb{P}( e_t^{\alpha} | e_{t-1,-\infty}^{\beta} ) \right], 
\end{equation}
where the subscript $e_{t-1,t-n}^{\beta}$ is shorthand for the sequence
$e_{t-1}^{\beta},e_{t-2}^{\beta},...,e_{t-n}^{\beta}$.\\
The memory co-order can be evaluated for each couple of links $\alpha$ and $\beta$,
so that the whole memory shape of the temporal network can be fully
characterised by the {\em co-memory matrix} $\mathbb{M}$, a $L \times L$ 
matrix with entries  
$\mathbb{M}_{\alpha \beta} = \Omega (  {\mathcal E}^\alpha ||  {\mathcal E}^\beta) $,
which accounts for the set of all possible co-orders.
Notice that, by construction, the memory co-order of a link $\alpha$
with itself is precisely the memory order of the link, that is
{ $ \Omega (  {\mathcal E}^\alpha ||  {\mathcal E}^\alpha) 
  = \Omega(\{ E_t^{\alpha} \}_t )$}, so that matrix  $\mathbb{M}$ contains in the
diagonal entries information on the memory present in the evolution
of each link considered as an independent dynamical process from the rest
of the network.\\ 

\noindent Based on matrix $\mathbb{M}$ it is now possible to introduce another 
scalar projection of the network's memory, which we label {\em effective memory} 
$\Omega_{\text{eff}}(\mathcal{G})$, as the
maximum value of the co-orders over the network link pairs:
\begin{equation}
  {	\Omega_{\text{eff}}(\mathcal{G}) := \max_{\alpha,\beta} \left[
\Omega(  {\mathcal E}^{\alpha} ||  {\mathcal E}^{\beta}) \right]
}
\end{equation}

It is easy to prove the following theorems, which relate the effective
memory to the memory of the network as a whole. 

\begin{theorem}

Given a temporal network $\mathcal{G}$ with $L$ edges and edge
processes $\{ \mathcal{E}^{\alpha} \}$ for $\alpha = 0,...,L$, we have: 
\begin{equation}
	\Omega(\mathcal{G}) \leq \max_{i,j} (\Omega( \mathcal{E}^{\alpha} || \mathcal{E}^{\beta})) =: \Omega_{\text{\text{eff}}}(\mathcal{G}),
	\label{SI:mem_bound_eqn}
\end{equation}
where $\Omega(\mathcal{G})$ is the memory of the temporal network, 
$\Omega( \mathcal{E}^{\alpha} || \mathcal{E}^{\beta})$ are the memory co-orders and
$\Omega_{\text{eff}}(\mathcal{G})$ is the effective memory.
\end{theorem}

\proof{It is useful to introduce the following compact notation:
  $E^{1,L}_{0,n} \equiv E^1_{0,n},..., E^L_{0,n} $ 
  and $e^{1,L}_{0,n} = e^1_{0,n},...,e^L_{0,n} $.
  The $n_{th}$ order conditional entropy $h_n (\mathcal{G} ) = h_n ( \{ {E}_t^{\alpha} \})$ 
can be written as: 
\begin{equation}
	h_n (\mathcal{G}) = - \sum_{e_{0,n}^{1,L}} \mathbb{P}(e_{0,n}^{1,L}) \log \frac{\mathbb{P}(e_{0,n}^{1,L})}{\mathbb{P}(e_{1,n}^{1,L})}.
\end{equation}
Expanding the joint probabilities in terms of conditional
probabilities over the different links: 
\begin{equation}
	\mathbb{P}(e_{0,n}^{1,L}) = \prod_{\alpha=1}^{L} \mathbb{P}(e_{0,n}^{\alpha} | e_{0,n}^{1,\alpha-1} )
\end{equation}
we can now write an expression for $h_n (\mathcal{G} )$ in
terms of the contributions from the different links:
\begin{eqnarray}
\begin{aligned}
	h_n (\mathcal{G}) =&  - \sum_{e_{0,n}^{1,L}} \prod_{\gamma=1}^{L} \mathbb{P}(e_{0,n}^{\gamma} | e_{0,n}^{1,\gamma-1} ) \sum_{\alpha=1}^{L} \log \frac{\mathbb{P}(e_{0,n}^{\alpha} | e_{0,n}^{1,\alpha-1} )}{\mathbb{P}(e_{1,n}^{\alpha} | e_{1,n}^{1,\alpha-1} )}\\
	=& - \sum_{\alpha=1}^{L} \sum_{e_{0,n}^{1,L}} \mathbb{P}(e_{0,n}^{\alpha} | e_{0,n}^{1,\alpha-1} ) \log \frac{\mathbb{P}(e_{0,n}^{\alpha} | e_{0,n}^{1,\alpha-1} )}{\mathbb{P}(e_{1,n}^{\alpha} | e_{1,n}^{1,\alpha-1} )} \prod_{\gamma \neq \alpha} \mathbb{P}(e_{0,n}^{\gamma} | e_{0,n}^{1,\gamma-1} ).
\label{eqn:gen_ent}
\end{aligned}
\end{eqnarray}
We now note that the 
memory $p$ must be consistent with equations \ref{gen_mem_def} and \ref{temp_net_mem_def}
and hence must be the minimum value of $n$ which maximises $h_n (\mathcal{G})$.
Hence, let us define our prospective effective memory $p$, which we will then show to be an upper bound, as: 
\begin{equation}
	p = \max_{\alpha,\beta} \left[ \Omega(\mathcal{E}^{\alpha} || \mathcal{E}^{\beta} )  \right].
	\label{eqn:co_ord_def}
\end{equation}
To test that this value of $p$ is indeed a point for which
we obtain the maximum value for our conditional entropy
 we take a value of $ n \geq p$ so that $n = p + d$. 
 The entropy is then given by
\begin{eqnarray}
\begin{aligned}
	h_n (\mathcal{G}) =&  - \sum_{\alpha=1}^{L} \sum_{e_{0,n}^{1,L}} \mathbb{P}(e_{0}^{\alpha} | e_{1,p}^{\alpha}, e_{0,p}^{1,\alpha-1} ) \mathbb{P}( e_{1,n}^{\alpha} ) \log \frac{ \mathbb{P}(e_{0}^{1,\alpha} | e_{1,p}^{1,\alpha})}{ \mathbb{P}(e_{0}^{1,\alpha-1} | e_{1,p}^{1,\alpha-1})} \prod_{\gamma \neq \alpha} \mathbb{P}(e_{0}^{\gamma} | e_{1,p}^{\gamma} , e_{0,p}^{1,\gamma-1} ) \mathbb{P}(e_{1,n}^{\gamma}), \\
	=& h_p (\mathcal{G}) \sum_{e_{p+1,p+d}^{1,L}} \mathbb{P}(e_{p+1,p+d}^{\alpha}) \prod_{\gamma \neq \alpha} \mathbb{P}(e_{p+1,p+d}^{\gamma}), \\
	=& h_p (\mathcal{G}) \prod_{\alpha} \sum_{e_{p+1,p+d}^{\alpha}} \mathbb{P}(e_{p+1,p+d}^{\alpha}), \\
	=& h_p (\mathcal{G}).
\end{aligned}
\end{eqnarray}
This shows us that 
$h_n (\mathcal{G}) = h_p (\mathcal{G})$, and so $\Omega(\mathcal{G}) \leq p$.
Since we defined $p$ to be the maximum value of the possible co-orders of the 
network, and hence the effective memory, we have now proved that 
$\Omega(\mathcal{G}) \leq \max_{\alpha,\beta} \left[ \Omega( \mathcal{E}^{\alpha} || \mathcal{E}^{\beta}) \right] =: \Omega_{\text{eff}}(\mathcal{G})$.
That is to say, the memory of the temporal network 
is bounded above by the furthest time into the past of any link
that has influence on the evolution of any other link.\qed}\\

\begin{corollary}
  If for every pair of links $\alpha$ and $\beta$, with $\alpha \neq \beta$,
  the two stochastic processes $ \mathcal{E}^{\alpha}$ and $\mathcal{E}^{\beta} $ are
  independent, 
  i.e. {$\forall \tau \, \mathbb{P}(e_{1,\tau}^{\alpha} , e_{1,\tau}^{\beta}) = \mathbb{P}(e_{1,\tau}^{\alpha})
  \mathbb{P}(e_{1,\tau}^{\beta})$,} 
  then
\begin{equation}
	\Omega(\mathcal{G}) = \Omega_{\text{eff}}(\mathcal{G}).
\end{equation}
\label{corr:independent_mem}
\end{corollary}
\proof{The proof can be obtained by looking at Eq.~\ref{eqn:gen_ent}
  directly. When links are independent, the conditional probabilities
  {become} $ \mathbb{P}(e_{0,n}^{\alpha} |
  e_{0,n}^{1,\alpha-1} ) = \mathbb{P}(e_{0,n}^{\alpha})$. This allows
  us to write: 
\begin{equation}
	h_n (\mathcal{G}) = \sum_{\alpha=1}^{L} h_n( \mathcal{E}^{\alpha}).
\end{equation} 
Clearly from this, combined with Eq.~\ref{eqn:markov_ent}, we must have
$\Omega(\mathcal{G}) = \max_{\alpha} \left[ \Omega(\mathcal{E}^{\alpha})\right]$.
But since links are independent,  
$\Omega(\mathcal{E}^{\alpha} || \mathcal{E}^{\beta}) = \Omega(\mathcal{E}^{\alpha} ||
\mathcal{E}^{\beta} ) \delta_{\alpha, \beta}$, and so:
\begin{equation}
 \Omega(\mathcal{G})= \max_{\alpha} \left[ \mathcal{E}^{\alpha} \right] = \max_{\alpha,\beta} \left[ \Omega(\mathcal{E}^{\alpha} || \mathcal{E}^{\beta}) \right] =: \Omega_{\text{eff}}(\mathcal{G}),
\end{equation}
{exactly as we aimed} to prove.\qed}\\

\begin{corollary}
  In the most general case of a temporal network $\mathcal{G}$, 
  the memories of its links as given by the values $\Omega( \mathcal{E}^{\alpha})$,
  with $\alpha=1,\ldots,L$, do not provide an upper or lower bound for
  $\Omega(\mathcal{G})$. That is to say, in general: 
\begin{eqnarray}
	\Omega(\mathcal{G})  \nleq \max_{\alpha} \left[ \Omega(\mathcal{E}^{\alpha}) \right], \\
	\Omega(\mathcal{G})  \ngeq \max_{\alpha} \left[ \Omega(\mathcal{E}^{\alpha}) \right],
\end{eqnarray}
except in special cases, such as the one considered in the previous corollary.
\end{corollary}

\proof{It is sufficient to provide two examples:  
one in which the maximum memory of the links is less than the memory of the network,
and one in which it is greater.

\noindent First we look at a case in which the maximum link memory is
less than the network memory. Consider $\mathcal{G}$ formed by two
links ruled by the stochastic processes $\{ E^1_t \}$ and $\{ E^2_t
\}$.  Now assume $E^2_t \sim Bernoulli(y)$ and $E^1_t$ is drawn from a
modified DAR($p$) process so that $E^1_t = Q_t E^2_{t-Z_t} +
(1-Q_t)Y_t$, where $Q_t \sim Bernoulli(q)$, $Y_t \sim Bernoulli(y)$
(with the same value of $y$ as $E_t^2$) and $Z_t \sim
Uniform(1,p)$, for some values of $q$ and $p$. Then $\Omega(\{ E_t^1 \}) =
0$ and $\Omega(\{ E_t^2 \}) = 0$, even though $\Omega(\mathcal{G}) = p$.
 
\noindent Second, we consider a case in which the maximum link memory is greater
than the network memory. Consider again $\mathcal{G}$ with 
two links ruled by $\{E^1_t\}$ and $\{E^2_t\}$. Assume now  
$E^1_t = Q_t E^2_{t-Z_t} + (1-Q_t)Y_t$ and $E^2_t = Q_t E^1_{t-Z_t} + (1-Q_t)Y_t$,
where $Q_t \sim Bernoulli(q)$, $Y_t \sim Bernoulli(y)$. 
As in the previous case the memory of the network is $p$. however, if we
substitute our expression for $E_t^2$ into our formula for $E^1_t$, then 
we obtain
$E_t^1 = Q_t^1 (Q_t^2 E_{t-(Z^1_t + Z^2_t)}^1 + (1-Q_t^2) Y_t^2) + (1-Q_t^1) Y_t^1$. 
Clearly then $\Omega(\{ E_t^1 \}) = 2p$, which is greater than
$\Omega(\mathcal{G})$. We will further explore the reasoning
for this in later sections.
\qed}\\

\noindent {Conceptually speaking, observe that the effective memory is similar to Granger causality
in that it considers the influence of past states of a process on the
present evolution of another process \cite{granger1969investigating}.
However, they are not the same. Indeed, in the toy models that we
introduce later on to demonstrate the cases where the scalar memory $\Omega\left(
\mathcal{G} \right)$ does not capture the influence of memory on
spreading processes, it can also be seen that Granger causality
suffers from the same issue, as here the extent of any causality
between the two links is precisely the scalar memory.}

\subsection{Estimating the co-memory matrix $\mathbb{M}$}
We have shown that the memory of a temporal network can be understood
in terms of the co-orders of its links, which represent the memory that one link has
of another. As in the main text, we define the co-order { $\Omega( \mathcal{X} || \mathcal{Y} )$ of a link
process $\mathcal{X}$ composed of random variables $X_t$ with realisations $x_t$, 
with respect to link process $\mathcal{Y}$ composed of random variables $Y_t$ with realisations $y_t$}, as 
\begin{equation}
\Omega( \mathcal{X} || \mathcal{Y} ) =  \min_{p} \left[ p : \mathbb{P} ( x_t | y_{t-1,t-p} ) = \mathbb{P} ( x_t | y_{t-1,-\infty} ) \right].
\end{equation}
What remains to be found however is a way of estimating this value.
We will here adapt the form of the efficient determination criterion (EDC) given in 
Eq.~\ref{estimator_functions} and \ref{esitmator_values}.
Firstly, it is clear that both the number of states $m$ and the number of observations $T$
are consistent with their applications to sequences in general. Specifically $m$ will be 
2 (since links are either present or not) and $T$ is defined by the data being used.
This means that we now need only focus our attentions on the log-likelihood function
$\log L_k$.
Again, given a pair of sequences $\{X_t\}$ and $\{Y_t\}$ with realisations $x_t$ and $y_t$ respectively,
and where $t = 0,...,T$, the likelihood  
of observing the full sequence $X_t$ given the last $k$ values of $Y_t$, denoted by
$L(X_t | Y_{t-1,t-k})$ is given by
\begin{eqnarray}
\begin{aligned}
	L(X_t | Y_{t-1,t-k}) =& \prod_{i=0}^{T} \mathbb{P}(x_i | y_{i-1,i-k}),\\
	=& \prod_{x_i} \prod_{y_{i-1,i-k}} \mathbb{P}(x_i | y_{i-1,i-k})^{n(x_i,y_{i-1,i-k})}.
\end{aligned}
\end{eqnarray}
Where, similarly to before, $y_{i-1,i-k}$ is the joint $y_{i-1},y_{i-2},...y_{i-k}$, and
the counting function $n(x_i,y_{i-1,i-k})$ is defined by
\begin{equation}
	n(x_i,y_{i-1,i-k}) = \sum_{i=k+1}^{T} I(X_i = x_i,Y_{i-1}=y_{i-1} ,..., Y_{i-k} = y_{i-k}).
\end{equation}
Taking the empirical estimate for the conditional
\begin{equation}
	\mathbb{P}(x_i | y_{i-1,i-k}) \approx \frac{n(x_i,y_{i-1,i-k})}{n(y_{i-1,i-k})},
\end{equation}
we can then write the log-likelihood as
\begin{equation}
	\log L_k =   \sum_{x_i} \sum_{y_{i-1,i-k}} n(x_i,y_{i-1,i-k}) \log  \frac{n(x_i,y_{i-1,i-k})}{n(y_{i-1,i-k})}.
\end{equation}
Now, precisely as before, we obtain the estimator
\begin{equation}
	\text{EDC}(k) = -2 \log L_k + 2 m^{k+1} \log \log T,
\end{equation}
giving co-order estimate
\begin{equation}
	p_{\text{EDC}} = \argmin_{0 \leq k \leq K} \text{EDC}(k).
\end{equation}
An example where we show explicitly the value taken by the estimator 
$\text{EDC}(k)$ for a concrete co-order {$\Omega(\mathcal{E}^1||\mathcal{E}^1)$} is 
depicted in Fig.\ref{fig:ACF_EDC}. In that figure we also plot the 
autocorrelation function of the signal $\{E^1_t\}$, defined in the usual way 
ACF($\tau$)$\propto \langle E^1_t \cdot E^1_{t+\tau}\rangle_t$.

\subsection{Pair Memory $\Omega_{\text{pair}}(\mathcal{G})$}
Here we introduce an alternative approach to estimating the scalar memory of a temporal network.
We do this by taking pairs of links
and analysing the memory of the pair as if it were its own temporal network.
By this we mean that, if we had two links $\mathcal{E}^{\alpha}$ and $\mathcal{E}^{\beta}$ ($\alpha \ne \beta$) then rather than 
looking at the set $\{ \Omega(\mathcal{E}^{\alpha} ), \Omega(\mathcal{E}^{\beta} ), \Omega \left( \mathcal{E}^{\alpha} || \mathcal{E}^{\beta} \right) ,  \Omega \left( \mathcal{E}^{\beta} || \mathcal{E}^{\alpha} \right) \}$,
we look at the order of the random vector directly: $\Omega \left( (\mathcal{E}^{\alpha} , \mathcal{E}^{\alpha}) \right)$.\\

\noindent An advantage to this alternative approach is that, while it is possible to measure the co-orders of pairs of links directly, this does not immediately allow us to make use of the great body of 
work that has been done on the estimation of the memory of a general symbolic sequence.
We also show that in general this is a better estimate of the scalar memory than the effective memory. \\

\noindent From \ref{SI:mem_bound_eqn} we know that the memory of a temporal network 
$\mathcal{G}$
with generating edge processes $\mathcal{E}^{\alpha}$ is bounded above by 
\begin{equation}
	\Omega ( \mathcal{G} ) \leq \max_{\alpha,\beta} \left( \Omega(\mathcal{E}^{\alpha} || \mathcal{E}^{\beta} ) \right)=\max_{\alpha,\beta} \left[ \mathbb{M}_{\alpha \beta} \right]=:\Omega_{\text{\text{eff}}}(\mathcal{G}) .
\end{equation}
Now, consider two links from this network: $\mathcal{E}^1$ and $\mathcal{E}^{2}$.
In isolation they form their own temporal network $\mathcal{G}^{1,2}$
whose scakar memory is bounded above by $\max \{ \Omega(\mathcal{E}^{1} ), \Omega(\mathcal{E}^{2} ), \Omega \left( \mathcal{E}^{1} || \mathcal{E}^{2} \right) ,  \Omega \left( \mathcal{E}^{2} || \mathcal{E}^{1} \right)  \}$.
Since this new temporal network contains only two links, 
it can only exist in 4 possible states
$(e_t^1,e_t^2) \in \{ (0,0), (0,1), (1,0), (1,1) \}$.
This is not an unreasonably large state space, and so we can estimate the memory 
directly. All that remains to do is index this state space.
To do this let us first detail a more general concept:
given two time series $\{X_t\} , \{Y_t\} \in \{0,1\}$ with realisations $x_t$ and $y_t$ respectively, define the product 
 time series $\{Z_t\} \in \{0,3\}$ with realisations $z_t$, as $z_t = f(x_t,y_t)$
 where $f: \{0,1\} \times \{0,1\} \to \left[0,3\right]$ is any bijection
(for example $z_t = x_t + 2 y_t$ ). Let us denote $Z_t = f(X_t,Y_t) \, \forall t$.
Let us now take $Z_t = f(E_t^1,E_t^2)$. We know that, since $f$ is a bijection,
 the possible states of $Z_t$ are simply labels for the possible states of 
 $(e_t^1,e_t^2)$. As such $\Omega(Z_t) = \Omega(\mathcal{G}^{1,2})$, and
 hence we can get the memory of the two link sub-network directly.
 What remains to be seen is how this translates to the memory of the
 temporal network as a whole.\\
 
 \noindent If we now consider each possible pair of links $E_t^i , E_t^j$, and 
 their product sequence $Z_t^{ij} = f(E_t^i , E_t^j )$, then we can obtain the following result:
 \begin{theorem}
 	For a temporal network $\mathcal{G}$ with link processes $\mathcal{E}^{\alpha}$ and where the product of pairs 
	$Z_t^{\alpha \beta} = f(E_t^{\alpha} , E_t^{\beta})$ are given by some bijection $f: \{0,1\} \times \{0,1\} \to \left[0,3\right]$,
	where the effective memory is given by $\Omega_{\text{\text{eff}}} ( \mathcal{G} ) =  \max_{\alpha,\beta} \left[ \Omega(\mathcal{E}^{\alpha} || \mathcal{E}^{\beta} )  \right]$,
	and the pair memory is given by $\Omega_{\text{pair}} (\mathcal{G}) = \max_{\alpha, \beta} \left[ \Omega(\{Z_t^{\alpha \beta}\}) \right]$,
	the following inequality holds:
	\begin{equation}
		\Omega( \mathcal{G} ) \leq \Omega_{pair} (\mathcal{G} ) \leq \Omega_{\text{\text{eff}}} ( \mathcal{G}).
	\end{equation}
 \end{theorem}
 
\proof{Let us prove the second of these inequalities first: $\Omega_{\text{pair}} (\mathcal{G}) \leq \Omega_{\text{\text{eff}}} ( \mathcal{G})$.
We know that for any pair of links $(\alpha,\beta)$, $\Omega(\{Z_t^{\alpha \beta} \} ) \leq \max \{ \Omega(\mathcal{E}^{\alpha} ), \Omega(\mathcal{E}^{\beta} ), \Omega ( \mathcal{E}^{\alpha} || \mathcal{E}^{\beta} ) ,  \Omega ( \mathcal{E}^{\beta} || \mathcal{E}^{\alpha} )  \}$. Let us assume that $\Omega (\{Z_t^{\alpha \beta}\}) $ is maximal for the links 
$(\alpha, \beta) = (\alpha' , \beta')$, then 
\begin{eqnarray}
\begin{aligned}
	\max_{\alpha, \beta} \left( \Omega(\{Z_t^{\alpha \beta}\}) \right) =& \, \Omega(\{Z_t^{\alpha' \beta'}\}) \\
	\leq&  \max \{ \Omega(\mathcal{E}^{\alpha'} ), \Omega(\mathcal{E}^{\beta'} ), \Omega \left( \mathcal{E}^{\alpha'} || \mathcal{E}^{\beta'} \right) ,  \Omega \left( \mathcal{E}^{\beta'} || \mathcal{E}^{\alpha'} \right)  \} \\
	\leq&  \max_{\alpha,\beta} \left[ \Omega(\mathcal{E}^{\alpha} || \mathcal{E}^{\beta} )  \right].
\end{aligned}
\end{eqnarray}
As required. 

\noindent Now the first of the inequalities. Assume without loss of generality that $\Omega(\mathcal{G}) = p$.
Then there must exist at least one link $\mathcal{E}^m$ with 
$p =  \min_{n} ( n : \mathbb{P} ( e_t^m | e_{t-1,t-n}^{1,L} ) = \mathbb{P} ( e_t^m | e_{t-1,-\infty}^{1,L} ) )$, i.e. there
must be at least one link $\ell$ which remembers some part of the network $p$ time steps ago.
We then see that there is at least one link $\mathcal{E}^{\ell}$ which is remembered at least $p$ time steps ago,
i.e. $ \mathbb{P} ( e_t^{m} | e_{t-1,t-p}^{1,L})$ is a function of $\mathcal{E}^{\ell}$ (and other links and time indices). 
Taking 
the pair process $\{Z_t^{m \ell}\}$ we must hence have that $\Omega(\{Z_t^{m \ell}\}) \geq p$,
since its conditional must be a function containing terms at least $p$ steps into the past. Hence there exists
some $(\alpha,\beta)$ such that $\Omega( \{Z^{\alpha \beta}_t \}) \geq \Omega(\mathcal{G})$, and hence we must have that $\Omega( \mathcal{G} ) \leq \Omega_{pair} (\mathcal{G} )$.
This concludes our proof.\qed}\\
 
\noindent Notably, we must also have that $\Omega(\{Z_t^{\alpha \beta}\}) = \Omega(\{Z_t^{\beta\alpha}\})$, and so for a network with $L$ links
 only $L ( L-1)/2$ pairs of values $(\alpha,\beta)$ must be checked to find the maximum, and hence the
 estimated pair memory of the network, meaning that
 this approach may be faster to implement than finding the co-orders directly 
 in some cases (though, possibly because of the choice of estimator used in this work,
 this is not the case here).

 \section{Why the co-memory matrix $\mathbb{M}$ and the effective memory $\Omega_{\text{\text{eff}}}(\mathcal{G})$ are better defined and more useful concepts than the scalar memory
   $\Omega(\mathcal{G})$}

\begin{figure}[htb]
  \includegraphics[width = 0.85\textwidth]{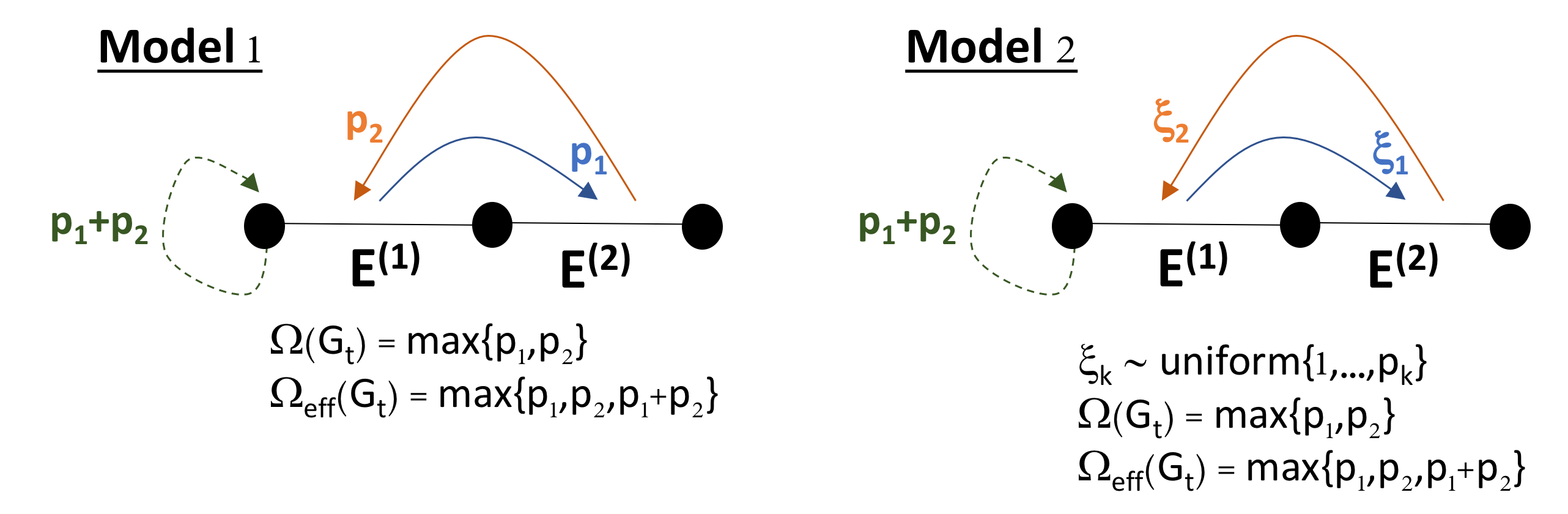}
  \caption{Models of temporal networks with two links and tunable memories. In
    both models 
  the temporal activity of each link depends on the past of the other
  link. However the memory kernel is different for the two models. In
  {\it model 1}, link 1 can copy the $t-p_1$ state of link 2, while 
  link 2 can copy the $t-p_2$ state of link 1, thereby inducing a
  virtual loop of memory $p_1+p_2$ in the dynamics of each link.
  In {\it model 2}, {the first link can copy a state of the second link,
  choosing uniformly at random from the previous $t-1,...,t-p_1$ available states}.
  Analogouly the second link can copy one of the
  $t-1,...,t-p_2$ states of the first link. This model also induces
  virtual loops of
    memory $p_1 + p_2$, although, in practice the estimation method
may not necessarily be able to identify these virtual loops due to virtual
loop decoherence (see text for details).} 
\label{fig:toymodel}
\end{figure}
\subsection{Two toy models and their virtual loops}

The main reason why it is important to properly define and
extract the memory of a temporal network is to unveil its influence
on the dynamics of processes occurring over the network. 
We will show here that the co-orders  $\mathbb{M}_{\alpha  \beta}=
\Omega (  {\mathcal E}^\alpha ||  {\mathcal E}^\beta) 
$ 
and the associated effective memory
$\Omega_{\text{\text{eff}}}(\mathcal{G})$ of a temporal network
$\mathcal{G}$ are better quantities to characterize 
spreading processes occurring over $\mathcal{G}$,
than a traditional definition such as $\Omega(\mathcal{G})$, which is not only computationally difficult
to estimate, but also suffers from fundamental drawbacks. We will illustrate
this by means of the following two temporal network models.
\\

\noindent {\bf Model 1 -- } The first temporal network we consider is
a chain graph with three nodes and two links with binary states 
${0,1}$. The temporal activity of the two links is ruled by the two coupled 
stochastic processes ${\mathcal E}^1 = \{ E_t^1 \}_{t=1,2,\ldots}$ and
$ {\mathcal E}^2 = \{ E_t^2 \}_{t=1,2,\ldots}$ defined as: 
\begin{eqnarray}
	E_t^1 = Q_t^1 E_{t-p_1}^2 + (1-Q_t^1) Y_t^1, \\
	E_t^2 = Q_t^2 E_{t-p_2}^1 + (1-Q_t^2) Y_t^2,
	\label{two_link_example_v1}
\end{eqnarray} 
with $Q_t^1,Q_t^2 \sim Bernoulli(q)$, $Y_t^1,Y_t^2 \sim Bernoulli(y)$
and $p_1$ and $p_2$ two positive integers.  This means that,
at each time step $t$, each link will either be sampled from a Bernoulli
trial, or it will copy from the past states of the other link.
In the latter case, link 1 will copy the state of link 2 at
time $t-p_1$, i.e. exactly $p_1$ steps back in the past,
while link 2 will copy the state of link 1 at
time $t-p_2$. The model is illustrated in the left
panel of Figure \ref{fig:toymodel}. 

We can now state and prove the following theorem on the memory
$\Omega(\mathcal{G})$ of model 1:
\\
 \begin{theorem}
If a temporal network $\mathcal{G}$ is determined by model 1 above,
then $\Omega(\mathcal{G}) = \max \{p_1,p_2\}$
 \end{theorem}
\proof{
 By construction we have that $\mathbb{P}(e^{1,2}_{t} | e^{1,2}_{t-1,t-p}) = \mathbb{P}(e^{1}_{t} | e^{1,2}_{t-1,t-p}) \mathbb{P}(e^{2}_{t} | e^{1,2}_{t-1,t-p})$. Note that $e^{1,2}_{t-1,t-p} = e^{1}_{t-1,t-p},e^{2}_{t-1,t-p}$.
 Then the conditional probability that link 1 is present is given by
 \begin{equation}
 	\mathbb{P}(e^{1}_{t} =1 | e^{1,2}_{t-1,t-p_1}) = q e^2_{t-p_1} + (1-q)y,
\end{equation}
 and similarly for $E^2_t$. It is then clear that for any
   $\delta > 0$
\begin{eqnarray}
\begin{aligned}
	\mathbb{P}(e^{\alpha}_{t} | e^{1,2}_{t-1,t-p_1}) = \mathbb{P}(e^{\alpha}_{t} | e^{1,2}_{t-1,t-(p_1+\delta)}),\\
	\mathbb{P}(e^{\alpha}_{t} | e^{1,2}_{t-1,t-p_1}) \neq \mathbb{P}(e^{\alpha}_{t} | e^{1,2}_{t-1,t-(p_1-\delta)}),
\end{aligned}
\end{eqnarray}	
for $\alpha=1,2$. Hence we must have that $\min_p \left[ p: \mathbb{P}(e^{1,2}_{t} | e^{1,2}_{t-1,t-p}) = \mathbb{P}(e^{1,2}_{t} | e^{1,2}_{t-1,t-\infty}) \right] = \max (p_1,p_2)$.\qed}\\

Interestingly, when we look at the different entries of the 
matrix $\mathbb{M}$, we get some {intriguing} results. In fact,
in this case the matrix is two-dimensional and,  
together with the terms
$ \Omega (  {\mathcal E}^1 ||  {\mathcal E}^2)=p_1 $ 
and
$ \Omega (  {\mathcal E}^2 ||  {\mathcal E}^1)=p_2 $ 
we must also evaluate the diagonal terms
$ \Omega (  {\mathcal E}^1 ||  {\mathcal E}^1)$  and
  $ \Omega (  {\mathcal E}^2 ||  {\mathcal E}^2)$.

\noindent Let us first consider $ \Omega (  {\mathcal E}^1 ||  {\mathcal E}^1)$.
We can treat the coupled
system in Eq.~\ref{two_link_example_v1}
by re-writing an expression for  
$E_t^1$ containing only terms related to link 1 directly.
We obtain: 
\begin{equation}
	E_t^1 = Q_t^1 (Q_t^2 E_{t-(p_1 + p_2)}^1 + (1-Q_t^2) Y_t^2) + (1-Q_t^1) Y_t^1.
\end{equation}
which clearly shows that 
$ \Omega (  {\mathcal E}^1 ||  {\mathcal E}^1) = p_1 + p_2$. Similarly we can prove that
$\Omega (  {\mathcal E}^2 ||  {\mathcal E}^2) = p_1 + p_2$. These results have been
confirmed by numerically simulating the model
and measuring the co-orders directly.
By construction, this means that the effective memory of this system
is $\Omega_{\text{\text{eff}}}(\mathcal{G}) = p_1 + p_2$, which is
possibly up to twice the value of the scalar memory
$\Omega(\mathcal{G})$.  If, without any loss of generality, we set $p_1>p_2$
and then fix $p_1$ and let $p_2$ vary, we can thereby construct a
variety of temporal networks with exactly the same scalar memory
$\Omega(\mathcal{G})=p_1$, but with different co-memory matrices and a
 tunable effective memory
$\Omega_{\text{\text{eff}}}(\mathcal{G})=p_1+p_2$. In this case, the
 difference between $\Omega_{\text{\text{eff}}}(\mathcal{G})$
 and $\Omega(\mathcal{G})$  
 is the result of an induced {\it virtual loop} (VL) in the
dynamics of the temporal network: when link 1 draws from the memory of
link 2, it may effectively be drawing from its own, more distant,
past. {Importantly, these contributions to the memory of the
  network are intrinsically indirect: the source of memory present
  microscopically at link $E_t^1$ and $E_t^2$ is induced by the
  coupling of the link activities.} 
 Of course, these effects are
obtained as we restrict the observation state space to single links. These  
mechanisms are similar to what happens when we project a low-order
Markov chain defined on a given state space onto a smaller dimensional
state space. While these effects can be seen as an artefact of such a
projection, they turn out to have important consequences on dynamical
processes taking place on temporal networks, as we will show in
the next section. But first let us introduce a slightly more realistic
model of a network with two links, in which the virtual loops have
a different structure.\\

\noindent {\bf Model 2 --}
In order to study the effects of virtual loops in a slightly more
realistic --but still controlled-- setting, we introduce a second toy
model in which, given the same values of $p_1$ and $p_2$ from our
first model, we obtain exactly the same scalar memory
$\Omega(\mathcal{G})$, but with a different memory kernel. The link
evolution of this model is illustrated in the right panel of
Figure \ref{fig:toymodel} and 
is specified by a coupled pair of slightly
modified DAR($p$) processes:
\begin{eqnarray}
	E_t^1 = Q_t^1 E_{t-Z^1_t}^2 + (1-Q_t^1) Y_t^1, \\
	E_t^2 = Q_t^2 E_{t-Z^2_t}^1 + (1-Q_t^2) Y_t^2,
	\label{two_link_example_v2}
\end{eqnarray} 
where $Q_t^1,Q_t^2 \sim Bernoulli(q)$, $Y_t^1,Y_t^2 \sim Bernoulli(y)$, and 
$Z_t^1 \sim Uniform(1,p_1)$ and $Z_t^2 \sim Uniform(1,p_2)$.
As in model 1, at each time step, each link in model 2 will either be sampled from
a Bernoulli trial, or it will copy from the past states of the other
link. However, the first link will not exactly copy the state of link 2
at time $t-p_1$, but it will copy one of the previous $t-1,...,t-p_1$ states of the
second link selected at random with uniform probability.
The scalar memory of model 2 is clearly $\Omega(\mathcal{G}) = \max
\{p_1,p_2\}$, as this is the furthest point into the past of the
network that is required to when generating the next step of its
evolution. A formal proof of this will be given in theorem 5, as this
model is a specific case of the eCDARN($p$) model described in section
IVA.  Again, as in the previous model we find interesting patterns
when we look at the different elements of the co-memory matrix
$\mathbb{M}$. As in model 1 we have
$ \Omega (  {\mathcal E}^1 ||  {\mathcal E}^2)=p_1 $ 
and
$ \Omega (  {\mathcal E}^2 ||  {\mathcal E}^1)=p_2. $
To evaluate $ \Omega (  {\mathcal E}^1 ||  {\mathcal E}^1)$  we
rewrite $E_t^1$ using only terms related to link 1, obtaining: 
\begin{equation}
	E_t^1 = Q_t^1 (Q_t^2 E_{t-(Z^1_t + Z^2_t)}^1 + (1-Q_t^2) Y_t^2) + (1-Q_t^1) Y_t^1.
\end{equation}
From this we can then see that the dynamics of the first link can  
equivalently be described by the following process:
\begin{equation}
	E_t^1 = \bar{Q}_t E_{t-(Z^1_t + Z^2_t)}^1 + \left(1- \bar{Q}_t \right) Y_t,
\end{equation}
where $\bar{Q}_t \sim Bernoulli(q^2)$ and $Y_t \sim Bernoulli(y)$.
Hence the stochastic process $\{ E_t^1 \}_{t=1,2,\ldots}$
is a form of DAR($p$) process of order $p_1 + p_2$ \cite{jacobs1978discrete}.
Implicitly, this tells us that
$ \Omega (  {\mathcal E}^1 ||  {\mathcal E}^1) =p_1 + p_2$, 
and analogously we can get that 
$ \Omega (  {\mathcal E}^2 ||  {\mathcal E}^2) =p_1 + p_2$.
Again these results are confirmed by measuring the co-orders directly
through numerical simulations of the model.  Summing up, the effective
memory of this system is, as expected
$\Omega_{\text{\text{eff}}}(\mathcal{G}) = p_1 + p_2$.
Note that if $p_1 = p_2$ then $\Omega_{\text{eff}}(\mathcal{G}) = 2 \Omega(\mathcal{G})$.\\ 

\noindent In exactly the same way as before, and without loss of generality we
can fix $p_1>p_2$ and let $p_2$ vary to  
construct a variety of temporal networks with exactly the same 
scalar memory $\Omega(\mathcal{G})=p_1$, but with different co-memory
matrices and tunable effective memories 
$\Omega_{\text{\text{eff}}}(\mathcal{G})$.\\

\begin{figure}
\begin{center}
\begin{tikzpicture}

\Text[x=-3.7, y=2.7]{Order 1}
\Vertex[x=-4.5, y=2.0, size=.6, color=blue!20, opacity=0.9, label=E$^1$]{N1two}
\Vertex[x=-3.0, y=2.0, size=.6, color=blue!20, opacity=0.9, label=E$^2$]{N2two}
\Edge[Direct, lw=.3, color=black, bend=25](N1two)(N2two)
\Edge[Direct, lw=.3, color=black, bend=25](N2two)(N1two)
\Edge[Direct, lw=.3, color=black, bend=25, style={dashed}, loopposition=-180, label=VL, position= left](N1two)(N1two)
\Edge[Direct, lw=.3, color=black, bend=25, style={dashed},  label=VL, position= right](N2two)(N2two)

\Text[x=5.0, y=2.2]{Order 2}
\Vertex[x=4.5, y=1.5, size=.6, color=blue!20, opacity=0.9,label=E$^1$]{N1three}
\Vertex[x=3.0, y=1.5, size=.6, color=blue!20, opacity=0.9,label=E$^2$]{N2three}
\Vertex[x=3.75, y=2.5, size=.6, color=blue!20, opacity=0.9,label=E$^3$]{N3three}
\Edge[Direct,lw=.3, color=black, bend=20](N1three)(N2three)
\Edge[Direct,lw=.3, color=black, bend=20](N2three)(N3three)
\Edge[Direct,lw=.3, color=black, bend=20](N3three)(N1three)
\Edge[Direct, lw=.3, color=black, bend=25, style={dashed}, loopposition=-45, label=VL, position= right](N1three)(N1three)
\Edge[Direct, lw=.3, color=black, bend=25, style={dashed}, loopposition=-135, label=VL, position= left](N2three)(N2three)
\Edge[Direct, lw=.3, color=black, bend=25, style={dashed}, loopposition=90, label=VL, position= above](N3three)(N3three)

\Text[x=-3.7, y=0.5]{Order 3}
\Vertex[x=-4.5, y=-1.5, size=.6, color=blue!20, opacity=0.9,label=E$^2$]{N1four}
\Vertex[x=-3.0, y=-1.5, size=.6, color=blue!20, opacity=0.9,label=E$^3$]{N2four}
\Vertex[x=-4.5, y=-0.5, size=.6, color=blue!20, opacity=0.9,label=E$^1$]{N3four}
\Vertex[x=-3.0, y=-0.5, size=.6, color=blue!20, opacity=0.9,label=E$^4$]{N4four}
\Edge[Direct,lw=.3, color=black, bend=-20](N1four)(N2four)
\Edge[Direct,lw=.3, color=black, bend=-20](N2four)(N4four)
\Edge[Direct,lw=.3, color=black, bend=-20](N4four)(N3four)
\Edge[Direct,lw=.3, color=black, bend=-20](N3four)(N1four)
\Edge[Direct, lw=.3, color=black, bend=25, style={dashed}, loopposition=-135, position= below](N1four)(N1four)
\Edge[Direct, lw=.3, color=black, bend=25, style={dashed}, loopposition=-45, position= below](N2four)(N2four)
\Edge[Direct, lw=.3, color=black, bend=25, style={dashed}, loopposition=135, position= above](N3four)(N3four)
\Edge[Direct, lw=.3, color=black, bend=25, style={dashed}, loopposition=45, position= above](N4four)(N4four)

\Text[x=5.0, y=-0.25]{Order 4}
\Vertex[x=4.75, y=-2.5, size=.6, color=blue!20, opacity=0.9,label=E$^1$]{N1five}
\Vertex[x=2.75, y=-2.5, size=.6, color=blue!20, opacity=0.9,label=E$^2$]{N2five}
\Vertex[x=3.75, y=-0.5, size=.6, color=blue!20, opacity=0.9,label=E$^4$]{N4five}
\Vertex[x=5.2, y=-1.2, size=.6, color=blue!20, opacity=0.9,label=E$^5$]{N5five}
\Vertex[x=2.3, y=-1.2, size=.6, color=blue!20, opacity=0.9,label=E$^3$]{N3five}
\Edge[Direct,lw=.3, color=black, bend=20](N1five)(N2five)
\Edge[Direct,lw=.3, color=black, bend=20](N2five)(N3five)
\Edge[Direct,lw=.3, color=black, bend=20](N3five)(N4five)
\Edge[Direct,lw=.3, color=black, bend=20](N4five)(N5five)
\Edge[Direct,lw=.3, color=black, bend=20](N5five)(N1five)
\Edge[Direct, lw=.3, color=black, bend=25, style={dashed}, loopposition=-45, position= below](N1five)(N1five)
\Edge[Direct, lw=.3, color=black, bend=25, style={dashed}, loopposition=-135, position= below](N2five)(N2five)
\Edge[Direct, lw=.3, color=black, bend=25, style={dashed}, loopposition=135, position= above](N3five)(N3five)
\Edge[Direct, lw=.3, color=black, bend=25, style={dashed}, loopposition=90, position= above](N4five)(N4five)
\Edge[Direct, lw=.3, color=black, bend=25, style={dashed}, loopposition=45, position= above](N5five)(N5five)

\end{tikzpicture}
\end{center}
\caption{Examples of cyclic Bayesian networks describing the link
  dependencies from which virtual loops emerge. Virtual loops produce
  {long} memory 
  contributions in the
  diagonal entries of the co-memory
  matrix. These loops are induced by appropriately coupling the
  dynamics of the links. Shown are 
  possible configurations where VLs of order 1, 2, 3 and 4
  can emerge. The nodes in each diagram denote the stochastic process 
  ${\mathcal E}^{\alpha}$ associated to link $\alpha=1,2, \ldots,L$
  of a temporal network. Solid arrows denote temporal dependencies
  between links,  whereas dashed arrows indicated induced VLs.}
\label{HO_VL}
\end{figure}
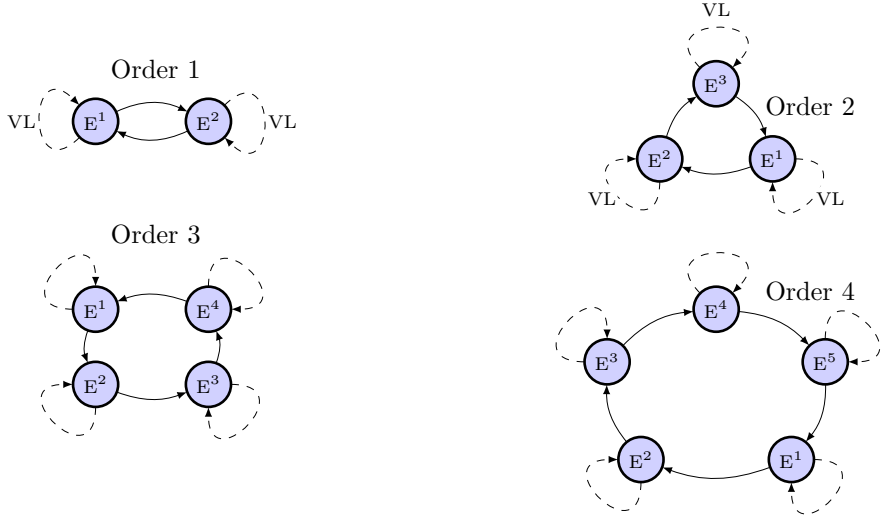

\subsection{Virtual loops of arbitrary order and conditions for VLs to emerge}

The temporal network toy models discussed above consist of a linear chain of three nodes
and two links. As illustrated in figure \ref{HO_VL}, 
a convenient way of representing these temporal networks is by 
a graphical model 
with two nodes, describing the two links (stochastic processes) 
  of the temporal networks, connected by
directed arrows that express the temporal dependence
structure between different network links.
This representation is known in the
literature as a Bayesian network (BN), and when links indeed describe causal relationships, it is often termed as a Bayesian or causal network. 
While BNs are usually directed
acyclic graphs (DAGs) by definition, it is easy to see that the
BN associated to our toy temporal network models 
is indeed cyclic (CBN), {and this is indeed a sufficient condition for virtual loop effects to emerge in memory.}
\\

\noindent Models 1 and 2 are among the simplest temporal networks to provide a nontrivial
virtual loop (VL) structure. Because VLs in this case are induced via
a causal path involving only a pair of links ($E^1\to E^2 \to E^1
\Rightarrow E^1\lcirclearrowleft$), we denote this as a VL of order
1. Now, it is easy to construct `higher order'
virtual loops e.g. simply by building {temporal networks with an underlying Bayesian ring topology} where link 1
dynamically depends on link 2, link 2 dynamically depends on link 3,
etc., and link $n$ dynamically depends on link 1. 
CBNs with VLs of orders 2, 3 and 4 are shown in figure
\ref{HO_VL} together with one of order 1. 
In theory, VLs of higher order induce contributions to
the co-memory matrix with longer 
memory, yielding a longer effective
memory $\Omega_{\text{\text{eff}}}(\mathcal{G})$. In practice, virtual
loops with high memory are difficult to observe due to extremely long
time series being required to capture the effect. More importantly,
these VLs are stable as long as the temporal dependencies between links 
are fine tuned, and quickly dissipate otherwise, meaning that the
practical relevance of VLs in the context of real-world temporal
networks is less clear.  In the next subsection, by comparing models 1
and 2 in more detail, we will give a closer look to this mechanism of
``virtual loop decoherence''.\\

\noindent Before that, we should also highlight that the underlying BN needs not to be cyclic for virtual loops to emerge: this is a sufficient, not necessary condition. Indeed, when links are also auto-correlated (meaning that on top of the link temporal dependencies, we prescribe that the link has also an internal dynamics which is auto-correlated), then virtual loops can also emerge even in the case that the underlying BN is a priori acyclic. The reason is because in that case, the interplay between the auto and cross-correlated dynamics induce virtual Bayesian links. This can be better explained with an example. Suppose that link $\alpha$ has an internal auto-correlated activity, and on top of that, also depends on the past of link $\beta$. The interplay between the two dynamics is captured in the auto-correlated nature of $\alpha$, which displays signs of memory which are a mix of the ones obtained from its own internal dynamics and from the memory displayed by $\beta$, hence the actual memory of $\alpha$ (as measured by the co-order $\Omega({\mathcal E}^\alpha||{\mathcal E}^\alpha)$) is in general larger than the memory of the internal dynamics of $\alpha$. This effect is indistinguishable from a virtual loop, hence we can confidently label it in a similar way.\\
Accordingly, the following table summarises when virtual loops emerge in the system.
\begin{table}[htp]
\begin{center}
\begin{tabular}{|c|c|c|}
\hline
{\bf Autocorrelated (internal dynamics)} & {\bf Cross-correlated (link dependency)} & {\bf Virtual loops}\\
\hline
NO & NO & NO\\
NO & YES, ACYCLIC BN&NO\\
NO& YES, CYCLIC BN& YES\\
YES&NO&NO\\
YES&YES, ACYCLIC BN&YES\\
YES&YES, CYCLIC BN&YES\\
\hline
\end{tabular}
\end{center}
\caption{Summary of the situations where virtual loops are expected to emerge.}
\label{table:summ}
\end{table}%

\subsection{Virtual loops in other areas of physics and beyond}
{The virtual loops discussed above are indeed particular cases of ``causal loops" and therefore share some similarities with important concepts arising at the heart of
a number of key challenges in several areas of modern science.
For instance, when finding the marginal distributions of a collection of
random variables it is common to use the message passing, or belief
propagation, algorithm
\cite{pearl1982reverend,yedidia2003understanding}.  Such methods are
important in the study of Gaussian graphical models in machine
learning \cite{weiss2000correctness}, signal processing
\cite{baron2009bayesian}, and a plethora of other such inference
problems
\cite{ihler2009particle,felzenszwalb2006efficient,lokhov2014inferring}.
These message passing algorithms also have uses in the context of
  statistical physics, where they can be related to the Bethe-Peierls
  approximation (or the replica symmetric cavity method in the context
  of spin glasses), and in turn the Thouless, Anderson, Palmer
  equations for local magnetisations
  \cite{yedidia2005constructing,opper2001naive,kabashima2003cdma,neirotti2005improved}.
However, this approach becomes inexact precisely when the Bayesian
graphs that underly these problems have loops. Because of this the
study of how to best overcome this problem has been seen as deeply
important
\cite{murphy1999loopy,yedidia2001generalized,ihler2005loopy}, and
  indeed has been a subject of recent attention
  \cite{cantwell2019message}.  In a different vein, causal loops have
been a subject of interest in the study of Feynman
diagrams. ``One-loop" diagrams, in which there is a single causal
loop, have historically presented challenges to study
\cite{passarino1979one,t1974one}.  However when these challenges have
been overcome they have helped to explain such phenomena as the
Casimir effect \cite{mostepanenko1997casimir,jaffe2005casimir},
Hawking radiation \cite{russo1992end}, and the Lamb shift
\cite{czarnecki2005calculation}.}

\subsection{The phenomenon of virtual loop decoherence}

The difference between the two toy models relies in the way in which
the state of the system depends from the past states.  Even though in
theory the scalar memory $\Omega(\mathcal{G})$, the effective memory
$\Omega_{\text{\text{eff}}}(\mathcal{G})$ and the co-memory matrices
are identical in the two models, in practice the estimation of these
quantities is different.  In
our first model, the {extent of the} memory is localised: each link will,
when referencing from the past, always look at the {state of the other link} a fixed
number of time steps {into the past} ($p_1$ or $p_2$).  In our second model each
link will, when referencing from the past, uniformly pick from 
{among the past $p_1$ (or $p_2$) states of the other link}.
 In other words, {links in the second model} will
seldom {copy the past state} exactly $p_1$ or $p_2$ {time steps ago}.
While theoretically the local
co-memory of each link is still $p_1+p_2$ due to the {presence of} virtual loops,
{whether we can accurately estimate} this quantity is less
obvious. In theory we would require a very large observed time
series to {estimate the theoretical effective} memory $p_1+p_2$ {consistently}.
While this is a finite size effect,
{it is however important with regards to} processes running on top
of these networks, and therefore will affect e.g. the behaviour of
{spreading} processes we run on them, as we will show in the next
section.  It is therefore in this second model that we would expect to
observe what we term ``virtual loop decoherence": each link will
not always utilise the full extent of its potential memory, {in that} the virtual
loops will not always reference a point at time $t-(p_1 + p_2)$ in
their past history. Because of this we expect the influence of the
virtual loops to be limited in comparison to our first toy model, and
hence the {theoretical} co-order $\Omega( \mathcal{E}^1_t || \mathcal{E}^1_t) = p_1 + p_2$ {will} be
difficult to detect. {The effect of this is that} the estimated
values of the diagonal terms in the co-memory matrix {will in some
  cases be} smaller than the {effective} memory of the network,
  or in more extreme cases, less than $\Omega(\mathcal{G})$. { In either of
  these cases virtual loops will not contribute} to
$\Omega_{\text{\text{eff}}}(\mathcal{G})$ (virtual loop
decoherence). As a byproduct, decoherence would {cause} the
estimation of $\Omega_{\text{\text{eff}}}(\mathcal{G})$ to
{approach} the scalar memory $\Omega(\mathcal{G})$, meaning that in
those practical scenarios where VL decoherence {does} emerge, then
the scalar memory $\Omega(\mathcal{G})$ might after all be a good
approximation to the effective memory operating underneath. In other words, conceptually the correct scalar quantity under study is  $\Omega_{\text{\text{eff}}}(\mathcal{G})$ and not $\Omega(\mathcal{G})$, but when the system is free from VLs or these decohere, then both quantities tend to be close.\\

\begin{figure}[htb]
  \includegraphics[width = 0.33\textwidth]{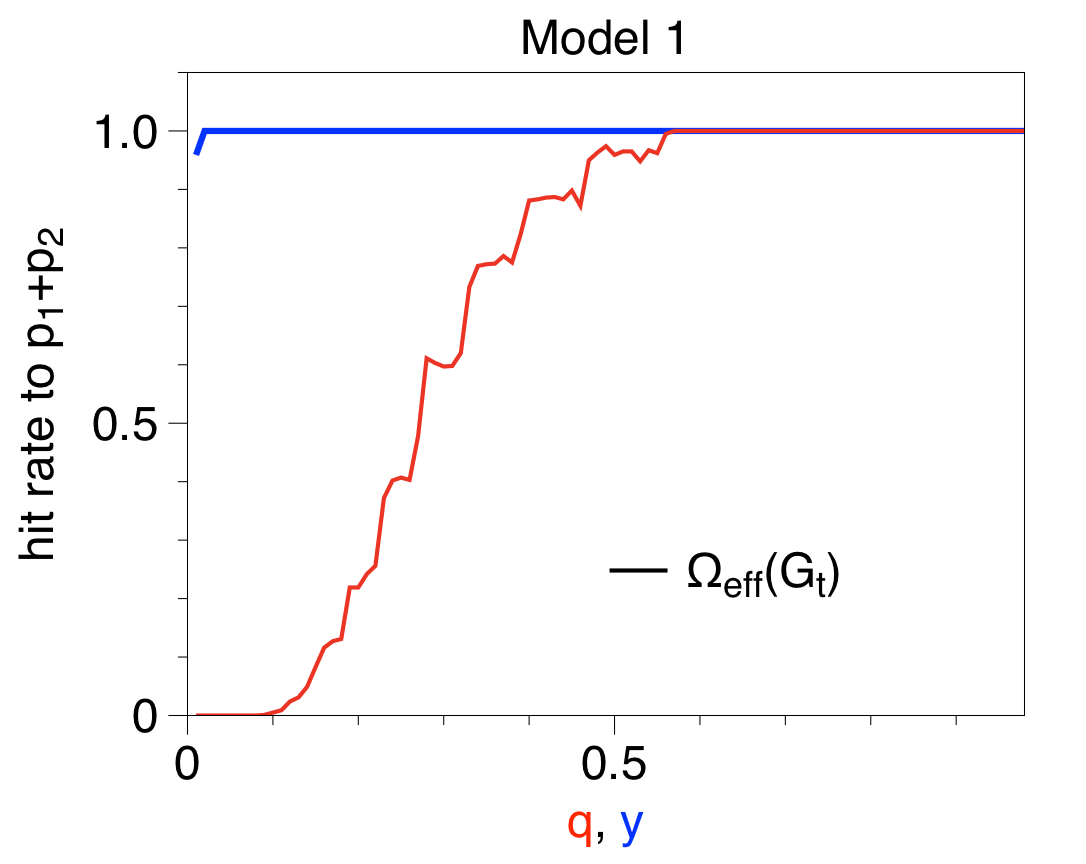}%
  \includegraphics[width = 0.33\textwidth]{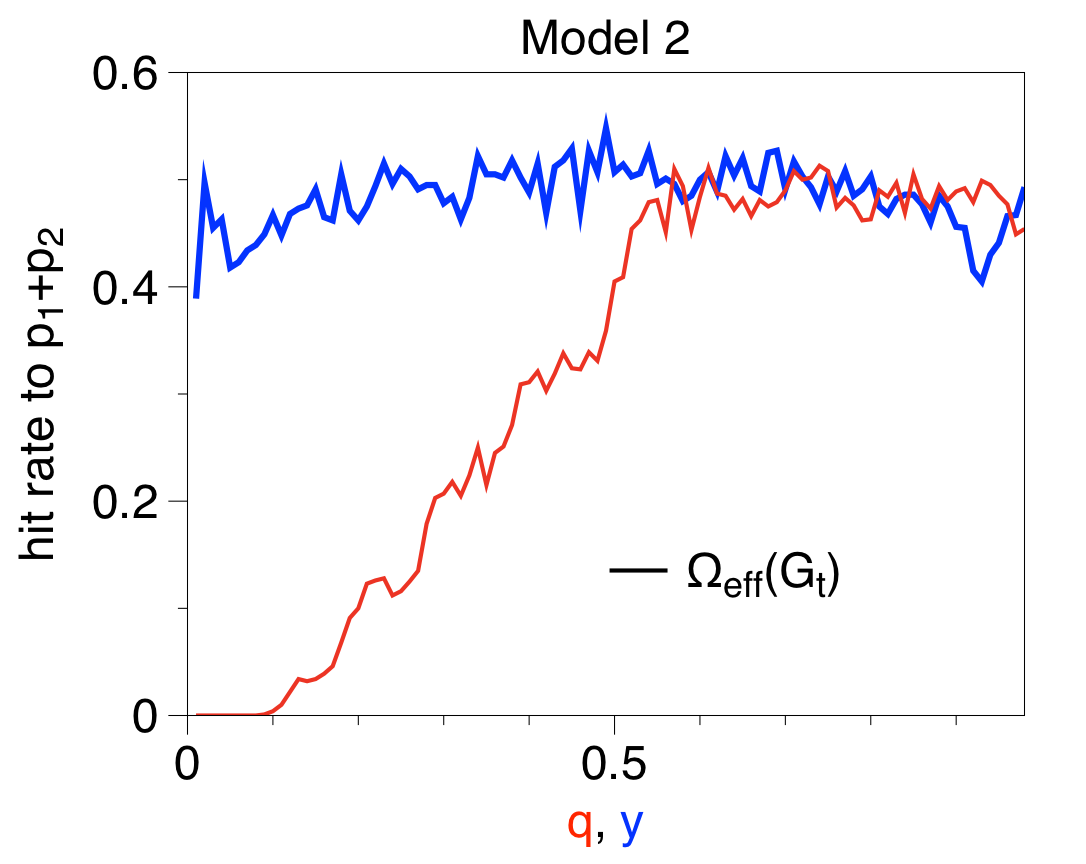}%
    \includegraphics[width = 0.33\textwidth]{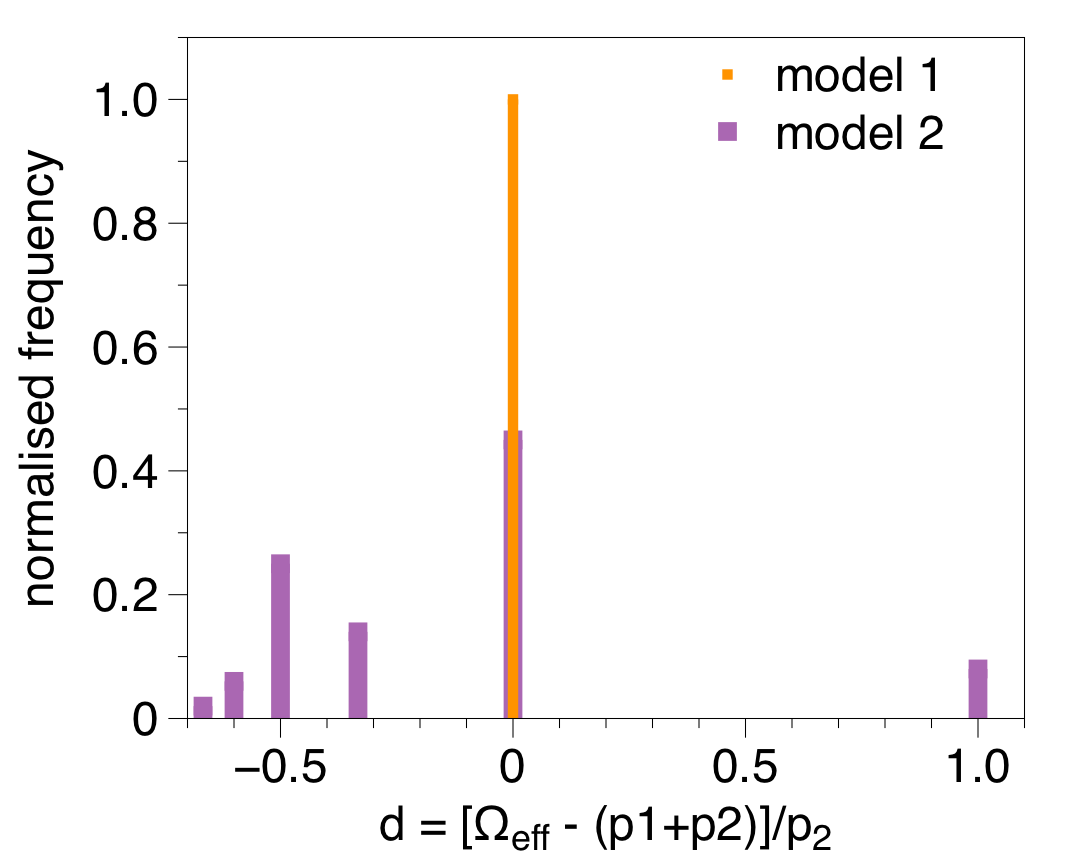}%
\caption{{\bf Hit rates and hisrograms for both toy models.} (Left and middle panels) Hit rate measuring the fraction of sampled toy model networks for which the estimation of $\Omega_{\text{eff}}(\mathcal{G})$ coincides with the local co-memory induced by the virtual loops, {as given by} 
$\Omega(E^1_t || E^1_t) = p_1 + p_2$, for model 1 (left panel) and model 2 (middle panel). For each realisation of the temporal network, we sample $p_1,p_2$ uniformly from a range $1,2,\dots,10$. Both models depend on parameters $q$ (memory length) and $y$ (link probability, {see the text for an explanation of these parameters}). {The red} curve displays the dependence of the hit rate on $q$ for a fixed $y=0.25$, whereas blue curve displays the dependence on $y$ for a fixed $q=0.9$. In every case we let {the network} evolve for $T=10^5$ time steps, and every point is the average of $10^4$ realisations (each realisation having a different $p_1,p_2$). In model 1 the virtual loops govern $\Omega_{\text{eff}}(\mathcal{G})$, whereas in model 2 virtual loop decoherence sets in and the estimated $\Omega_{\text{eff}}(\mathcal{G})$ is {no longer} systematically governed by the virtual loops. (Right panel) Dispersion histogram where we count the frequency of each estimated $\Omega_{\text{eff}}(\mathcal{G})$ (averaged over $10^4$ realisations {of the temporal network}). We compare this quantity to $p_1+p_2$ by measuring the normalised dispersion $d=[\Omega_{\text{eff}}(\mathcal{G}) - (p_1+p_2)]/2$ for models 1 and 2, with fixed parameters $q=0.9$ and $y=0.25$. In model 1 the dispersion is systematically zero ({a result of a} 100\% hit rate), i.e. virtual loops systematically govern the estimation of the effective memory. In model 2 virtual loops are decoherent and the {estimated value} of the effective memory tends to be smaller {than $p_1 + p_2$} accordingly. {We do not however reach} the asymptotic case $d=-1$ that {that we associate with} the effective memory matching the scalar memory, and {hence} virtual loops not contributing.}
\label{fig:toy_hit_rates}
\end{figure}

\noindent To analyse this we generate $10^4$ instances of both models with randomly selected values of 
$p_1$ and $p_2$ and a range of values for the memory strength $q$ and the link probability $y$, along with 
the number of time steps each network is generated for. For each instance we estimate the effective memory
$\Omega_{\text{eff}}(\mathcal{G})$, and record a ``hit" if this estimate is precisely $p_1 + p_2$. 
We plot in Fig.~\ref{fig:toy_hit_rates}  this hit rate, as a function of the parameters $q$ and $y$.
As expected, we see that in all but the few cases where the memory in each model is removed 
($q\approx0$), that the hit rate for our first model is markedly higher than for the second. In particular, the hit rate of model 1 is consistently 100\% for a large range of the model parameters. On the other hand, the hit rate for model 2 is typically smaller, hovering around 50\% for the same parameter range, suggesting that the virtual loops which are present only cause the estimated co-memory to be $p_1+p_2$ in at most half of the sampled cases, while for the rest these loops are decoherent. This effectively causes the estimation of $\Omega_{\text{eff}}(\mathcal{G})$ to approach the scalar memory $\Omega(\mathcal{G})=p_1$.\\

\noindent To further complement this analysis, we have computed, for each temporal network realisation of models 1 and 2, the normalised frequency histogram {of the} dispersion $d=[\Omega_{\text{eff}}(\mathcal{G}) - (p_1+p_2)]/p_2$. For a {given set of realisations} of each temporal network model, $d$ accounts for how well the estimated effective memory {approximates} the theoretical one ($p_1+p_2$, induced by the virtual loop), normalized over $p_2$. Finding $d=0$ means that the estimation matches the theory and virtual loops govern the effective memory. For $d\ne0$, virtual loop decoherence sets in. Typically, we expect that in this {(and indeed most)} scenarios {we will observe} $d<0$, meaning that the estimated memory contribution of any virtual loops is smaller than $p_1+p_2$. When this contribution gets smaller, $d$ approaches its minimum value $d=-1$, {the case associated with} virtual loops {being} completely decoherent and the effective memory {coinciding} with the scalar memory $\Omega(\mathcal{G})=p_1$. In the right panel of Figure~\ref{fig:toy_hit_rates} we depict the histograms of $d$ for an ensemble of $10^4$ realisations of models 1 and 2, with parameters $q=0.9$, $y=0.25$ and with $p_1$ and $p_2$ samples uniformly randomly from the range $1,..,5$ (inclusive). We systematically find $d=0$ for model 1, as expected given that the hit rate is 100\% for this model. In the case of model 2, we find that the histogram is more scattered, favouring situations with $d<0$. This means that the memory contribution to the effective memory of the virtual loops is decreased, and accordingly $\Omega_{\text{eff}}(\mathcal{G})$ gets closer to $\Omega(\mathcal{G})$, although this trend is never reached as VLs never completely decohere (see however theorem \ref{theo:CDARN_VL} for a rigorous proof that full virtual loop decoherence can take place in some systems when the size of the temporal network is large enough). Interestingly, we also find that in a small percentage of the ensemble we find $d=1$. This apparent paradox can be explained {by exploring} the EDC curves for these cases (see Fig.\ref{fig:ACF_EDC}). In every case where we find $d=1$, the assignments happen to be $p_1=2, p_2=1$, and while the estimator selects $\Omega(E^1_t||E^1_t)=4$ instead of 3, notice that the EDC curve is essentially flat at that neighbourhood.\\

\begin{figure}[htb]
  \includegraphics[width = 0.3\textwidth]{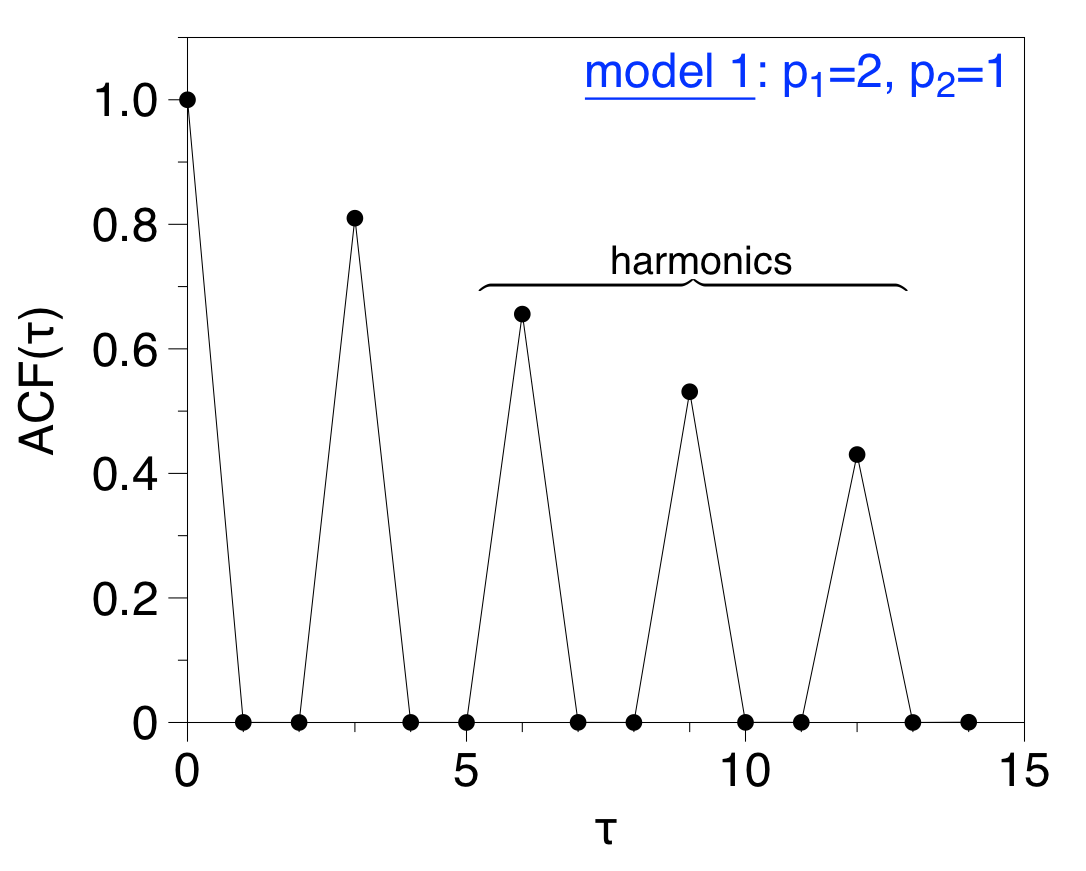}
       \includegraphics[width = 0.3\textwidth]{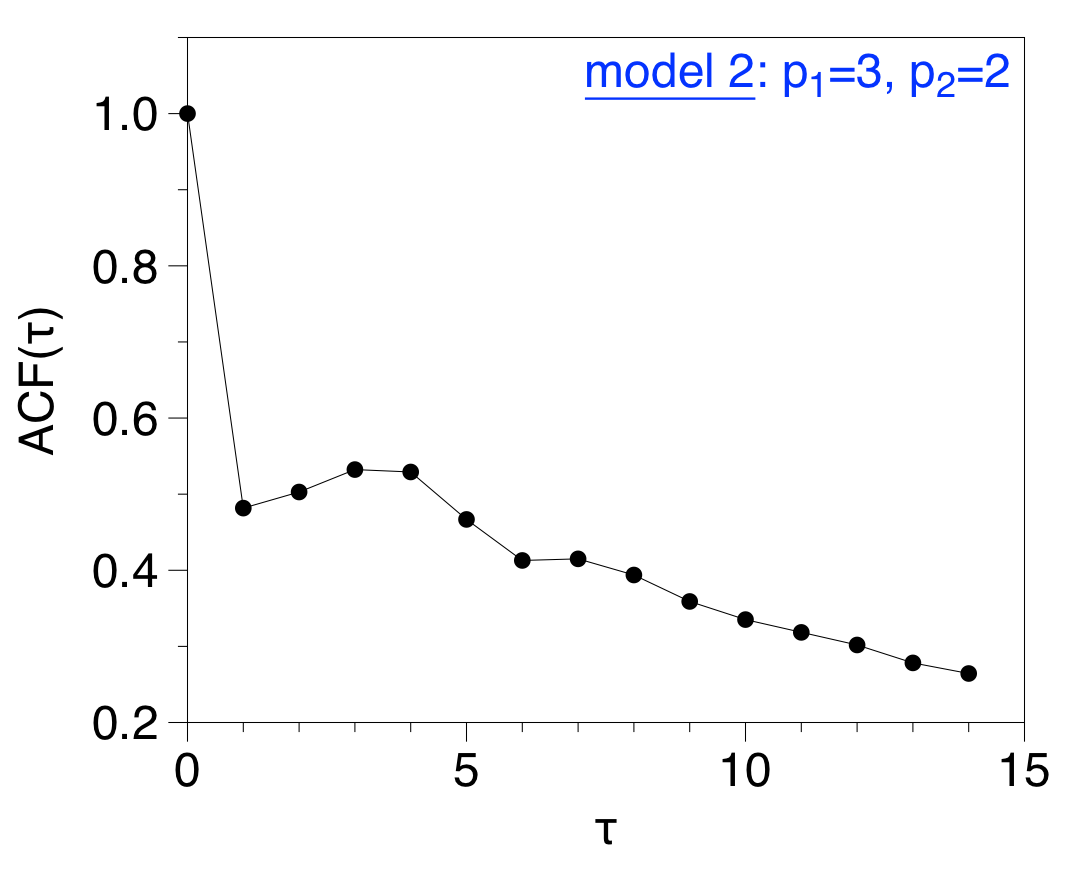}
               \includegraphics[width = 0.3\textwidth]{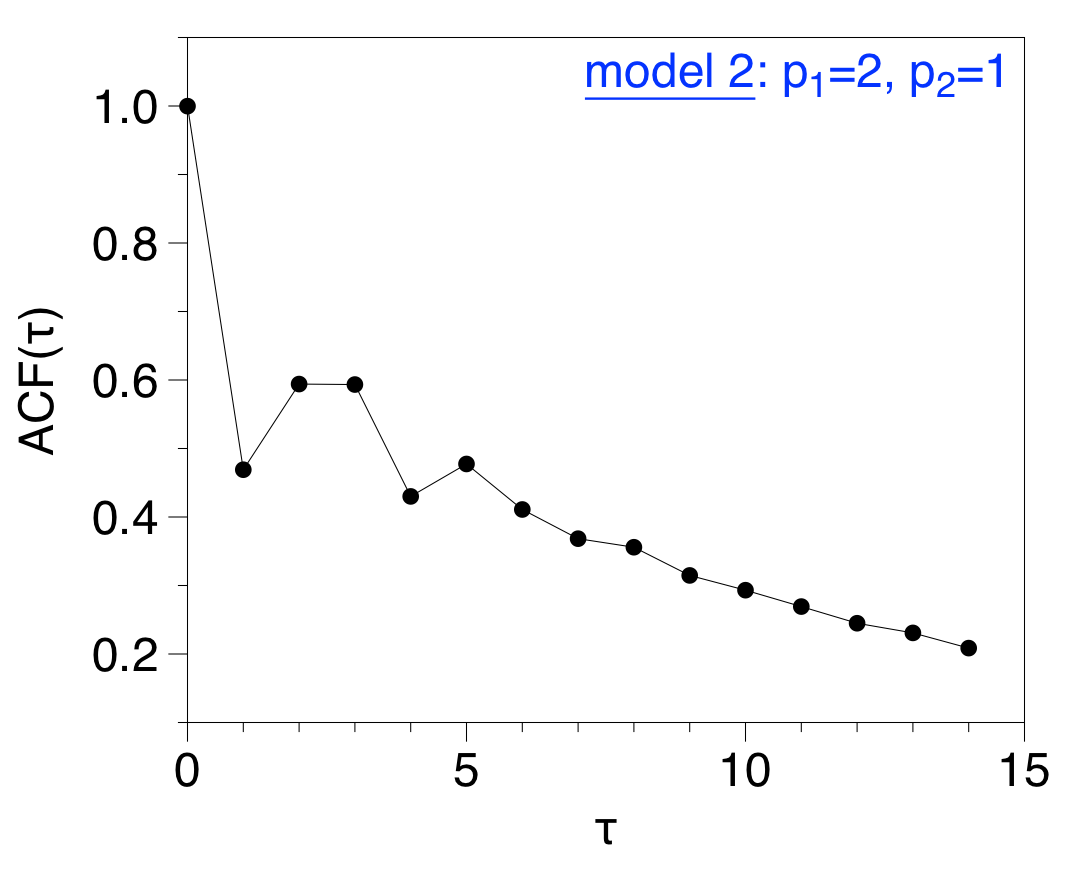}
   \includegraphics[width = 0.3\textwidth]{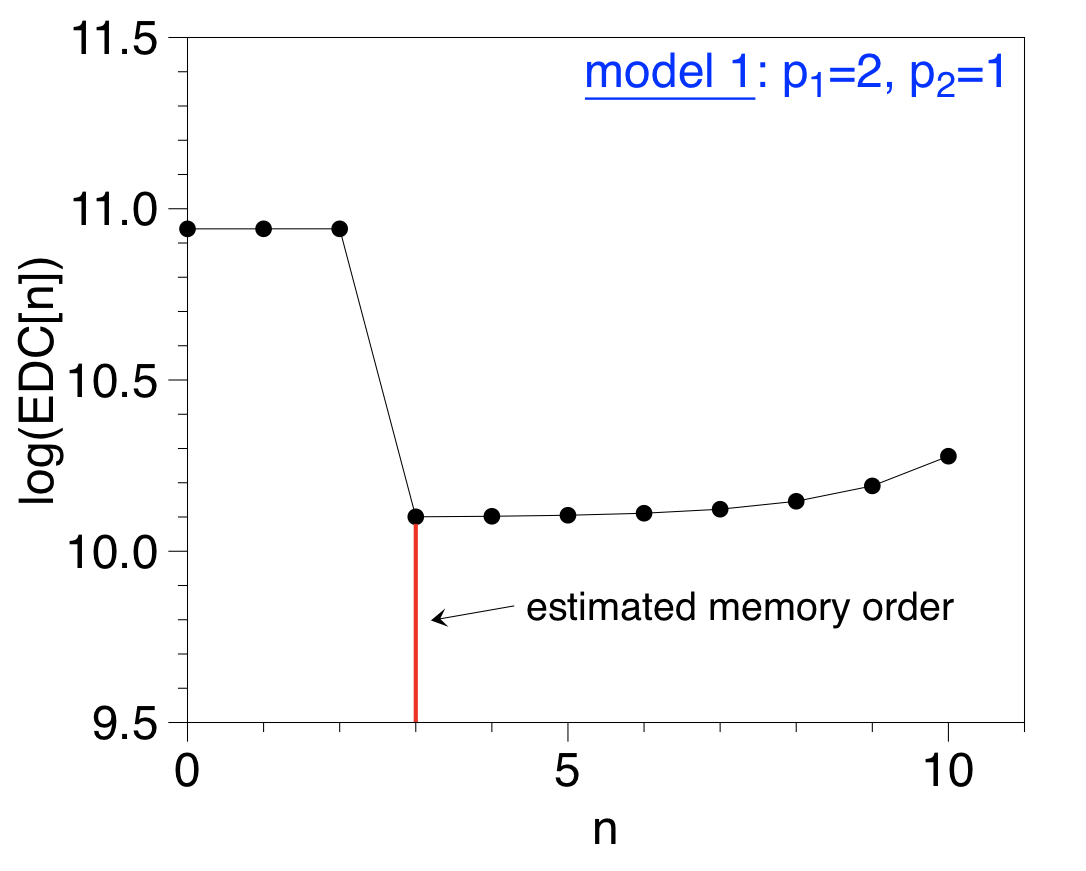}
   \includegraphics[width = 0.3\textwidth]{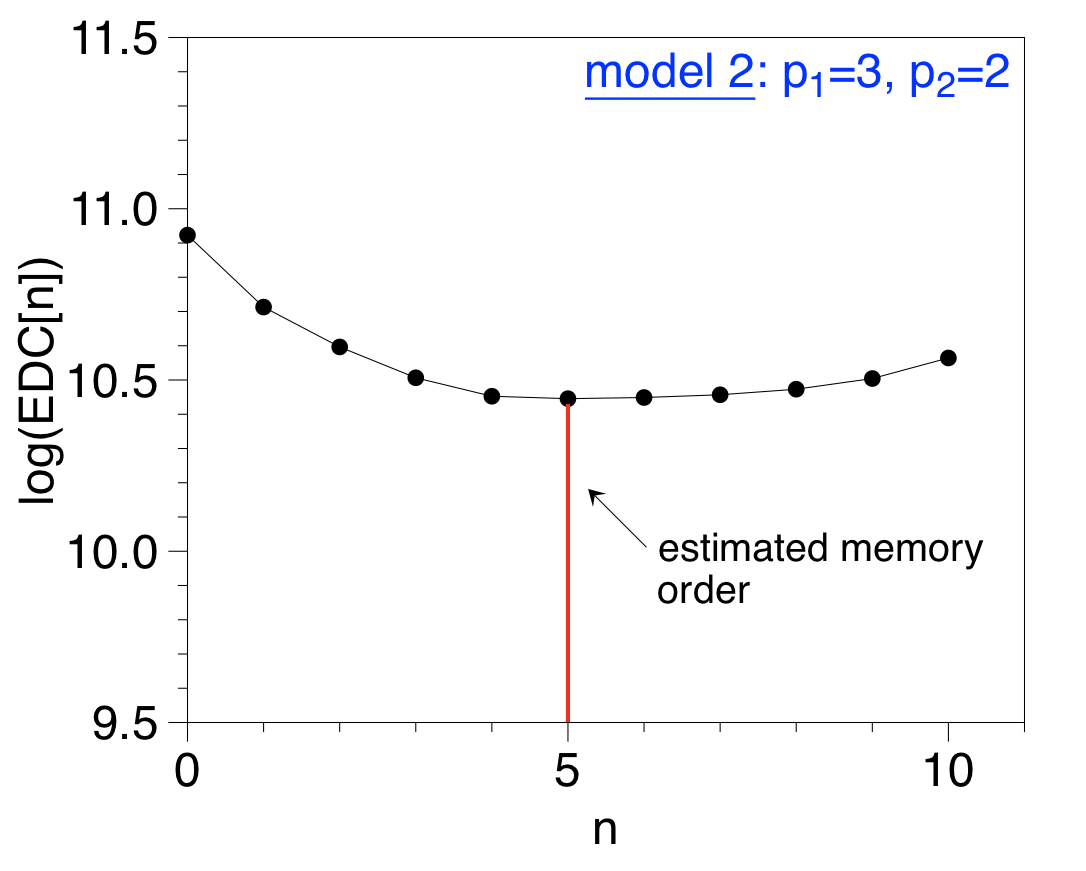}
   \includegraphics[width = 0.3\textwidth]{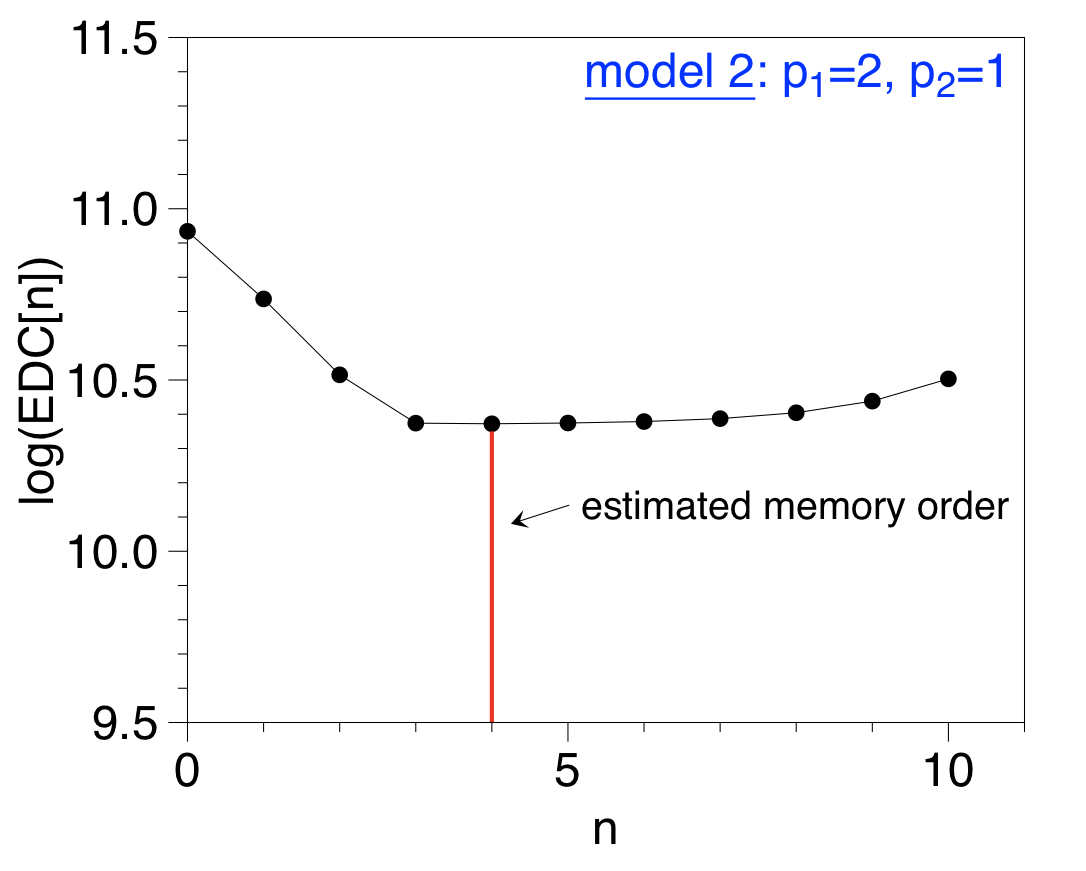}
\caption{{\bf Autocorrelation function and Efficient Determination Criterion curves for models 1 and 2.} Each column depicts the autocorrelation function ACF($\tau$) (top) and the efficient determination criterion curve EDC($n$) (bottom) of $\Omega(E^1_t || E^1_t)$, for one realisation of model 1 (left column) and two realisations of model 2 (middle and right columns). In every case the temporal networks have $10^5$ time steps. EDC estimates that the memory corresponds to the minimum of $\log(\text{EDC}(n))$. For model 1, the ACF clearly shows a peak at the correct memory $p_1+p_2$ and a succession of additional harmonics with smaller amplitude, and the EDC curve clearly captures the correct memory order. For model 2 the situation is less obvious: the ACF cannot be so easily interpreted. The two examples depict one where the theoretical memory $p_1+p_2=5$ is identified, and another where the estimation is different (note in this case that the EDC curve is relatively flat around $p_1+p_2$). }
\label{fig:ACF_EDC}
\end{figure}

\subsection{Epidemic Spreading defined on top of models 1 and 2}

Models 1 and 2 above provide examples where the scalar memory
$\Omega(\mathcal{G})$ of a temporal network is different from the
effective memory $\Omega_{\text{\text{eff}}}(\mathcal{G})$, due to the
emergence of virtual loops that affect the local memory structure of
the network. Now the question is, to what extent do the virtual loops
have a truly measurable effect and are therefore relevant in practice?
Here, we consider a spreading process over a temporal network.
  {We show that the dynamics of this spreading process are indeed highly sensitive
  to the shape of memory, and in the event a representative scalar quantity had to be used, we demonstrate accordingly that
  $\Omega_{\text{\text{eff}}}(\mathcal{G})$ is better suited
  to quantifying the real effects of memory than
  $\Omega(\mathcal{G})$, this last quantity being blind to any virtual
  loop contribution.}
  \\
\begin{figure}[htb]
  \includegraphics[width = 0.45\textwidth]{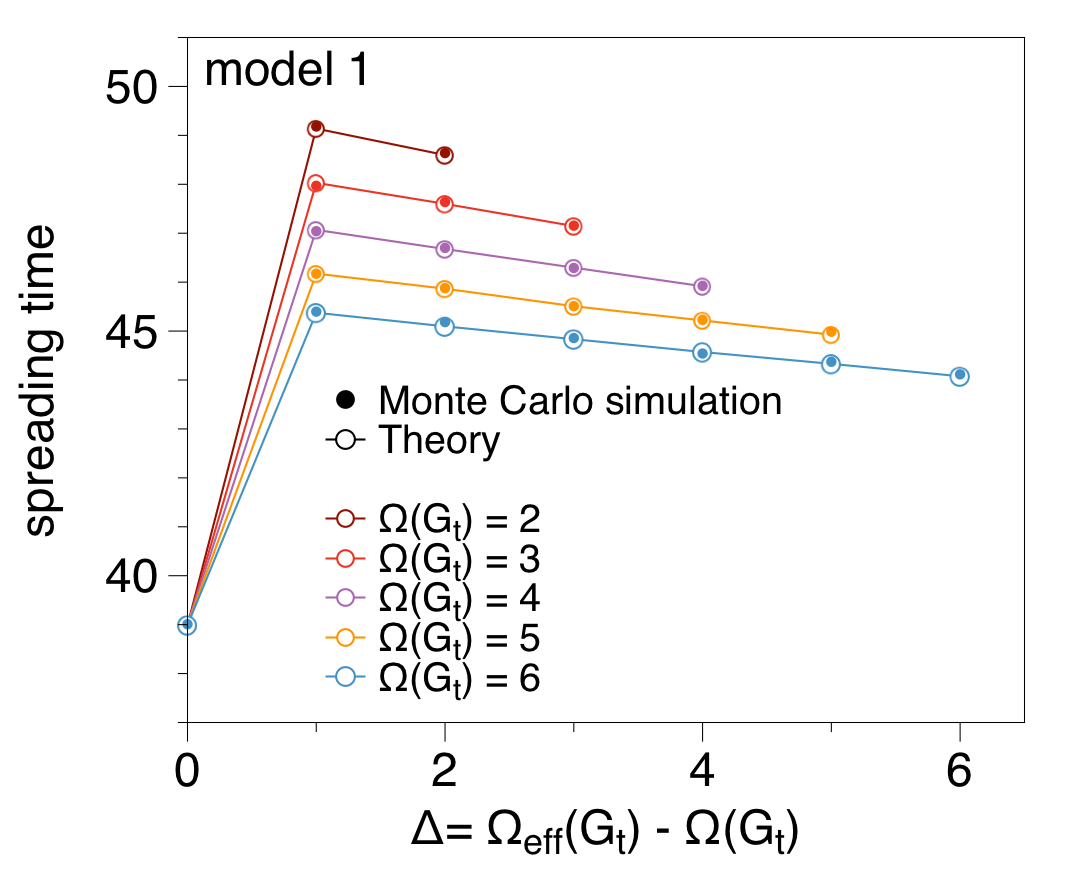}
    \includegraphics[width = 0.45\textwidth]{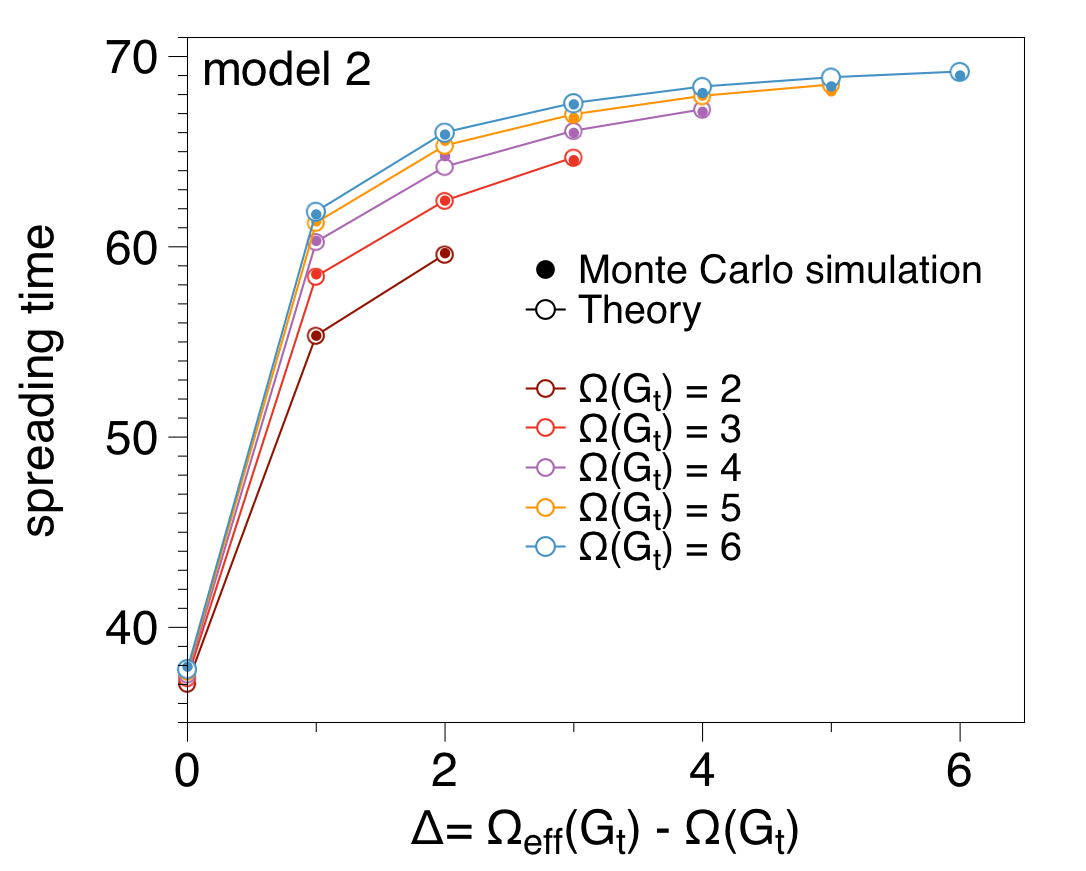}
\caption{Spreading times {for SI dynamics} on temporal
  networks generated by {\it
    model 1} (left) and {\it model 2} (right). Each curve represents
    networks with the same value 
    of the scalar memory $\Omega(\mathcal{G})$. For each of these
  networks the effective memory
  $\Omega_{\text{\text{eff}}}(\mathcal{G})$ derived from the co-memory
  matrix is different due to the emergence of virtual loops.
   Spreading times are shown to vary and this
   variation is well captured by changes in 
   $\Delta = \Omega_{\text{\text{eff}}}(\mathcal{G}) - \Omega(\mathcal{G})$,
   i.e. in the co-memory matrix. 
   Solid symbols are the results of Monte Carlo
  simulations, whereas hollow symbols are the theoretical predictions
  obtained by solving Eq. \ref{guacamoly}.}
\label{fig:toy}
\end{figure}

\noindent {We have implemented a Susceptible-Infected (SI) model for spreading dynamics on our
  two network toy models. In the SI model} a node can be in one of
two states: Infected (I) and Susceptible (S). At each
time step an infected node has a probability $\lambda$ of passing an
infection to any other node that it is connected to via a link.
Once a node is infected, it cannot become susceptible again,
the change is permanent. In our set-up we start the infection at node 1. 
The infection will then be transmitted over link 1 at time $t$ with a probability
$\lambda$ if $E_t^1 = 1$, while it will not be transmitted 
if $E_t^1 = 0$. Hence node 1 can infect node 2, then
from node 2 the infection can cross link 2 and finally infect the third node. 
To quantify the speed of the {spreading process} we will measure the
expected time taken to infect node 3, starting from node 1, and call
this the {\it spreading time}. This quantity can be evaluated either
via Monte Carlo simulations (averaging over several realisations of the
process), and also analytically.
As we will discuss below, both numerical simulations and analytical
  results reveal that the expected spreading time does indeed depend
on the virtual loops, i.e. on the precise structure of the co-memory
matrix. Furthermore, we will show that these effects {are} well
  accounted for by the effective memory
$\Omega_{\text{\text{eff}}}(\mathcal{G})$, while conversely they are not 
captured by the network memory $\Omega(\mathcal{G})$.

\subsection{Analytical {solutions of the SI dynamics  
 on }network models 1 and 2}

We here present the analytical derivation of the 
expected spreading time for an SI infection in our toy models.  We
will explicitly consider {\it model 2}. However, the same approach,
with small differences, which will be noted below, also works for model 1.

\noindent The stochastic processes ${\mathcal E}^1 = \{ E_t^1 \}_{t=1,2,\ldots}$ and
  $ {\mathcal E}^2 = \{ E_t^2 \}_{t=1,2,\ldots}$ are higher-order Markov chains 
  in the state-space $\{0,1\}$. 
They can then be transformed into 
{first order} Markov chains in an expanded state space.
The temporal network
 { ${\mathcal G} =  \{ {\mathcal E}^1, {\mathcal E}^2 \} = 
  \{ E_t^{1}, E_t^{2}\}_{t=1,2,\ldots}   $ } 
has realisation $(e^1_t, e^2_t)$
at time $t$. Since each link contains memory of the other, let us build 
this into a pair of ``state variables" $\alpha^1$ and $\alpha^2$, such that at a time $t$
$\alpha^1 = \{ e_t^1, e_{t-1}^1, ..., e_{t-(p_2-1)}^1 \}$ and similarly for $\alpha^2$. 
Any pair $(\alpha^1, \alpha^2)$ then captures all of the useful past states of the network.
Let the set of all such realisations
be denoted by $\mathcal{S}$, and the sets of possible values for $\alpha^1$ and $\alpha^2$
be $\mathcal{S}^1$ and $\mathcal{S}^2$ respectively. Link 1 has memory of the last $p_1$ 
steps of link 2, and link 2 has memory of the last $p_2$ steps of link 1, hence, since the link has 
two possible states at any one time, $|\mathcal{S}^1| = 2^{p_2}$ and   $|\mathcal{S}^2| = 2^{p_1}$.
For this to be useful we must additionally introduce some concept of ordering to the 
values of $\alpha^1$ and $\alpha^2$ by means of a labelling function. The simplest form
of this function, which we will here use, is 
\begin{equation}	
	l(\alpha^1) = \sum_{k=0}^{p_2} 2^k  \alpha^1_k,
\end{equation}
and similarly for $l(\alpha^2)$. 
 This is essentially taking the set of 0's and 1's
that represent the link histories contained in $\alpha$ and converting them to a
decimal number as if they were in binary.
We will implicitly assume that wherever we use $\alpha$,
or any state in $\mathcal{S}^1$ or $\mathcal{S}^2$, we are referring to the label $l(\alpha)$.\\

\noindent We can use this to describe the probabilistic evolution of each link over time. 
For initial states $(\alpha^1 , \alpha^2 )\in \mathcal{S} $ and target state $\beta^1$, the probability $\mathbb{P}(\alpha^1 \to \beta^1 | \alpha^1, \alpha^2)$ 
that link 1 goes from state $\alpha^1$ to state $\beta^1$
given $\alpha^1$ and $\beta^1$ defines a ``transition tensor" 
(as opposed to the traditional transition matrix) in the following way:
\begin{eqnarray}
\begin{aligned}
	T_{\alpha^1 \alpha^2}^{\beta^1} =& \left( q \frac{h(\alpha^2)}{p_1} + (1-q)y \right) \delta \left(\beta^1, 2^{p_1 - 1} + \floor{\frac{\alpha^1}{2}}\right) 
	+  \left(1- q \frac{h(\alpha^2)}{p_1} - (1-q)y \right) \delta \left(\beta^1, \floor{\frac{\alpha^1}{2}}\right),
\end{aligned}
\end{eqnarray}
where $h$ is the Hamming weight function, which counts the number of 1's in the
binary representation of its argument.
\\
Note that in the case of {\it model 1} this would be instead
\begin{eqnarray}
\begin{aligned}
	T_{\alpha^1 \alpha^2}^{\beta^1} =& \left( q s_{p_1}(\alpha^2) + (1-q)y \right) \delta \left(\beta^1, 2^{p_1 - 1} + \floor{\frac{\alpha^1}{2}}\right) 
	+  \left(1- q s_{p_2}(\alpha^2) - (1-q)y \right) \delta \left(\beta^1, \floor{\frac{\alpha^1}{2}}\right),
\end{aligned}
\end{eqnarray}
where $s_{p}(\alpha)$ is the value of the $p_{th}$ most significant entry in the binary representation
of $\alpha$, which is the state of the link at $p$ time steps in the past.
It is now a simple task to incorporate the spreading an of infection across a link:
we simply associate to each link another state $\iota$ which governs the infection.
$\iota^1 = 1$ if link 1 has passed an infection, and 0 if it has not. In this 
way the two pairs $(\alpha^1,\iota^1),(\alpha^2,\iota^2)$ completely describe the 
state of both the links and the infection passage over the system. 
What remains is to find the probabilities of an infection passing over each link given
the states of the links. That is $P(\iota^1 \to \bar{\iota}^1 | \beta^1, \iota^1) = {^1}\Lambda^{\beta^1,\bar{\iota}^1}_{\iota^1} $
and $\mathbb{P}(\iota^2 \to \bar{\iota}^2 | \beta^2, \iota^1) = {^2}\Lambda^{\beta^2,\bar{\iota}^2}_{\iota^1}$. 
Denoting by $H$ the set of link states where a link is present, i.e. $H := \{\alpha^i : l(\alpha^i) \geq 2^{p_j - 1} | j \neq i \}$,
and $L_{\bar{\iota}^i} = \lambda \delta(\bar{\iota}^i,1) + (1-\lambda)\delta(\bar{\iota}^i , 0)$, we can then write:
\begin{eqnarray}
\begin{aligned}
	{^1}\Lambda^{\beta^1,\bar{\iota}^1}_{\iota^1} = & \,
	\delta(\iota^1,1) \delta(\bar{\iota}^1,1) + \delta(\iota^1,0) \left( \chi_{H} (\beta^1) L_{\bar{\iota}^1} 
	+ (1 - \chi_{H}(\beta^1)) \delta(\bar{\iota}^1, 0) \right), \\	
	  {^2}\Lambda^{\beta^2,\bar{\iota}^2}_{\iota^1} =& \,
	  \chi_H(\beta^2) \left(\delta(\iota^1,1) L_{\bar{\iota^2}} + \delta(\iota^1,0) \delta(\bar{\iota}^2,0) \right)
	  + (1- \chi_H(\beta^2)) \delta(\bar{\iota}^2,0),
\end{aligned}
\end{eqnarray}
where $\chi_H(\alpha)$ is the indicator function for $\alpha$ in set $H$.
Using the transition tensor method outlined in \cite{Williams_2019} we can then find the expected spreading 
times $\tau_{\alpha^1,\alpha^2,\iota^1,\iota^2}$ for an infection 
given a starting state $\alpha^1,\alpha^2,\iota^1,\iota^2$ as the minimal solution to
the following set of linear equations:
\begin{equation}
	\tau_{\alpha^1,\alpha^2,\iota^1,\iota^2} = 1 + \sum_{\beta^1} \sum_{\beta^2} \sum_{\bar{\iota}^1} \sum_{\bar{\iota}^2} 
	T_{\alpha^1 \alpha^2}^{\beta^1} T_{\alpha^1 \alpha^2}^{\beta^2}  {^1}\Lambda^{\beta^1,\bar{\iota}^1}_{\iota^1}  {^2}\Lambda^{\beta^2,\bar{\iota}^2}_{\iota^1} \tau_{\beta^1,\beta^2,\bar{\iota}^1,\bar{\iota}^2}.
	\label{guacamoly}
\end{equation}
{This can then be averaged over some set of initial conditions
to give the expected spreading time.
In our case we will take the initial conditions to be the steady state of the network.
In practice this can be viewed as the state of the network after a large number of steps.
For numerical simulations we always allow the network to evolve to equilibrium before any 
spreading process is started, whereas for analytical calculations we take the steady state to be
the left eigenvalue of the transition matrix for the system corresponding to eigenvalue 1.} 
\\

\noindent To summarise, we fix the scalar memory {$\Omega(\mathcal{G})$ of the temporal network
in our toy models to be $p$ by fixing $p_1 = p$ (implicitly then $p_1 \geq p_2$).}
We then allow the {shape of the memory in the network} to
vary by changing the value of $p_2$, and measure the expected time
taken for an infection to spread over the three nodes in the
network as a function of $p_2$.
This has been done for a number of values of $p$, for both model 1 and model 2. 
The results are reported in Fig.~\ref{fig:toy}
where we plot the spreading times of the SI epidemics as a function of 
  $\Delta = \Omega_{\text{\text{eff}}}(\mathcal{G}) - \Omega(\mathcal{G})$.
  {This value $\Delta$ is used so that curves are aligned on the x-axis for any value of $p_1$.}
  Analytical results are in excellent agreement with Monte Carlo
simulations
and show that the quantity $\Omega_{\text{\text{eff}}}(\mathcal{G})$
  is able to describe well
differences in the relevant quantities that describe  
the dynamics of the SI process.
Conversely
$\Omega(\mathcal{G})$ is not able to account for the different values of
spreading times obtained in networks with the same $p_1$ and different $p_2$. 
Indeed, {in the event $\Omega(\mathcal{G})$ was well suited}, then spreading rates should remain constant, as $\Omega(\mathcal{G})$ is actually constant for all temporal networks corresponding to
each curve.  
Results indicate that the spreading time is actually not
constant, and this variation correlates with the effective memory
$\Omega_{\text{\text{eff}}}(\mathcal{G})$. This demonstrates that the
{shape of the memory, as defined by the}
co-order matrix $\mathbb{M}$ --and hence the effective memory-- are
far better at characterising the spreading rate than the scalar memory
of the temporal network.

\subsection{Inter-event time statistics}

With our toy models we have demonstrated that the scalar memory of a
temporal network $\Omega(\mathcal{G})$ is not necessarily the right
quantity to characterise the way that memory influences the spread of
an infection over the network.  In order to investigate further this
situation we can take a different, complementary approach. It is well
known that a memoryless (Poisson) stochastic process has an
exponential inter-event time distribution. Processes with
memory must have inter-event times which deviate from an exponential
distribution. A so-called burstiness parameter has been proposed in
  \cite{goh2008burstiness} to quantify such deviations. 
The burstiness parameter $B$ of a the time series is defined as:  
\begin{equation}
	B = \frac{\sigma - \left<\tau \right>}{\sigma + \left<\tau \right>},
\end{equation}
where $\left< \tau \right>$ and $\sigma$ are respectively mean and
standard deviation of the inter-event times.  The expression above is
equal to zero when the time series corresponds to a (memoryless)
Poisson process, since for an exponential distribution the mean and
standard deviation coincide.  When the time series is regular,
$\sigma=0$ and thus $B=-1$, so values of $B$ in the range $(0,-1)$
denote more regular behavior than Poisson.
On the other end, for $B>0$ the fluctuations in the inter-event times 
are larger than Poisson, denoting an increase of burstiness.
Finally, in the limit of large $\sigma$, e.g. when the interevent time series
is power law distributed, the value of $B$ tends to 1.\\

\noindent Given our previous example {of a spreading process},
it is not clear that burstiness of the link evolution processes should a-priori be 
associated with the scalar memory $\Omega(\mathcal{G})$ of the temporal network. 
Indeed we show now that this is not the case.
{If the behaviour of a link in a network is bursty, then that link displays memory, and
if it is not bursty then it does not have memory. Hence, we argue that
for a measure of the memory of a network to accurately capture the behaviour of inter event 
time statistics it must reflect this: if links are not bursty then the network should not have memory,
and if they are then the network should have memory.
We explore this in the context of our two toy models.}\\

\noindent We consider model 1 {with two parameter settings, namely:} 
(i) $p_1=1$ and $p_2=0$ and (ii) $p_1=p_2=1$. 
Notice that for both choices of parameters, models 1 and 2 are equivalent.\\
In case (i) it is easy to see that {the scalar}
memory and effective memory coincide,
$\Omega(\mathcal{G})=\Omega_{\text{eff}}(\mathcal{G})=1$, {and so 
we expect to see little or no burstiness,} whereas in 
case (ii) the temporal network is still first-order Markov, i.e. 
$\Omega(\mathcal{G})=1$, but the effective memory is larger,  
$\Omega_{\text{eff}}(\mathcal{G})=2$. In both cases, we have measured
the burstiness parameter $B$ as a function of the memory strength $q$
of a link. The results reported in
Fig.~\ref{fig:toy_burst}.
\begin{figure}[htb]
  \includegraphics[width = 0.5\textwidth]{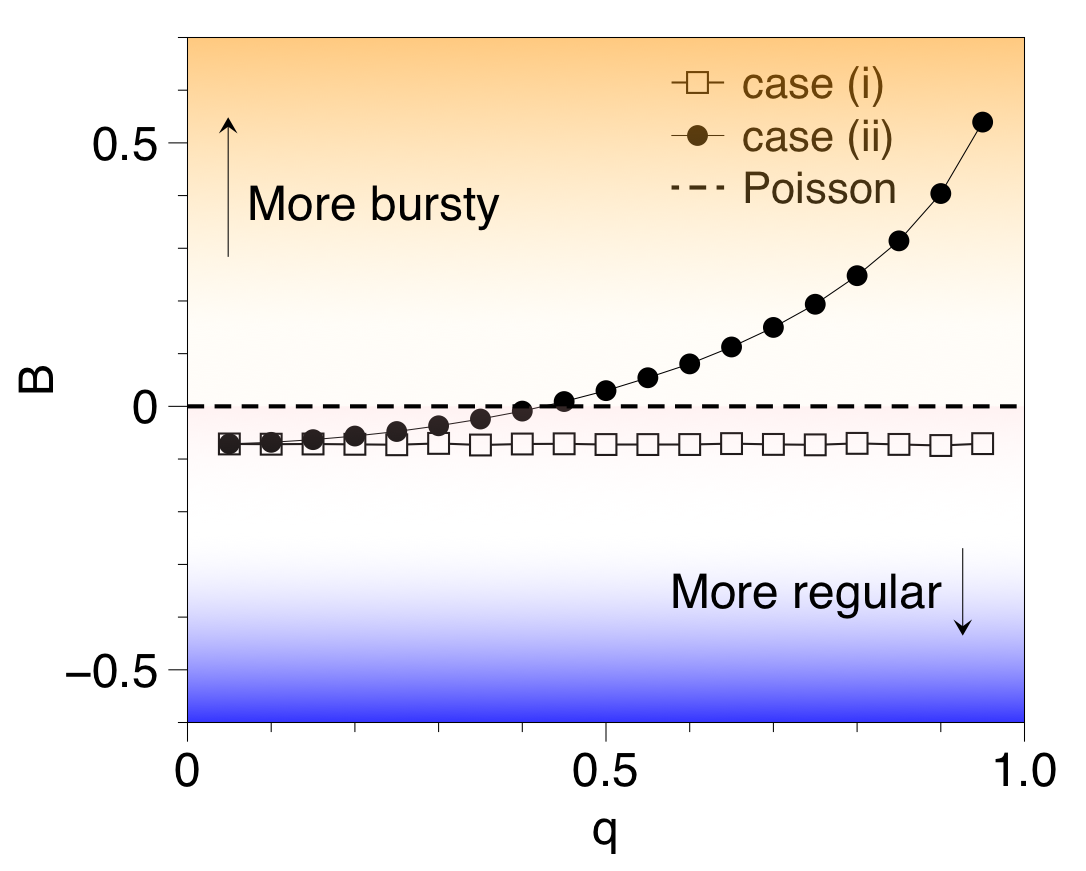}
  \caption{ Burstiness of link dynamics in model 1 as a function of the
    memory strength $q$. Two different cases have been considered 
    (i) $p_1=1, p_2=0$, for
    which $\Omega(\mathcal{G})=\Omega_{\text{eff}}(\mathcal{G})=1$),
    where $B\approx 0$ for the whole range, and (ii) $p_1=p_2=1)$,
    for which we still have $\Omega(\mathcal{G})=1$ but the effective
    memory is larger $\Omega_{\text{eff}}(\mathcal{G})=2$, where for a
    large range of values of $q$ we find $B>0$ which denotes
    burstiness in the link activity, hence highlighting the presence
    of memory in the internal activity of the links, a signature which
    is captured by $\Omega_{\text{eff}}(\mathcal{G})$ but not by
    $\Omega(\mathcal{G})$, {as the effective memory places the networks 
    in a regime where we can observe burstiness}.
    Each case is simulated for $10^8$ time
    steps.}
\label{fig:toy_burst}
\end{figure}
show that case (i) is non-bursty as expected, while case (ii) displays varying degrees
of burstiness, and therefore
  denoting the presence of a non-trivial memory shape,
even if the scalar memory is
still Markov. 
This is another indication that the memory
$\Omega(\mathcal{G})$ of the network does not 
provide a good description of dynamical dependencies and memories
 at the level of links and link pairs.
Which links are bursty, this being an indication of the presence of memory,
is far better captured by the effective memory.

\section{Validating the framework on synthetic temporal networks}

{In this section we provide more details on the models for generating
temporal networks with specific memory characteristics that we have
discussed in the main text. In particular, we have considered four
different types of generators of synthetic temporal networks: the 
so-called ``Discrete Auto-Regressive Network of order $p$" (DARN($p$)), extended DARN($p$),
``Correlated Discrete Auto-Regressive Network of order $p$" (CDARN($p$)) and extended
CDARN($p$) models. We will first introduce the models, find their  scalar memories, and then 
present a full account of their memory estimation. We will then discuss the presence of virtual loops,
and show that in some cases these loops become completely decoherent for large network sizes.}

\subsection{Model definitions and ground truth proofs for $\Omega(\mathcal{G})$ } 

  {
  All of the models here are derived from the so called ``Discrete Auto-Regressive process of order $p$", 
  or DAR($p$) process, introduced by Jacobs and Lewis \cite{jacobs1978discrete}. This time series incorporates a
  dependence on past states with random generation of new states in a simple way. In words, 
  at each step of the process one first decides if the new state will be random or drawn from memory,
  if it is random then we decide what it will be, if not then we copy a single value chosen among the past $p$
  states. Formally
  \begin{equation}
  	X_{t+1} = Q_t X_{t-Z_t} + (1-Q_t)Y_t. 
  \end{equation}
  Here $Q_t \in \{0,1\}$ and $Z_t \in \{1,...,p\}$.
  For our purposes we only consider $X_t \in \{0,1\}$ and so we must fix $Y_t \in \{0,1\}$.
  This gives us a random process with memory $p$ that can be extended to generate temporal networks.}

\begin{itemize}

\item DARN($p$):
{This model, introduced in \cite{Williams_2019}, represents the simplest possible extension of the 
DAR($p$) process to a temporal network: for a network with $N$ nodes 
we assign to each of the $L = N(N-1)/2$ possible links
an independent DAR($p$) process.
}
  For each link $\alpha=1,\dots,L$ with
  probability $q$ the state of the link is copied from its past,
  sampling uniformly at random from its {own history} up to $p$
  steps in the past.  With probability $1-q$ the link state will be
  drawn at random following a Bernoulli process with probability $y$.
  Summing up, the model depends on three parameters $\{p,q,y\}$.
  {The first parameter $y$ controls the density of the network. The second,
  $q$, tunes the strength of the memory term in the process with
  respect to the memoryless term. The final parameter, $p$, controls the
  length of the memory, which can be thought of as the number of time
  steps before the autocorrelation function decays exponentially \cite{Williams_2019}.} 
  Since each link in this model is
  independent we can make use of corollary \ref{corr:independent_mem},
  and see that we must have $\Omega(\mathcal{G}) = p$.

\item eDARN($p$): 
The extension to the DARN($p$) model, {as invented for this work}, is then the case where 
each link {is allowed to have a different memory length}, rather than there being a single fixed 
value of $p$. {Hence we specify that each link is governed by an independent DAR($p$) process
but the value of $p$ is allowed to vary for each link.} 
This gives us a way of generating temporal networks with independent links,
but with varying link memories. This model therefore depends on parameters $\{\rho(p),q,y\}$, where $\rho(p)$ is the distribution of memory lengths from which one samples the memory length of each link. If we define $\bar{p}$ as the maximum value for the memory which is drawn from the distribution $\rho(p)$, then, again, since links are independent in this model, we must have $\Omega(\mathcal{G}) = \bar{p}$ by corollary \ref{corr:independent_mem}.

\item CDARN($p$): {This model, as introduced in \cite{williams2019_diff}, represents a simple way of extending the DARN($p$) model to include correlations between the dynamics of links.} Similarly to the DARN($p$) model, however when the link is to copy its state from memory (with probability $q$), it does not necessarily copy this from its own past history: with probability $1-c$ it will copy from {its self}, and with probability $c$ it will choose one of the other $L-1$ links uniformly at random and copy that link state uniformly at random from the near past up to $p$ steps in the past. Clearly the DARN($p$) model is the special case of the CDARN($p$) model in which $c = 0$. This model depends on parameters $\{p,q,y,c\}$.\\
Formally, we can define the CDARN($p$) model in terms of random variables as follows:
the time varying adjacency matrix $A_t = \{a^{ij}_t\}$ is given by
\begin{equation}
	 a_t^{ij} = Q^{ij}_t a_{(t-Z^{ij}_t )}^{M^{ij}_t} + (1-Q^{ij}_t)Y^{ij}_t, 
\end{equation}
where $Q^{ij}_t \sim Bernoulli(q)$, $Y^{ij}_t \sim Bernoulli(y)$,
$Z^{ij}_t \sim Uniform(1,p)$ and $M^{ij}_t$ randomly picks among
the available links in the network, so that for a link process $a^{ij}_t$, 
$\mathbb{P} \left(M^{ij}_t = (i,j) \right) = 1- c$, and $\mathbb{P} \left(M^{ij}_t = (k,l) \right) = (1-c) / (L-1)$ 
for $(i,j) \neq (k,l)$ and $L = N(N-1)/2$. We now state and prove a theorem on the scalar memory of this temporal network process:
\begin{theorem}
Let $\mathcal{G}$ be a temporal network generated by CDARN($p$). Then $\Omega(\mathcal{G})=p$.
\end{theorem}
\proof{To prove that the scalar memory of a CDARN($p$) 
network is indeed $p$, let us consider the conditional probability of observing a link directly.
First, fix link $(i,j)$, and let us associate with each link a linear label $(i,j) \to \ell \in 1,...,L$, then the conditional probability of observing a link at time $t$ given the past $p$
steps is given by
\begin{equation}
	\mathbb{P}(a^{\ell}_t = 1 | A_{t-1},..., A_{t-p}) = (1-q)y + q \left( \frac{c}{(L-1)p} \sum_{\ell' \neq \ell} \sum_{k=1}^{p} a^{\ell'}_{t-k} + \frac{(1-c)}{p} \sum_{k=1}^{p} a^{\ell}_{t-k} \right).
\label{CDARNp_conditional}
\end{equation} 
Since, by construction, 
\begin{equation}
	\mathbb{P}(a^{1,L}_t = 1 | A_{t-1},..., A_{t-p}) = \prod_{\ell} \mathbb{P}(a^{\ell}_t = 1 | A_{t-1},..., A_{t-p}),
\end{equation}
and the dynamics of each link are symmetric under any
{relabelling of links in}
the temporal network, it is 
enough for us to consider a single link.\\ 
Now consider the same conditional, but with some number $\epsilon$ extra past steps: $\mathbb{P}(a^{\ell}_t = 1 | A_{t-1},..., A_{t-(p+\epsilon)})$. 
Since the memory term $t - Z_t^{\ell}$ will never take the values $t- (p+1)$ to $t-(p+\epsilon)$, the conditional in Eq.~\ref{CDARNp_conditional}
will be unchanged, hence $\Omega(\mathcal{G}) \leq p$. Now we look at what happens when we 
remove some number $d < p$ of past state from the conditional, defining $\delta = p-d$:
\begin{eqnarray}
\begin{aligned}
	\mathbb{P}(a^{\ell}_t = 1 | A_{t-1},..., A_{t-\delta}) = \sum_{a^{1,L}_{t-\delta,t-p}} \mathbb{P}(a^{\ell}_t = 1 | A_{t-1},..., A_{t-p})\mathbb{P}(a^{1,L}_{t-\delta,t-p}), \\
	= (1-q)y  \sum_{a^{1,L}_{t-\delta,t-p}} \mathbb{P}(a^{1,L}_{t-p}) +  \sum_{a^{1,L}_{t-\delta,t-p}} \sum_{k=1}^{p} \left( \frac{q(1-c)}{p} a^{\ell}_{t-k} \mathbb{P}(a^{1,L}_{t-\delta,t-p}) + \frac{qc}{(L-1)p} \sum_{\ell' \neq \ell} a^{\ell'}_{t-k} \mathbb{P}(a^{1,L}_{t-\delta,t-p}) \right).
\end{aligned}
\end{eqnarray}
From this not only do we see that $\mathbb{P}(a^{\ell}_t = 1 | A_{t-1},..., A_{t-\delta}) \neq \mathbb{P}(a^{\ell}_t = 1 | A_{t-1},..., A_{t-p})$, for any such $\delta$, assuming that $q,y\neq1$ or $0$, in which case there is no memory. Hence we must have 
$\Omega(\mathcal{G}) \geq p$. The only remaining option then is $\Omega(\mathcal{G}) = p$, concluding our proof.\qed}
\end{itemize}
\begin{itemize}
\item eCDARN($p$): The extension to the CDARN($p$) model, {as invented for this work,} is one in which the memory kernel from which we sample when the link copies from its own history can be different to the memory kernel when the link copies from another link. That is, not only are links allowed to have different memory lengths, but also the memory lengths {used} when a link refers to its self are allowed to be 
different to the memory length {used} when referring to other links. In this way the random variable $Z^{ij}_t$ now becomes dependent on the value of 
$M^{ij}_t$. If for a link $(i,j)$ the value $M^{ij}_t = (i,j)$, then $Z^{ij}_t \sim Uniform(1,p_{\text{self}}^{ij} )$,
but if $M^{ij}_t \neq (i,j)$ then $Z^{ij}_t \sim Uniform(1,p_{\text{other}}^{ij})$, for two, possibly different, values
of $p_{\text{self}}^{ij}$ and $p_{\text{other}}^{ij}$. The following statement can now be proved:
\begin{theorem}
Let $\mathcal{G}$ be generated by the eCDARN($p$) above. Then $\Omega(\mathcal{G}) = \max_{ij} \max (p_{\text{self}}^{ij}, p_{\text{other}}^{ij})$.
\end{theorem} 
\proof{As before, 
we first write down the conditional probability:
\begin{equation}
	\mathbb{P}(a^{\ell}_t = 1 | A_{t-1},..., A_{t-p}) = (1-q)y + q \left( \frac{c}{(L-1)p_{other}} \sum_{\ell' \neq \ell} \sum_{k=1}^{p_{other}} a^{\ell'}_{t-k} + \frac{(1-c)}{p_{self}} \sum_{k=1}^{p_{self}} a^{\ell}_{t-k} \right).
\label{eCDARNp_conditional}
\end{equation} 
Again, by construction in this model we have
\begin{equation}
	\mathbb{P}(a^{1,L}_t = 1 | A_{t-1},..., A_{t-p}) = \prod_{\ell} \mathbb{P}(a^{\ell}_t = 1 | A_{t-1},..., A_{t-p}),
\end{equation}
hence $p = \min_n (n: \mathbb{P}(a^{1,L}_t  | A_{t-1},..., A_{t-n}) = \mathbb{P}(a^{1,L}_t  | A_{t-1},..., A_{t-\infty}))$ 
if and only if $ p \geq \min_n (\mathbb{P}(a^{\ell}_t  | A_{t-1},..., A_{t-n}) = \mathbb{P}(a^{\ell}_t  | A_{t-1},..., A_{t-\infty}))$ for all $\ell$
and $ p = \min_n (\mathbb{P}(a^{\ell}_t  | A_{t-1},..., A_{t-n}) = \mathbb{P}(a^{\ell}_t  | A_{t-1},..., A_{t-\infty}))$ for at least one value of $\ell$. 
Hence, if for each $\ell$ we have $ \min_n (\mathbb{P}(a^{\ell}_t  | A_{t-1},..., A_{t-n}) = \mathbb{P}(a^{\ell}_t  | A_{t-1},..., A_{t-\infty})) = \max(p_{\text{self}}, p_{\text{other}})$, then we must have $\Omega(\mathcal{G}) =  \max_{ij} \max (p_{\text{self}}^{ij}, p_{\text{other}}^{ij})$.
All that then remains to prove is that $ \min_n (\mathbb{P}(a^{\ell}_t  | A_{t-1},..., A_{t-n}) = \mathbb{P}(a^{\ell}_t  | A_{t-1},..., A_{t-\infty})) = \max(p_{\text{self}}, p_{\text{other}})$. This is a trivial extension of the proof for the CDARN($p$) model, but with the conditional probability 
now being of the form in Eq.~\ref{eCDARNp_conditional}. Hence we must have that $\Omega(\mathcal{G}) = \max(p_{\text{self}}, p_{\text{other}})$, 
as required.\qed}
\end{itemize}

Finally, each model in practice also depends on an additional variable: the length of time series $T$ for which it is sampled.
While this does not influence the dynamics of the model it will influence any estimated value for
the memory, and so we consider it here.
Altogether, these models give us a wide range of test cases with a number of features that we might expect from 
real world networks.

\begin{remark} It's important to note that, in theory, virtual loops cannot emerge in the DARN($p$) or the eDARN($p$) models, but in principle should emerge in the CDARN($p$) and eCDARN($p$) models as in these latter cases we are probabilistically coupling links, and {this} coupling can induce casual loops (possibly of different orders) among sets of links. This means that we expect $\Omega_{\text{\text{eff}}}(\mathcal{G})$ to coincide with $\Omega(\mathcal{G})$ for the DARN($p$) or the eDARN($p$) models, but we should find a difference in the CDARN($p$) {and eCDARN($p$)} models. In the next subsection we will investigate the extent of that mismatch, and the role played by virtual loop decoherence.
\end{remark}

\begin{remark}
{In general, when a given link $\alpha\in[1,2,\dots,L]$ samples its future state from the past of a different link in an eCARN($p$) model, one can specify which is the set of links from which $\alpha$ will sample from. As discussed in section IIIb, a natural way to encode this is by building up a Bayesian causal graph of $L$ nodes, where each of the nodes corresponds to a link in the original temporal network. The $L \times L$ adjacency matrix ${\bf C}=\{c_{\alpha\beta}\}$ is such that $c_{\alpha\beta}=1$ if $\beta$ is in the set from which $\alpha$ can sample its future from, and 0 if $\beta$ is not in this set.\\
Different Bayesian causal graphs can thus be specified as to describe the set of links' past from which a given link copy its future state. In Fig1 of the main manuscript we choose two different examples to showcase how different link causal structures bring about different co-memory matrices. Example (a) in Fig 1 of the main manuscript has a Bayesian causal ring (i.e. $c_{\alpha\alpha\pm1}=1$ and $c_{\alpha\beta}=0$ for $\beta\ne \alpha\pm1$, subject to periodic boundary conditions), where link $\alpha \in [1,2,\dots,L]$ updates its future state either from its past uniformly between 0 and $p$, or from the past of $\alpha \pm 1$. Example (b) in Fig1 of the main manuscript on the other hand has more elaborate Bayesian causal graph described in fig.\ref{fig:cau} below, and when link $\alpha$ updates its future, it does so following an eCDARN($p$) model by either looking at its own past and randomly sampling it between 0 and $p_\text{self}^\alpha$ states in the past, where  $p_\text{self}^\alpha \sim Uniform\{3,6\}$, or sampling the past of one of $\alpha$'s neighbourhood in the link causal graph, by copying uniformly between 0 and $p_\text{other}^\alpha$ states in the past, where $p_\text{other}^\alpha\sim Uniform\{3,6\}$.}
\end{remark}

\begin{figure}[htb]
  \includegraphics[width = 0.38\textwidth]{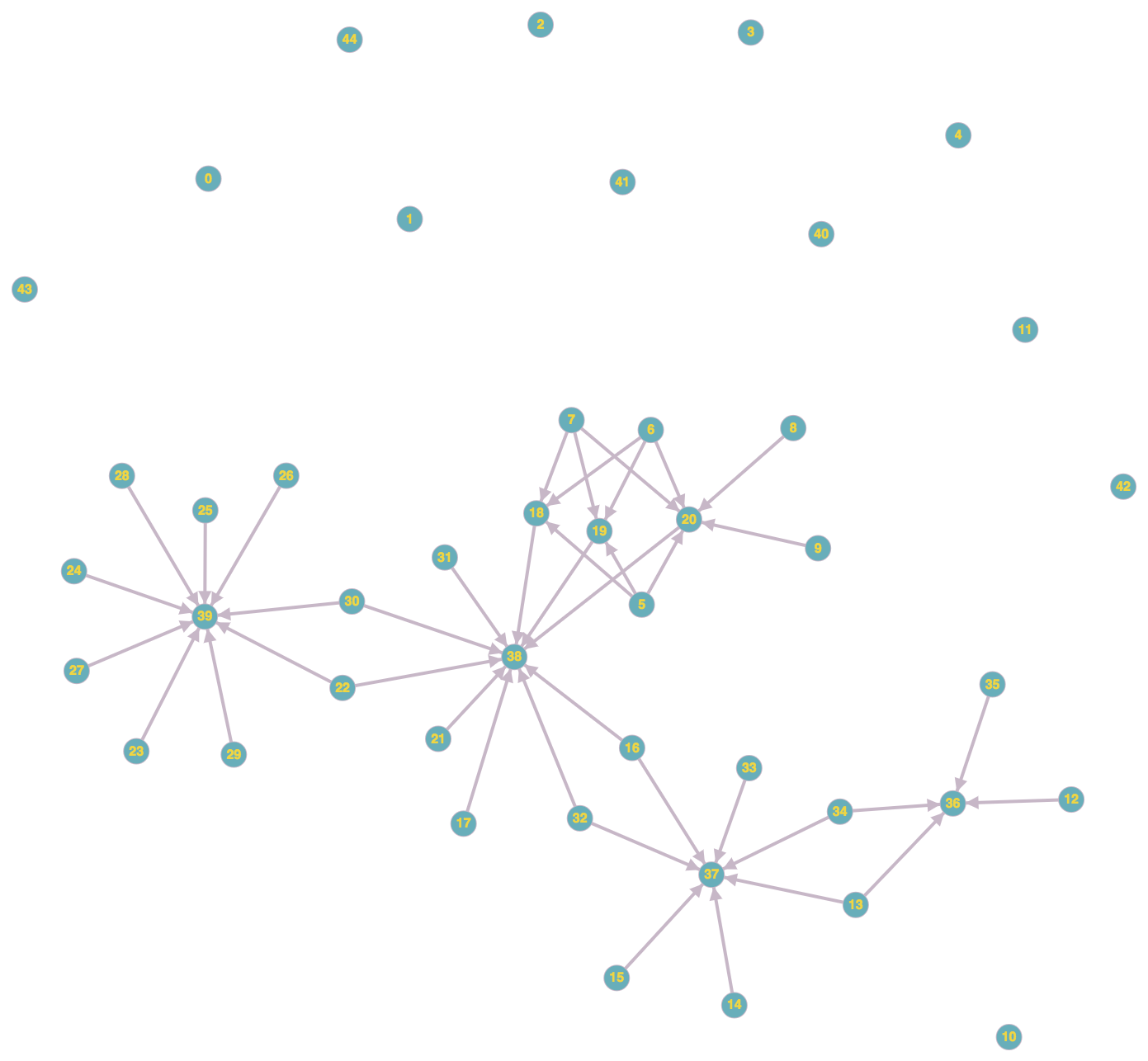}%
\caption{Bayesian causal graph of the specific eCDARN($p$) model used in the panel (b) of Fig.1 of the main manuscript. Nodes in this graph correspond to the links in the original temporal network, and two nodes are connected by a directed edge if the associated pair of links in the original network are causally connected. In this particular examples, only a subset of links in the original temporal network can actually update their state from the past of other links, where some others evolve independently.}
\label{fig:cau}
\end{figure}

\subsection{Estimator accuracy on synthetic networks}
For each of the four synthetic models we have described, 
we have studied two key quantities that help us to compare the estimated values of both 
$\Omega_{\text{\text{eff}}}$ and $\Omega_{\text{pair}}$ with the analytical value of $\Omega(\mathcal{G})$.
These are the {\it hit rate} and the {\it average distance}, as earlier presented in Fig.~\ref{fig:toy_hit_rates}.
The hit rate gives the probability that a given estimate for the value $\Omega_{\text{\text{eff}}}$ (or $\Omega_{\text{pair}}$) 
is {\it precisely} the scalar memory of the network, while the average distance gives the 
value of either $\left| \Omega_{\text{\text{eff}}}(\mathcal{G}) - \Omega(\mathcal{G}) \right|$ or $\left| \Omega_{\text{pair}}(\mathcal{G}) - \Omega(\mathcal{G}) \right|$,
averaged over several realisations of the network model.
In each case the networks generated have a fixed number of nodes $N=10$.
We then allow in turn one of the parameters $q,y,T$ (and, where applicable $c$) to vary,
while fixing the others to the following values: $q=0.9, y= 0.1, c=0.1, T=10^6$.
For each set of parameters $10^3$ realisations of the model are generated, each with 
randomly chosen values for the memory lengths 
($p$ for DARN($p$) and CDARN($p$), $p^{ij}$ for eDARN($p$) and ($p^{ij}_{\text{self}},p^{ij}_{\text{other}}$) for eCDARN($p$))
from the range $1,...,10$. We then plot the hit rates and average distances for each model 
as a function of each free parameter in Fig.~\ref{fig:synthetic_hit_rates}.\\

\begin{figure}[htb]
  \includegraphics[width = 0.4\textwidth]{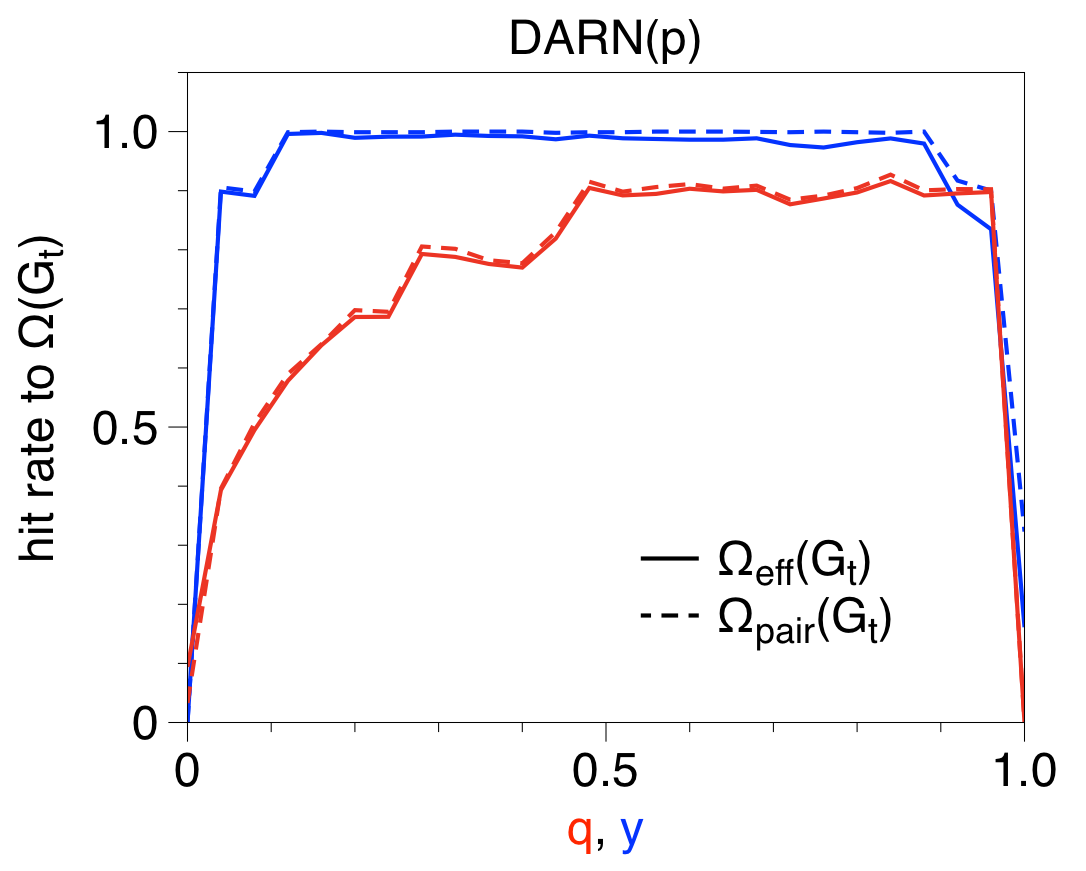}
    \includegraphics[width = 0.4\textwidth]{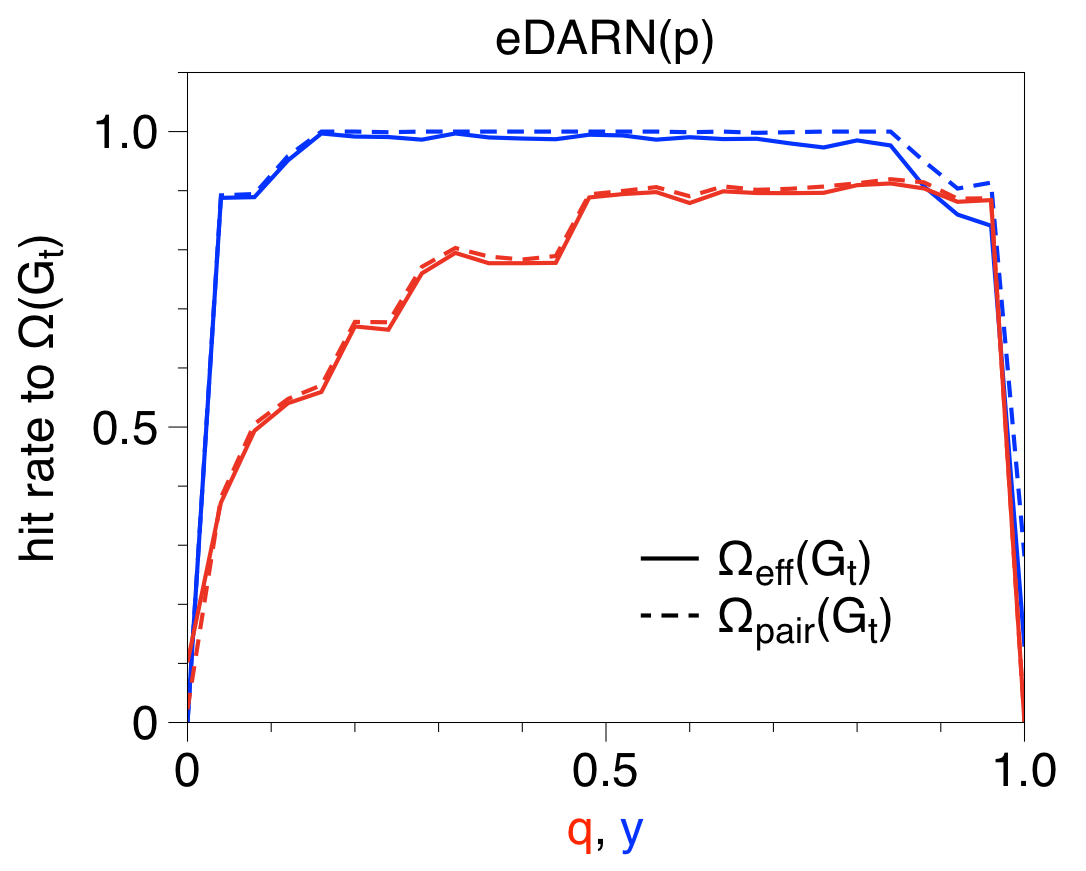}
      \includegraphics[width = 0.4\textwidth]{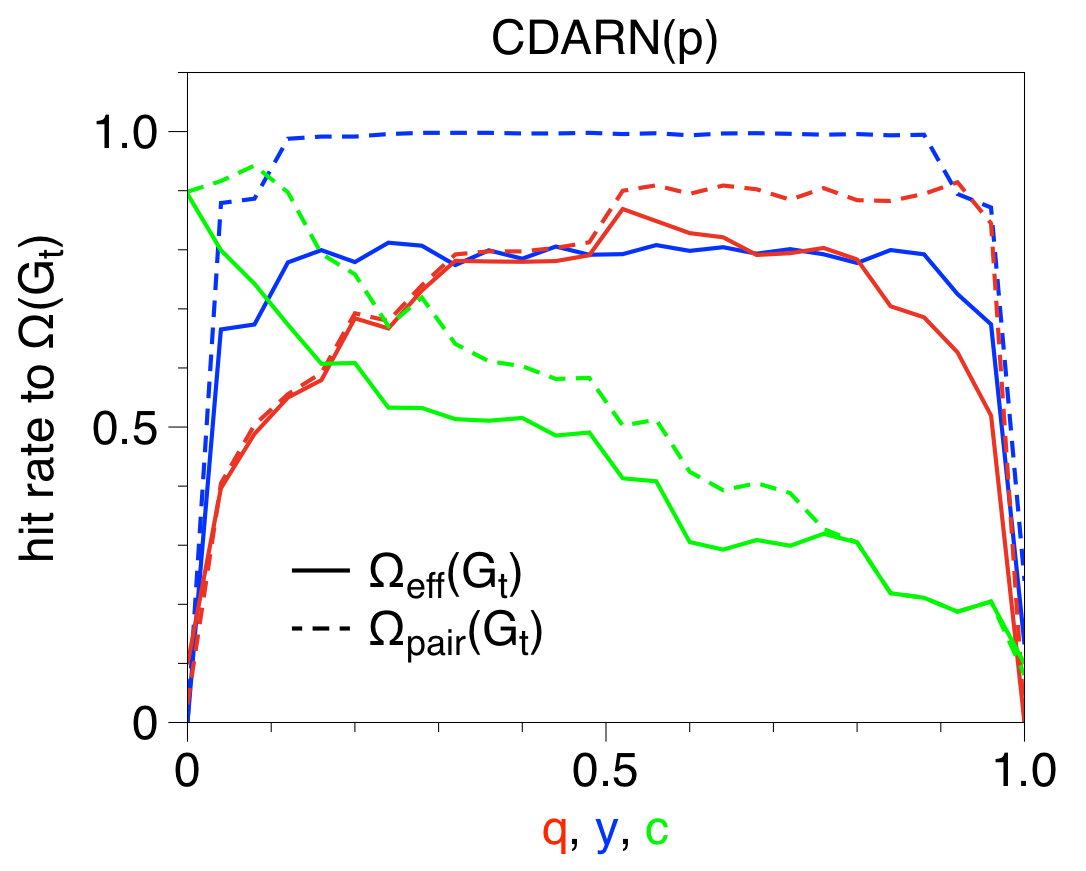}
        \includegraphics[width = 0.4\textwidth]{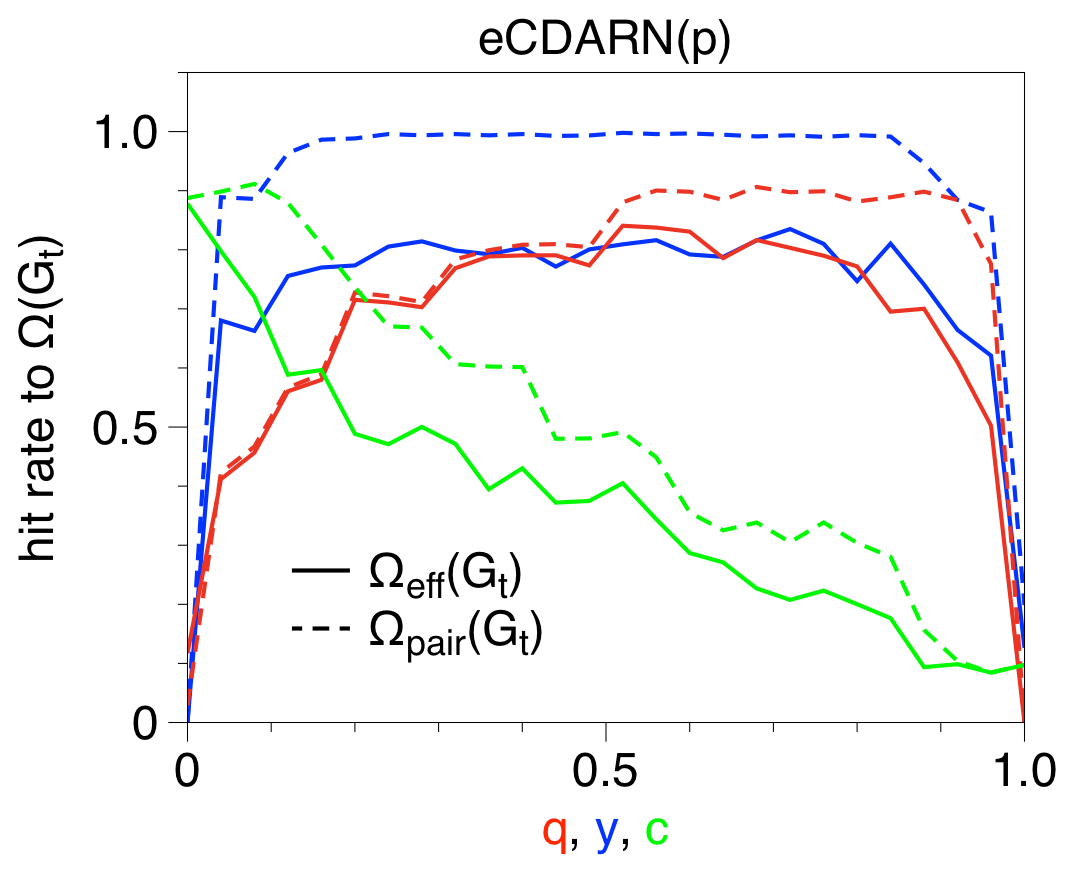}

\caption{{\bf Hit rates to scalar memory $\Omega(\mathcal{G})$ for synthetic models.} In each case we compute the percentage of the times within an ensemble of $10^3$ realisations that the estimated effective memory $\Omega_{\text{eff}}(\mathcal{G})$ and estimated pair memory $\Omega_{\text{pair}}(\mathcal{G})$ exactly match the scalar memory $\Omega(\mathcal{G})$ (the memory parameter $p$ is randomly sampled from \textsc{Uniform}\{{1,...,10}\} for each realisation). 
Models depend on parameters $q, y$ and (where applicable) $c$, so each curve scans hit rates for the whole range of a given parameter and fix the values of the other parameters to $q=0.9, y= 0.1, c=0.1$ (in every case, time series size is $T=10^6$). In DARN($p$) and eDARN($p$) models where virtual loops are by construction absent, $\Omega_{\text{eff}}(\mathcal{G})=\Omega_{\text{pair}}(\mathcal{G})$ and their estimation typically coincide with $\Omega(\mathcal{G})$ for a large range of model parameters, as expected. In CDARN($p$) and eCDARN($p$) models, (probabilistic) virtual loops are expected to kick in, inducing a mismatch between $\Omega_{\text{eff}}(\mathcal{G})$ and $\Omega(\mathcal{G})$ (the mismatch is notably smaller for $\Omega_{\text{pair}}(\mathcal{G})$ as this quantity disregards diagonal entries of the co-memory matrix).
}
\label{fig:synthetic_hit_rates}
\end{figure}
\noindent Let us discuss these results in detail. Considering the DARN($p$) and eDARN($p$) models:
\begin{itemize}
\item First, we observe that 
both $\Omega_{\text{\text{eff}}}(\mathcal{G})$ and $\Omega_{\text{pair}}(\mathcal{G})$ coincide with $\Omega(\mathcal{G})$ and provide a very good estimates for a wide 
range of parameters. Indeed they are both 100\% accurate when $q=0.9$
and $y$ is between $\sim0.2$ and $\sim0.8$. 
\item Notice that the ranges where the performance is worse are actually expected: when $y$ is either 
very small or very large, then each link series $E_t^i$ will be dominated by 
either 0 or 1, in the limit of all 1's or 0's we would observe no memory as the system is 
deterministic, and indeed close to this we would expect it to be hard to observe any memory.
This manifests as a sharp drop in both hit rate and average distance.
\item When $q$ is small we would also expect memory to be harder to detect, as 
it is used less often, and hence any correlations with the past are less significant.
However, the increase in hit rate is a more smooth function of $q$, as is the average distance.
\item The two memories ($\Omega_{\text{\text{eff}}}(\mathcal{G})$ and $\Omega_{\text{pair}}(\mathcal{G})$) also perform identically 
for both models. This is again expected: by construction there are no virtual loops of any kind in these network models,
and so $\Omega_{\text{\text{eff}}}(\mathcal{G}) = \Omega_{\text{pair}}(\mathcal{G})$.
\end{itemize}
In the case of the CDARN($p$) and eCDARN($p$) models:
\begin{itemize}
\item  We observe that $\Omega_{\text{\text{eff}}}(\mathcal{G}) \ne \Omega_{\text{pair}}(\mathcal{G})$. This is again expected: note that by construction the CDARN($p$) and eCDARN($p$) models should have a considerable number of 
virtual loops, as induced by the cross-correlations that are present between each link, and indeed we would
expect these virtual loops to be of a variety of orders. By definition, $\Omega_{\text{pair}}(\mathcal{G})$ cannot capture {the effects of virtual loops}
{of order greater than one}, hence the mismatch between {these two measures}.
\item Interestingly,  we also observe that under a range of parameters, both $\Omega_{\text{\text{eff}}}$ and $\Omega_{\text{pair}}$ remain quite close to $\Omega(\mathcal{G})$, manifesting a strong virtual loop decoherence in these cases.
\item Similarly to before, there are also ranges of both $y$ and $q$ where `performance' is worse, and this is also expected.

\item {We also see that as $c\to 1$ the hit rate significantly drops.
This can be explained by an increase in the significance of virtual loops.
When $c=0$ then links are independent, and so there are no virtual loops to 
influence $\Omega_{\text{eff}}(\mathcal{G})$. As we increase $c$ we allow for virtual loops, however 
they will be decoherent, and so the hit rate will still initially be high. As $c$ approaches 1,
we have made the influence of virtual loops as strong as possible, and so the hit rate will be 
at its lowest.}
\end{itemize}

\subsection{Virtual loops in the synthetic networks}
Let us further emphasise here the role played by virtual loops in the four synthetic models discussed above.
The hit rate and average distance between $\Omega(\mathcal{G})$ and $\Omega_{\text{\text{eff}}}(\mathcal{G})$ 
demonstrate a number of features of the effects of virtual loops on memory in temporal networks.
For both the DARN($p$) and eDARN($p$) models, where loops are not present, the values of both
$\Omega_{\text{\text{eff}}}(\mathcal{G})$ and $\Omega_{\text{pair}}(\mathcal{G})$ are almost identical, and for most parameters give a very good approximation of
$\Omega(\mathcal{G})$. For the CDARN($p$) and eCDARN($p$) models this is not the case. The two correlated models should
clearly display virtual loops, as they are inherent in the structure of the models in the same way as
the toy models described earlier. Indeed, we know that when calculating $\Omega_{\text{pair}}(\mathcal{G})$ then 
a number of the virtual loops will be accounted for, and so the significant differences between 
the hit rates and distances of
$\Omega_{\text{\text{eff}}}(\mathcal{G})$ and $\Omega_{\text{pair}}(\mathcal{G})$ for both the CDARN($p$) and eCDARN($p$) models
can confidently be ascribed to the presence of these loops.
It is unexpected however that $\Omega_{pair}(\mathcal{G})$ accounts for the full extend of these loops;
virtual loops of varying orders should be present in this system. 
Furthermore, in both cases we
observe that $\Omega_{\text{\text{eff}}}(\mathcal{G})$ and $\Omega_{\text{pair}}(\mathcal{G})$ are reasonable estimates for the scalar memory
$\Omega(\mathcal{G})$ for a range of parameter values. This can only be the case if, in the same way as for our toy models, these virtual loops
display decoherence, and are hence the resulting increase in pair and effective memory is masked.\\ 

\subsection{A theorem on full virtual loop decoherence effect for large network size} 
The CDARN($p$) model, as we have presented here, maintains a high degree of symmetry: each link has the same memory strength $q$, link density $y$ and memory length $p$. Each link also, with the same probability $c$, can look into the past $p$ states of every other link, and will do so uniformly. In summary, the network would behave in exactly the same way if the links were swapped (in general if there were any isomorphism). One might be tempted to think that this ``synchronicity" might bring out memory and virtual loops in an extreme way, since virtual loops of every order should be present in the system. Indeed we see that in our relatively small synthetic networks, when $c$ is large, and hence there is a strong reliance on virtual loops, the scalar and effective memory of the network seldom coincide. When these networks get larger however, we observe something unexpected: all virtual loops become increasingly decoherent. As the number of links grows larger, so does the pool of links from which a past state can be drawn, because of this we find that the memory in virtual loops ``averages out", a concept which we formally detail in the theorem below.

\begin{theorem}
Consider a CDARN($p$) network with $L$ links, memory strength $q$, link density $y$ and memory length $p$.
The conditional probability of a link $\ell$ occurring at time $t$, given the past $p$ states of the network is as follows:
\begin{equation}
	\mathbb{P}(a^{\ell}_t | a^{1,L}_{t,t-p} ) = (1-q)y + q \left( (1-c) \phi_{self} + c \phi_{other} \right),
\end{equation}
where $\phi_{self}$ and $\phi_{other}$ represent the contributions to the conditional from the past $p$ states of the link $\ell$ and every other link respectively. As $L \to \infty$, $\phi_{other}$ tends to a constant, and hence the link $\ell$ has no memory of the past states of any other link.
\label{theo:CDARN_VL}
\end{theorem}
\begin{proof}
The conditional probability of observing a link $\ell$ at time $t$ given the past $p$ states of the network is as follows:
\begin{equation}
	\mathbb{P}(a^{\ell}_t = 1 | a^{1,L}_{t,t-p}) = (1-q)y + q \left( \frac{(1-c)}{p} \sum_{k=1}^{p} a^{\ell}_{t-k} + \frac{c}{(L-1)p} \sum_{\ell' \neq \ell} \sum_{k=1}^{p} a^{\ell'}_{t-k} \right).
\end{equation}
We can therefor see that our memory kernels $\phi_{self}$ and $\phi_{other}$ are given as:
\begin{eqnarray}
\begin{aligned}
	\phi_{self} =&  \frac{(1-c)}{p} \sum_{k=1}^{p} a^{\ell}_{t-k},\\
	\phi_{other} =& \frac{c}{(L-1)p} \sum_{\ell' \neq \ell} \sum_{k=1}^{p} a^{\ell'}_{t-k}.
\end{aligned}
\end{eqnarray}
We need only focus on $\phi_{other}$.
First, let us consider the average value
\begin{equation}
	\left< a^{\ell'}_{t-k} \right>_{\ell'} = \mathbb{P}(a^{\ell'}_{t-k} = 1).
\end{equation}
The CDARN($p$) network is taken to be in a stationary state, and so the symmetry of the links under isomorphism 
guarantees us that $P(a^{\ell'}_{t-k})$ is the same for each link $\ell'$ and for each time $t-k$. Hence we can write 
$ P(a^{\ell'}_{t-k}) = \bar{a}$ for some constant $\bar{a}$. Then we must have, for any of the $L-1$ possible values of $\ell'$,
\begin{equation}
	\left< a^{\ell'}_{t-k} \right>_{\ell'} = \bar{a}.
\end{equation}

Now, $\phi_{other}$ can be re-written as follows:
\begin{equation}
	\phi_{other} = \frac{1}{p} \sum_{k=1}^p  \frac{1}{L-1} \sum_{\ell' \neq \ell} a^{\ell'}_{t-k}.
\end{equation} 
Then, by the law of large numbers we can express this in terms of the sample average:
\begin{eqnarray}
\begin{aligned}
	\phi_{other} =& \frac{1}{p} \sum_{k=1}^{p} \left< a^{\ell'}_{t-k} \right>_{\ell'}, \\
	=&  \frac{1}{p} \sum_{k=1}^{p} \bar{a},\\
	=&\bar{a}.
\end{aligned}
\end{eqnarray}
Hence $\phi_{other} \to \bar{a}$ as $L \to \infty$. Since there are no terms containing links other than $\ell$ in $\phi_{self}$, then we can conclude that the conditional probability is such that, in the same limit $L \to \infty$,
\begin{equation}
\mathbb{P}(a^{\ell}_t = 1 | a^{1,L}_{t,t-p})  \to \mathbb{P}(a^{\ell}_t = 1 | a^{\ell}_{t,t-p}),
\end{equation}
and so any memory of other links is lost.
\end{proof}

To summarise, in this section we have validated that the co-memory matrix correctly displays the memory of synthetic temporal networks. The maximum of this co-memory matrix, defined as the effective network memory $\Omega_{\text{\text{eff}}}(\mathcal{G})$ coincides with the scalar memory $\Omega(\mathcal{G})$, {and so is a good estimate of it,} when virtual loops are not present. 
In the cases where these {loops} are present, the two quantities in principle differ, and we know that one should consider $\Omega_{\text{\text{eff}}}(\mathcal{G})$ rather than $\Omega(\mathcal{G})$ if {in this} case. 
{Interestingly,} virtual loop decoherence {is also clearly present in} these networks. Moreover, when the network size (number of {links}) gets larger, theorem \ref{theo:CDARN_VL} suggests that we {asymptotically should} expect full virtual loop decoherence. This is particularly important when considering real-world temporal networks. We can hence conclude that while virtual loops play a role {in the memory of these networks}, in practice virtual loop decoherence might limit that role, and as such we would expect that $\Omega_{\text{\text{eff}}}(\mathcal{G})$ does not substantially differ from $\Omega(\mathcal{G})$ in real, complex temporal networks.

\section{Applications to real-world temporal networks}

\subsection{Processing empirical network data}
We have here presented temporal networks taken from 6 data sets, and 
at two different temporal resolutions. A varying amount of work 
has to be done to each of the used data sets before they can be fed into 
our memory estimator as a discrete time temporal network.
As part of the software developed for this research, we have 
developed a method to convert a time stamped edge list of the form
$(v_1,v_2,t)$ into a set of time series $X_t^{v_1,v_2}$ which represents
the edge process connecting nodes $v_1$ and $v_2$. Once we have this time series
then we can make use of our memory estimators. Hence the aim is 
now to take each data set and convert it into this time stamped edge list form.
We will here give a brief overview of each dataset and the steps required to process it.\\

\noindent The first two datasets correspond to a type of network which we can label as online social communication networks:
\begin{itemize} 
\item Text message interactions between college students (CM) \cite{panzarasa2009patterns}. This data 
represents messages sent between (anonymised) student users of an online communication platform
at the University of California, Irvine, over a period of 7 months.
Processing this is a simple task as it comes in what is almost a suitable format; 
the raw data is given as triplets $(I.D_1, I.D_2, t)$ for each pair of individuals
$I.D_1$ and $I.D_2$ interacting at time t measured in seconds from some starting 
value. Here we simply index the individuals from 0 to $N$ and shift the $t$ values 
so that the first link in the data occurs at time 0.
\item Email communications (EM) data set \cite{email_bb}.
This data set covers internal e-mail communications
between employees of a mid-sized manufacturing company 
over a period of nine months.
The data is (ignoring other irrelevant data) in the form
$(I.D_1, I.D_2, DT)$ where now $DT$ is in a date-time format to a resolution of seconds.
Hence in constructing our time stamped edge list we again index the $I.D$ values appropriately,
and again convert the date-time to the number of seconds elapsed since the first 
link in the data set occurred.
\end{itemize}
The third dataset is a social interaction or contact network (RM), that is to say, it is also a social network like the first two cases, but it is an `offline' one and can also be seen as a mobility network.
This data is taken from the ``Reality Mining" data set \cite{eagle2006reality}, collected from the interaction of 94 students at MIT
over 8 months. The data we have used here was taken from bluetooth interactions between
the phones carried by the subjects of the study. Each interaction indicates that the two individuals 
were within at least 5-10 meters of each other. Each device scans for other devices in its proximity every 5 minutes, 
however because a pair of devices can recognise each other independently this can produce 
two interactions every 5 minutes with an average inter interaction duration of 2.5 minutes.
The data is (ignoring other irrelevant data) in the form
$(I.D_1, I.D_2, DT)$ where now $DT$ is in a date-time format to a resolution of seconds.
Hence in constructing our time stamped edge list we again index the $I.D$ values appropriately,
and again convert the date-time to the number of seconds elapsed since the first 
link in the data set occurred.\\

\noindent The three remaining datasets correspond to what we could label as engineered transportation networks designed for the public transport in Paris \cite{kujala2018collection} via bus (PB), train (PT) and underground (PU), i.e. these are engineered, offline infrastructure networks.
These are comprised of a weeks worth of records for the movements of public transport 
systems from stop to stop.
This data is structured in a different way; 
each set contains the times taken for each journey occurring between any two 
stops on the bus (PB), train (PT) and underground (PU) public transport systems
in Paris. This data (again ignoring irrelevant data) is of the form 
$(I.D_1, I.D_2, t_{start}, t_{stop})$, where the values $I.D_1$ and $I.D_2$ 
represent the origin an destination stops respectively, and $t_{start}$ and $t_{stop}$
represent the date-times at which the origin was left and the destination was reached 
respectively, to a resolution of one second.\\

\noindent We here index the origin and destinations as expected, and convert the date-time values 
to the number of seconds elapsed since the first date-time value in the data set.
We then take the link between the origin and destination to be present in every second 
which elapses between $t_{start}$ and $t_{stop}$. If two journeys overlap then the link 
is kept {active} until the latter of the two journeys is completed.\\

\noindent {\bf Coarse-graining at different resolution timescales -- } Once we have a time stamped edge list for a data set we filter out the
100 links which occur the most. We then produce two temporal networks by,
 for each link, integrating over two different timescales $\Delta t=1$ minute (60 seconds) and  $\Delta t=10$ minutes (600 seconds).
That is to say, given a temporal network $\mathcal{G}$ with 
edge processes $E_t^i$, and with a unit timescale, which extends 
from time $t=0$ to $t=T$,
we define a new temporal network $\tilde{\mathcal{G}}$ with edge processes 
$\tilde{E}_{k}^i$, and with timescale {$\Delta t = 60, 600$} seconds, which extends 
from time $k=0$ to $k=T/{\Delta t} - 1$. The links in this second network 
are then drawn from the first in the following way: 
$\tilde{E}_k^i = 1$ if for any $t \in \{k \Delta t,..., (k+1) \Delta t \}$, 
$E_t^i = 1$. In this way we effectively ``integrate'' the time series
for each link over our time scale $\Delta t$.

\begin{figure}[htb]
  \includegraphics[width = 0.38\textwidth]{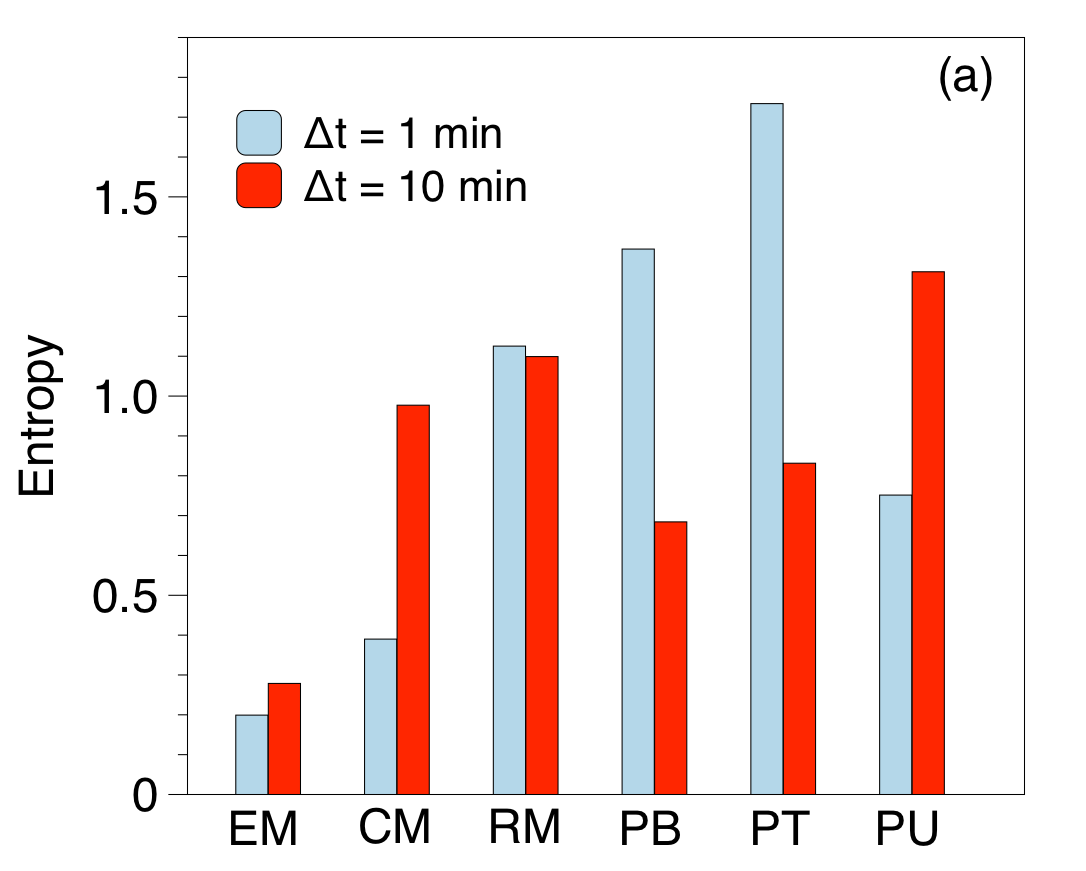}%
  \includegraphics[width = 0.38\textwidth]{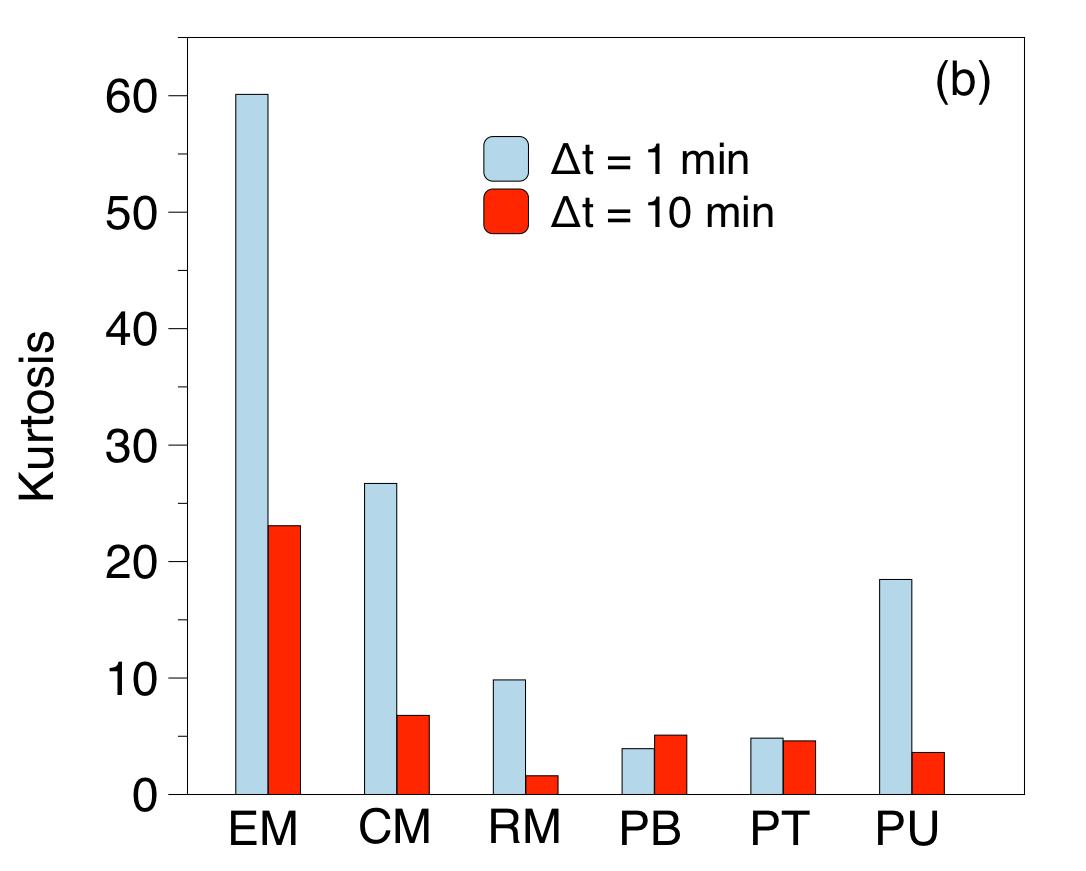}
\caption{{ \bf Assessing co-memory histogram heterogeneity} Entropy (panel a) and kurtosis (panel b) of the co-memory histogram for each of the six real-world temporal networks, at $\Delta t = 1$ min (blue) and $\Delta t = 10$ min (red).}
\label{fig:entropy}
\end{figure}

\subsection{Heterogeneity of co-memory histograms: entropy and kurtosis}
In Figure 3 of the main manuscript we have plotted the co-memory matrix histograms for the six temporal networks at the two resolution timescales. Here we give a further exploration about the shape of these histograms by computing their entropy and kurtosis. The entropy of a probability distribution function $p(x)$ is defined as $-\sum p(x)\log p(x)$ and characterises how `uneven' the distribution is, reaching a maximum when the distribution is uniform and reaching the minimum (zero) when the probability is fully concentrated. Accordingly, entropy {describes} in a scalar metric how heterogeneous the microscopic memory kernel of a given temporal network is. The kurtosis is the fourth standardised moment of $p(x)$, and is a measure of its ``tailedness".\\ 

\noindent Results are shown in Fig.\ref{fig:entropy}. The fact that transportation networks tend to have large entropy suggest that many different memory co-orders are detected, i.e. these networks display a highly heterogeneous memory kernel. On the other hand, we find that the online social networks have a strong kurtosis, meaning that even if most of the links have a weak memory kernel, there are a few whose co-order is large. Since $\Omega_{\text{eff}}(\mathcal{G})$ is defined as the maximum over all co-orders, from the kurtosis analysis we can conclude that online social communication networks analysed in this work display large memory but this only comes from a handful of links.

\subsection{Constructing memory communities and comparing networks in the $(\left< \Omega \right>_{\text{in}}^i, \left< \Omega \right>_{\text{out}}^i)$ plane}
Given the memory heterogeneity displayed in real temporal networks, we wish to further understand how a given links
activity influences the memory of the whole network, and,
in turn, how much the whole network has an influence on it. To quantify this we define two quantities: given the time series $E_t^l$
representing the evolution of each link $l \in 1,...,L$, 
the {\it average outgoing co-order} of a link is defined as
$$\left< \Omega \right>_{\text{out}}^i = \frac{1}{L} \sum_j \Omega(E_t^i || E_t^j).$$
This quantity characterises the net memory effect that the network as a whole has on link $i$. On the other hand, we define the {\it average incoming co-order} 
$$\left< \Omega \right>_{\text{in}}^i = \frac{1}{L}\sum_j \Omega(E_t^j || E_t^i),$$
characterising the net memory effect of the link $i$ on the whole network. The duple $(\left< \Omega \right>_{\text{in}}^i, \left< \Omega \right>_{\text{out}}^i)$ is therefore a compact representation of the role played by each link $i$, in this section we explore scatterplots of $\left< \Omega \right>_{\text{out}}^i$ vs $\left< \Omega \right>_{\text{in}}^i$. More concretly, for each of the six empirical temporal networks we have considered in this work, we focus on the top 100 most active links and make
scatter plots of $\left< \Omega \right>_{out}^i$ vs $\left< \Omega \right>_{in}^i$ for the two different resolution timescales $\Delta t=1$ and $10$ minutes.\\ 

\begin{figure}[htb]
  \includegraphics[width = 0.3\textwidth]{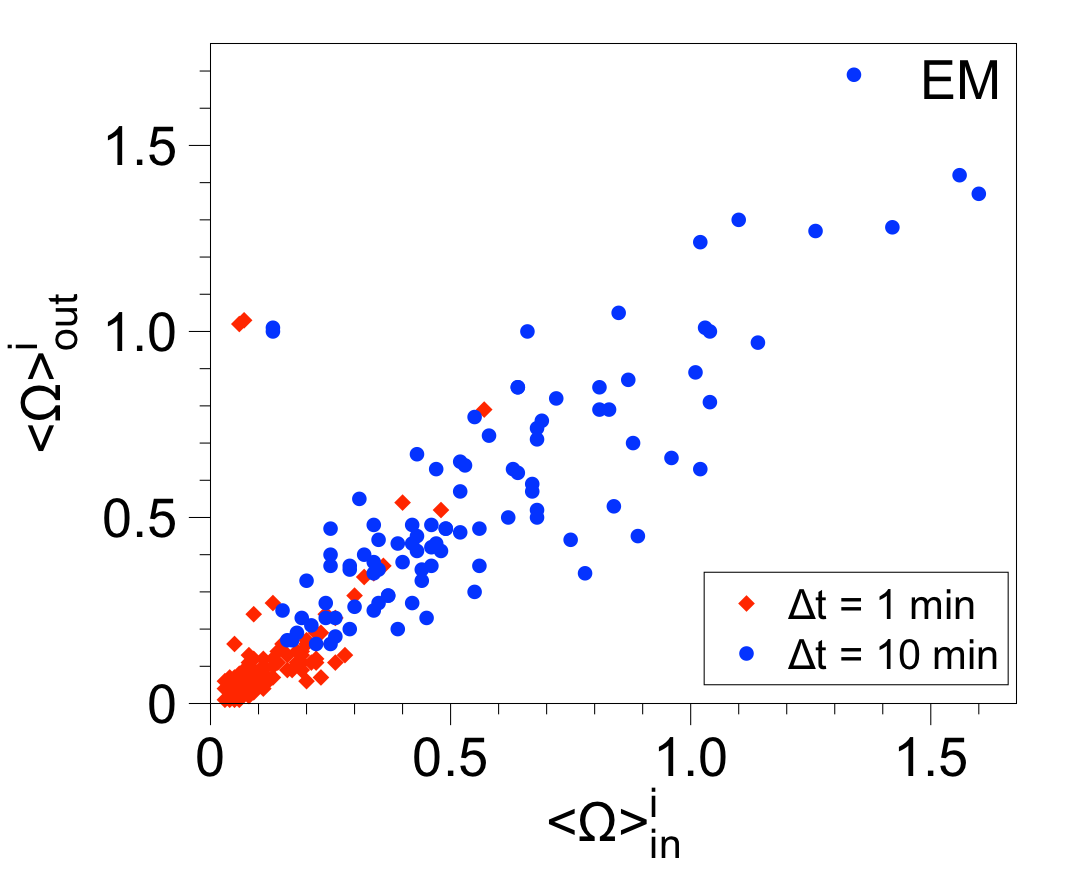}
    \includegraphics[width = 0.3\textwidth]{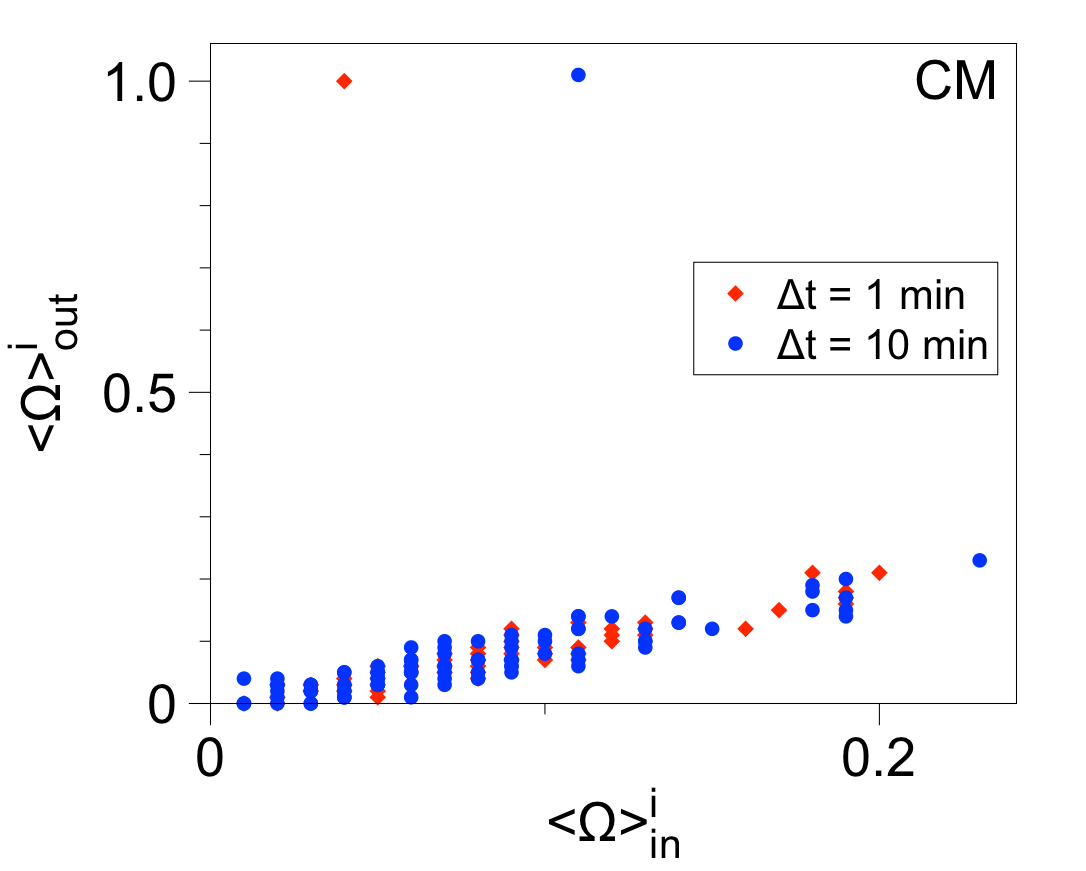}
      \includegraphics[width = 0.3\textwidth]{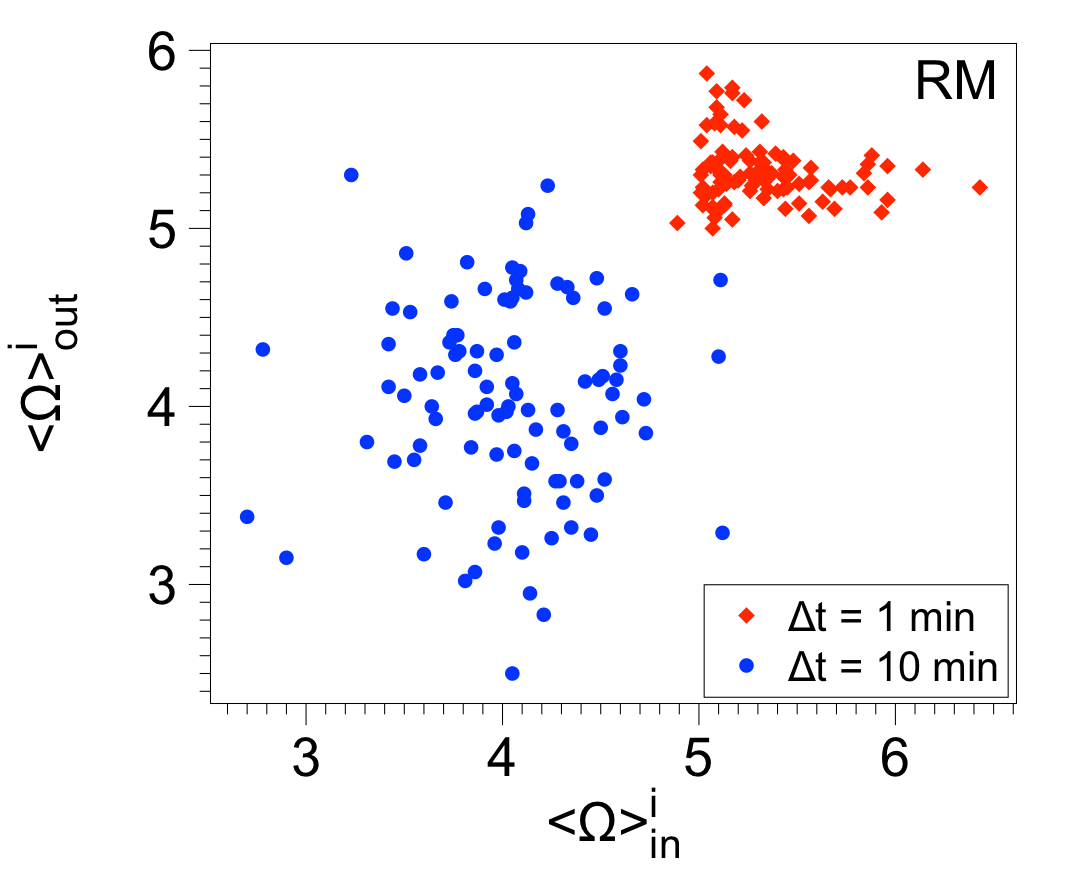}
        \includegraphics[width = 0.3\textwidth]{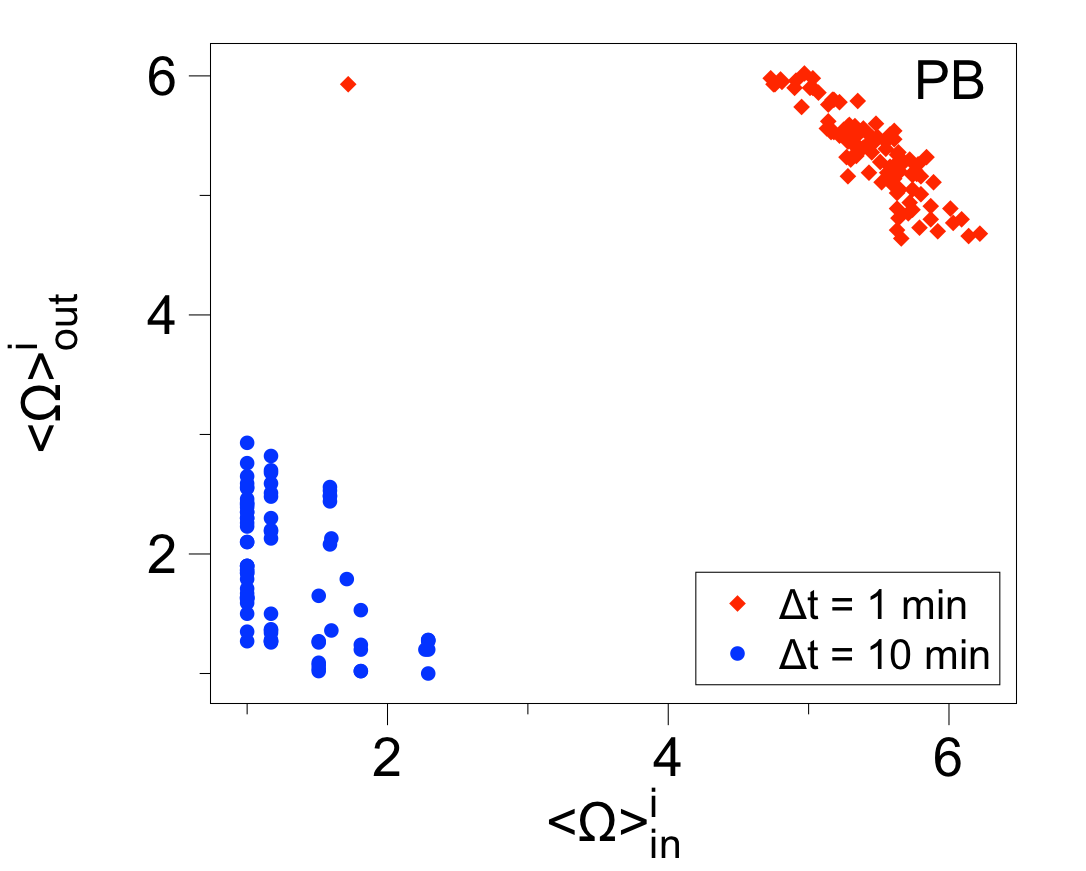}
          \includegraphics[width = 0.3\textwidth]{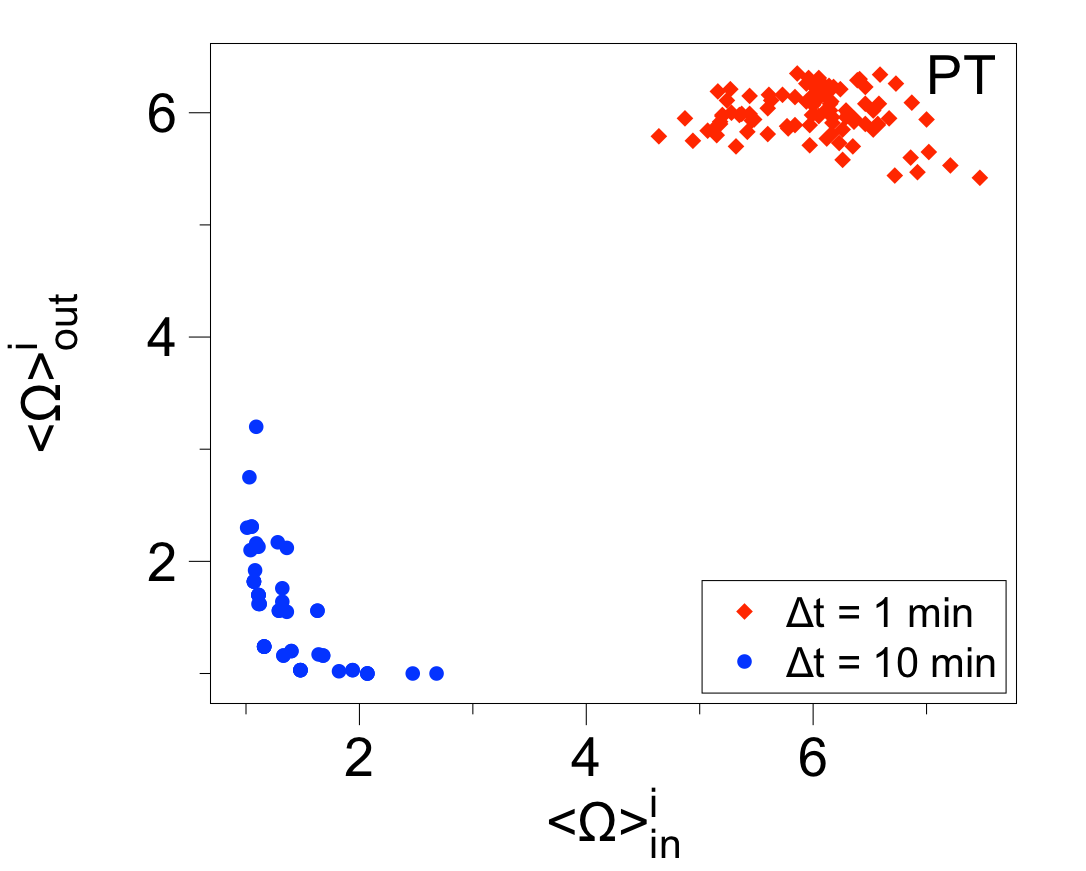}
            \includegraphics[width = 0.3\textwidth]{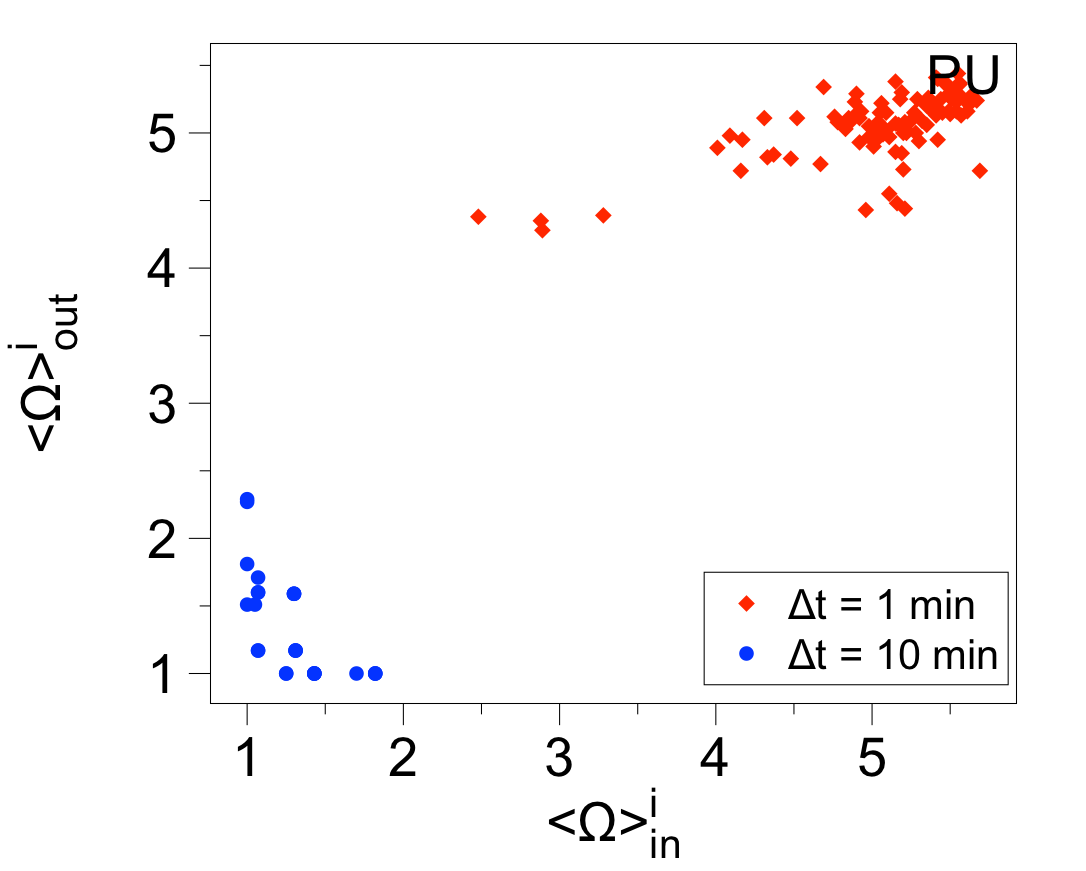}
\caption{{ \bf Comparison between the two resolution timescales.} For the six real-world temporal networks we scatter plot the average outgoing co-order  $\left< \Omega \right>_{out}^i$ vs the average incoming co-order $\left< \Omega \right>_{in}^i$ of the 100 most active links, and we systematically compare these for two resolution timescales: $\Delta t=1$ min (red diamonds) and $\Delta t=10$ min (blue dots). The online social interaction networks (EM, CM) have weak memory kernels which are maintained at the two resolution timescales. The offline social contact network (RM) have strong memory at both resolutions, but memory is overall larger at the lower resolution, thereby making timescale separation easy. All three engineered transportation networks (PB, PT, PU) display strong memory kernels only at the lower resolution timescale.} 
\label{fig:coord_scatter_2}
\end{figure}

\noindent Let us start by comparing, for each temporal network, these scatter plots for the two different resolution times. Results are plotted in the six panels of Fig.\ref{fig:coord_scatter_2}. The first observation is that transportation networks (bus (PB), train (PT) and underground (PU)) display very different memory properties at the $\Delta t=1$ and $\Delta t=10$ min resolution timescales, and hence their links systematically cluster apart. More particularly, for the $\Delta t=10$ the average co-orders are notably lower than for the $\Delta t=1$ scale, suggesting indeed that all the memory structure is captured at the $\Delta t=1$ scale, i.e. only one memory scale manifests, as expected
 as expected due to strong planning and scheduling restrictions.\\ 
 At the other extreme, many links in the two {\it online social communication} networks overlap in the scatterplots for the two resolution timescales, and memory is systematically weak. This effect is more acute in the college text message (CM) network than in the email (EM) network. Incidentally, for the CM network we find that there is a single link which, for both $\Delta t=1$ and $10$ min timescales, has a significantly larger $\left< \Omega \right>_{\text{out}}^i$ than the rest, meaning that there is one specific link whose activity is driven by the global activity of the network. The {\it offline social contact network} (RM) somehow interpolates the behaviour of the previous two groups: we can see that while links for the two timescales are closer together, they still cluster apart and the two different timescales are clearly visible. The memory of this network is strong at the two different timescales, concluding that we are indeed detecting two different memory timescales.\\

\begin{figure}[htb]
  \includegraphics[width = 0.35\textwidth]{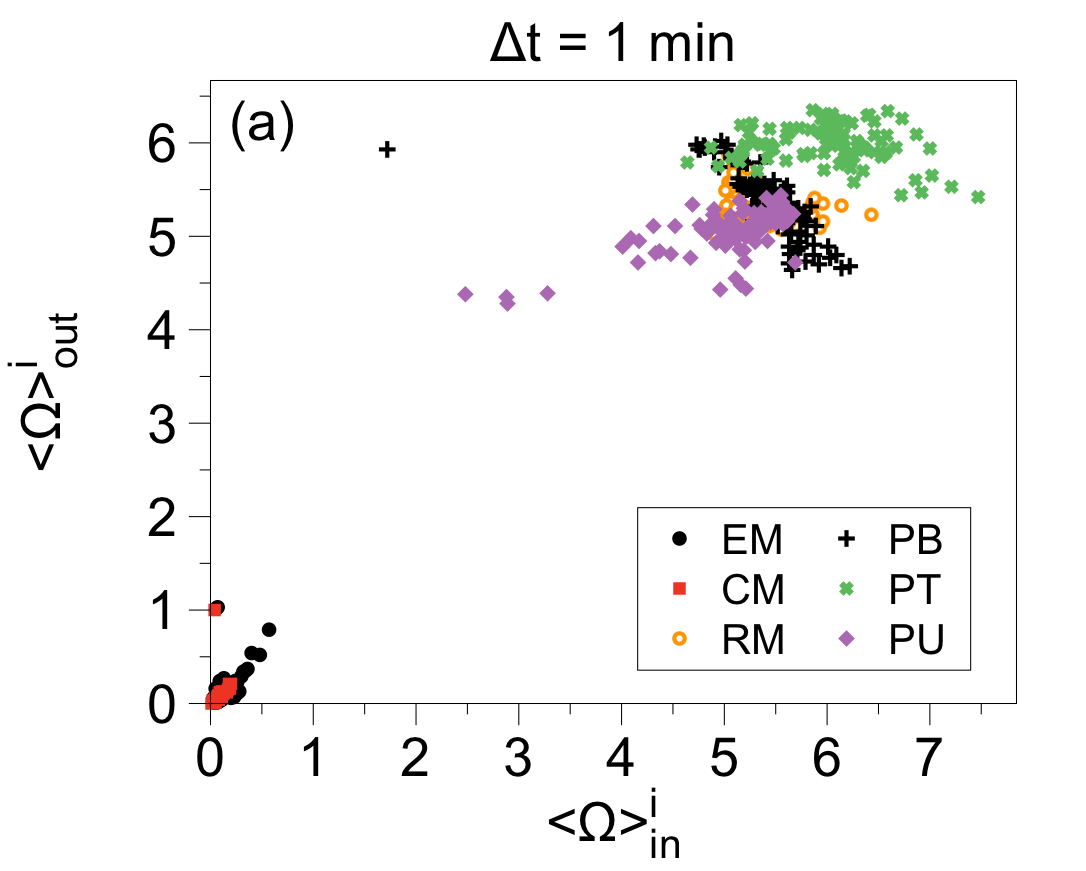}%
        \includegraphics[width = 0.35\textwidth]{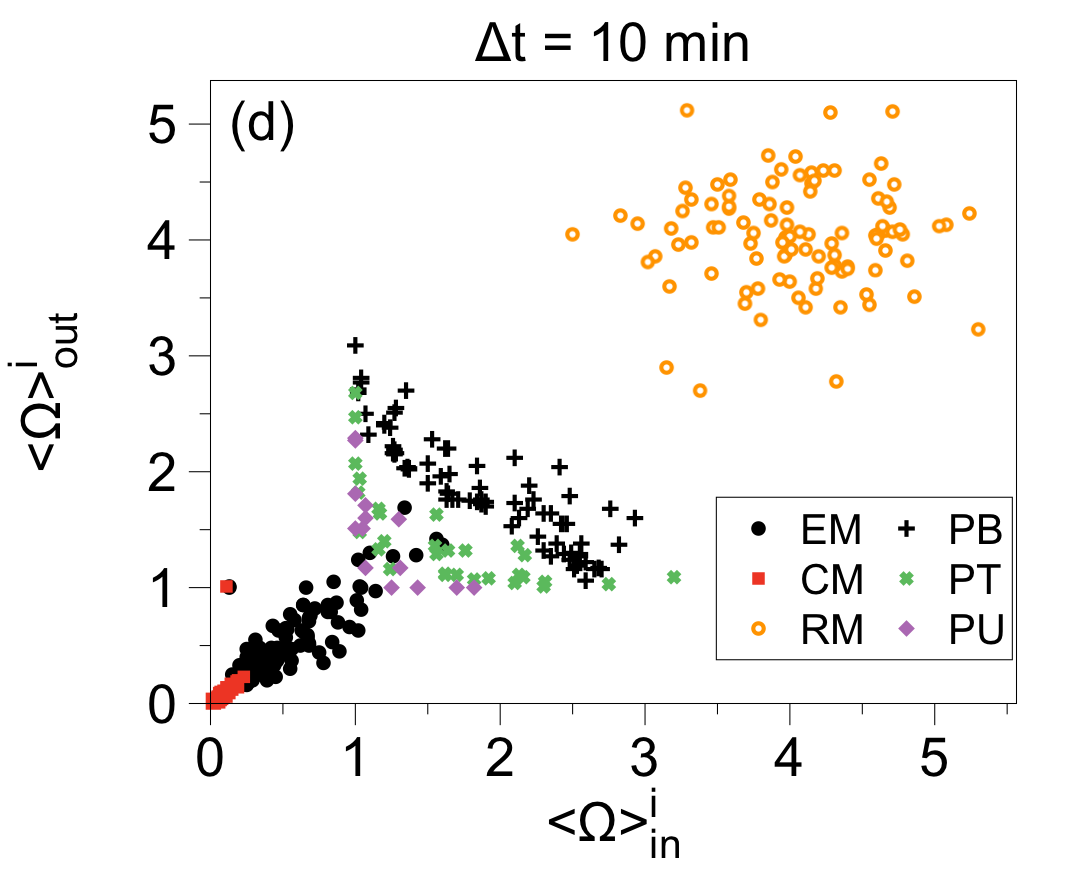}%
\caption{{ \bf Scatter plot of the average outgoing co-order  $\left< \Omega \right>_{out}^i$ vs the average incoming co-order $\left< \Omega \right>_{in}^i$ of the 100 most active links} for the six real-world temporal networks at the (a) $\Delta t = 1$ min  and (b) $\Delta t = 10$ min resolution timescales. At the lower resolution timescale online and offline temporal networks cluster together, and offline networks evidence stronger microscopic memory kernels. At the larger resolution timescale, only the (offline) social contact network (RM) displays consistently strong memory.}
\label{fig:coord_scatter}
\end{figure}

\noindent {Secondly}, we explore how these temporal networks compare to each other when projected in the $(\left< \Omega \right>_{\text{in}}^i, \left< \Omega \right>_{\text{out}}^i)$ plane.
In Fig.~\ref{fig:coord_scatter} we scatterplot the top 100 most active links for all six temporal networks at $\Delta t= 1$ min (panel (a)) and $\Delta t =$ 10 min (panel (d)). The first observation is that at the $\Delta t=1$ min, the two social communication networks cluster together, and the same is true for the three engineered transportation networks. Furthermore, there is a very clear separation between these two groups, indicating that the memory structure is very different: the online networks clearly display weaker memory than the offline ones. Perhaps unexpectedly, the social contact network (RM) clusters together with the engineered transportation networks (in particular, there is a large overlap with the bus network). Note however that, even if RM is a social network, it is (i) offline, like the second group, and also (ii) it is a mobility network, and therefore it is not unreasonable to find that its memory structure is more akin to the one displayed by a transportation one. When we switch the resolution scale to $\Delta t=10$ min, we observe that the engineered transportation networks are still grouped together, however they lose memory and therefore come closer to the social communication networks. As discussed before, the social contact network {retains its} high memory at the larger timescale, this being a manifestation of the second memory timescale present in this network.\\

\noindent We turn now to explore and quantify the extent to which links are clustered in the $(\left< \Omega \right>_{\text{in}}^i, \left< \Omega \right>_{\text{out}}^i)$ plane, by finding a 2-dimensional equivalent to their standard deviation $\sigma$ of the set of $L$
points representing each data set:
\begin{equation}
	\sigma =  \frac{1}{L-1} \sum_{i=1}^{L} \sqrt{\left( \left< \Omega \right>_{in}^i - \bar{\Omega}_{in}\right)^{2} +  \left( \left< \Omega \right>_{out}^i - \bar{\Omega}_{out}\right)^{2}},
\end{equation}
where $\bar{\Omega}_{in}$ and $\bar{\Omega}_{out}$ are the sample averages of 
$\left< \Omega \right>_{in}^i$ and $\left< \Omega \right>_{out}^i$ respectively.
If this value is large then there is a high average distance between points, and so
we can think of them as being less clustered. If this value is small then the points are
close together and we can think of them as being highly clustered. Results for this metric for each network are shown in table \ref{tab:cluster_metrics}. 
 This shows us that for $\Delta t = 1$ min the online social interaction networks (EM, CM) are far more clustered than the engineered networks, with the highest value of $\sigma$ for the social networks being 0.132, while the lowest value for engineered networks being 0.461, a fact which remains true for $\Delta t = 10$ min. The offline social contact network (RM) displays systematic spread, highlighting large memory link heterogeneity at different resolution times.\\
 
 \begin{table}[]
\begin{tabular}{|l|l|l|l|l|}
\hline
{\bf Network} & \multicolumn{2}{l|}{$\sigma$} & \multicolumn{2}{l|}{$D$} \\ \hline
$\Delta t$ & 1 & 10 & 1 & 10 \\ \hline
EM & 0.132 & 0.380 & 0.044 & 0.093 \\ \hline
CM & 0.063 & 0.069 & 0.017 & 0.019 \\ \hline
RM & 0.285 & 0.641 & 0.194 & 0.417 \\ \hline
PB & 0.408 & 0.674 & 0.421 & 0.623 \\ \hline
PT & 0.488 & 0.488 & 0.329 & 0.396 \\ \hline
PU & 0.461 & 0.415 & 0.229 & 0.342 \\ \hline
\end{tabular}
\caption{ Clustering metrics for each empirical data set and for both time resolutions $\Delta t = 1$ and 10 minutes.}
\label{tab:cluster_metrics}
\end{table}
 
\noindent Finally we consider the {\it input-output memory balance}, which measures how balanced the influence of the network is on each link
compared to each links influence on the network. For instance, if for all links we have 
$\left< \Omega \right>_{in}^i \approx \left< \Omega \right>_{out}^i$ then we know that
no link has a strong influence on the memory of the network, and equally no link has 
long memory of the rest of the network. To assess this we take the average distance from the line  $\left< \Omega \right>_{out}^i = \left< \Omega \right>_{in}^i$ for each network. More concretely, for each empirical network we find the average 
distance $D$ between the points in the scatter plot to the line $\left< \Omega \right>_{in}^i = \left< \Omega \right>_{out}^i$,
given $L$ sampled links, as follows:
\begin{equation}
	D = \frac{1}{L \sqrt{2}} \sum_{i=1}^{L} \left| \left< \Omega \right>_{in}^i -\left< \Omega \right>_{out}^i \right|.
\end{equation} 
The values of these metrics for each network are shown again in table \ref{tab:cluster_metrics}. 
Here we find that, as before, online social networks are far more balanced than transportation networks, {while} the offline social network {sits between the other two}.\\

Altogether, analysis of the co-memory matrix (and its projection in the co-memory histogram and the $(\left< \Omega \right>_{\text{in}}^i, \left< \Omega \right>_{\text{out}}^i)$ plane) at different resolutions provide insights on the microscopic memory heterogeneity displayed by real temporal networks, providing a simple method for discrimination between self-organised, online and offline engineered networks, and highlighting different and interpretable memory timescales. In particular, our results certify that while all networks considered display memory as estimated via $\Omega_{\text{eff}}(\mathcal{G})$, online social networks have a weaker and more homogeneous microscopic memory kernel than infrastructure transportation networks, and in between we find the case of the social contact network (RM) --a mix case which is `social', {like} the first group, but `offline' {like} the second--, and this latter unveils the presencee of two different memory timescales, probably related to two distinct mechanisms of social interaction between the university students: casual social interaction vs sharing a class together.

\section{Method implementations}
As part of this work we have provided a number of implementations
of our approach to finding the memory of a temporal network, or the distribution
of its co-orders, as based on the above mentioned efficient determination criterion (EDC)
given either a file containing a time stamped edge list or the input from a suitable 
function.
The intention here is to remove the complexity associated with implementing 
our method in an efficient way, and thus allow any future studies in this area 
to forge ahed without such an overhead. To this end we have provided versions 
in the following languages, each with their own advantages and disadvantages:
\begin{itemize}
\item C++
\item Java
\item Python 3.6
\item Python 2.7
\item Rust
\end{itemize}
The versions written in Python, and to some extent Java are intended to be used for testing of 
smaller data sets and prototyping of other experiments. This is because though the languages 
are common and their implementations are hopefully easy to understand, they
(again with the possible exception of Java) lack the raw speed and low memory overhead
of the other languages.\\
The C++ version is intended to be as fast as possible while 
maintaining the lowest possible memory overhead, and so is well suited 
to handling larger networks or the running of multiple experiments at once.\\
The Rust version (provided with both parallel and serial approaches) is
also intended to be as fast as possible, but does have a higher memory 
overhead when compared to the C++ version (though this is still the second
lowest memory usage), as such it is suited to larger networks, but may not 
be as good as C++ for running multiple experiments at once.
We have however provided a parallel implementation in Rust, which, provided there is no problem with 
memory requirements, is the fastest available. \\
The Java version is intended to be usably fast in any situation and with a moderate 
memory overhead, while being easy to work with.\\

\noindent We have tested the run-time of each implementation by finding the effective memory
of the edge list associated with the (EM) data set at a resolution of $\Delta t = 60 seconds$.
This comprises of 100 links over 390507 time steps.
These tests were run on a desktop PC (Ubuntu 18.04.3 LTS (64-bit)), with 
 a intel Core i7-6700k (4.00GHz, 8 core) processor, and 32GB of memory.
Each test was repeated 10 times and the results averaged. Runtimes are depicted in table \ref{table:runtime}.
\begin{table}[]
\begin{tabular}{|l|l|}
\hline
{\bf Version} & {\bf Average time (seconds)} \\ \hline
Rust (paralell implementation, rustc 1.39,  llvm 9.0, opt-level 2) & 14.843 \\ \hline
C++ (gcc9 -o2) & 59.453 \\ \hline
Java (OpenJDK 8) & 62.143 \\ \hline
Rust (serial implementation, rustc 1.39,  llvm 9.0, opt-level 3) & 62.216 \\ \hline
Python 3.7 (Parallel implementation) & 4206 \\ \hline
Python 2.7 & 7854 \\ \hline
Python 3.7 (Serial implementation) & 12691 \\ \hline
\end{tabular}
\caption{Runtimes for estimating the effective memory of the EM data set with $\Delta t= 60 seconds$ (see text for details).}
\label{table:runtime}
\end{table}
For access to the code and a more in depth description of 
its implementation see \url{github.com/oewilliams/temp-net-memory}.

\bibliographystyle{apsrev4-1.bst}

\begin{thebibliography}{67}%
\makeatletter
\providecommand \@ifxundefined [1]{%
 \@ifx{#1\undefined}
}%
\providecommand \@ifnum [1]{%
 \ifnum #1\expandafter \@firstoftwo
 \else \expandafter \@secondoftwo
 \fi
}%
\providecommand \@ifx [1]{%
 \ifx #1\expandafter \@firstoftwo
 \else \expandafter \@secondoftwo
 \fi
}%
\providecommand \natexlab [1]{#1}%
\providecommand \enquote  [1]{``#1''}%
\providecommand \bibnamefont  [1]{#1}%
\providecommand \bibfnamefont [1]{#1}%
\providecommand \citenamefont [1]{#1}%
\providecommand \href@noop [0]{\@secondoftwo}%
\providecommand \href [0]{\begingroup \@sanitize@url \@href}%
\providecommand \@href[1]{\@@startlink{#1}\@@href}%
\providecommand \@@href[1]{\endgroup#1\@@endlink}%
\providecommand \@sanitize@url [0]{\catcode `\\12\catcode `\$12\catcode
  `\&12\catcode `\#12\catcode `\^12\catcode `\_12\catcode `\%12\relax}%
\providecommand \@@startlink[1]{}%
\providecommand \@@endlink[0]{}%
\providecommand \url  [0]{\begingroup\@sanitize@url \@url }%
\providecommand \@url [1]{\endgroup\@href {#1}{\urlprefix }}%
\providecommand \urlprefix  [0]{URL }%
\providecommand \Eprint [0]{\href }%
\providecommand \doibase [0]{http://dx.doi.org/}%
\providecommand \selectlanguage [0]{\@gobble}%
\providecommand \bibinfo  [0]{\@secondoftwo}%
\providecommand \bibfield  [0]{\@secondoftwo}%
\providecommand \translation [1]{[#1]}%
\providecommand \BibitemOpen [0]{}%
\providecommand \bibitemStop [0]{}%
\providecommand \bibitemNoStop [0]{.\EOS\space}%
\providecommand \EOS [0]{\spacefactor3000\relax}%
\providecommand \BibitemShut  [1]{\csname bibitem#1\endcsname}%
\let\auto@bib@innerbib\@empty
\bibitem [{\citenamefont {Holme}\ and\ \citenamefont
  {Saram{\"a}ki}(2012)}]{Holme_rev12}%
  \BibitemOpen
  \bibfield  {author} {\bibinfo {author} {\bibfnamefont {P.}~\bibnamefont
  {Holme}}\ and\ \bibinfo {author} {\bibfnamefont {J.}~\bibnamefont
  {Saram{\"a}ki}},\ }\href {\doibase
  http://dx.doi.org/10.1016/j.physrep.2012.03.001} {\bibfield  {journal}
  {\bibinfo  {journal} {Physics Reports}\ }\textbf {\bibinfo {volume} {519}},\
  \bibinfo {pages} {97 } (\bibinfo {year} {2012})},\ \bibinfo {note} {temporal
  Networks}\BibitemShut {NoStop}%
\bibitem [{\citenamefont {Masuda}\ and\ \citenamefont
  {Lambiotte}(2016)}]{masuda_guide_temp_net}%
  \BibitemOpen
  \bibfield  {author} {\bibinfo {author} {\bibfnamefont {N.}~\bibnamefont
  {Masuda}}\ and\ \bibinfo {author} {\bibfnamefont {R.}~\bibnamefont
  {Lambiotte}},\ }\href@noop {} {\emph {\bibinfo {title} {A Guide to Temporal
  Networks}}}\ (\bibinfo  {publisher} {World Scientific (Europe)},\ \bibinfo
  {year} {2016})\BibitemShut {NoStop}%
\bibitem [{\citenamefont {Holme}\ and\ \citenamefont
  {Saram{\"a}ki}(2013)}]{holme2013temporal}%
  \BibitemOpen
  \bibfield  {author} {\bibinfo {author} {\bibfnamefont {P.}~\bibnamefont
  {Holme}}\ and\ \bibinfo {author} {\bibfnamefont {J.}~\bibnamefont
  {Saram{\"a}ki}},\ }\href@noop {} {\emph {\bibinfo {title} {Temporal
  networks}}}\ (\bibinfo  {publisher} {Springer},\ \bibinfo {year}
  {2013})\BibitemShut {NoStop}%
\bibitem [{\citenamefont {Holme}\ and\ \citenamefont
  {Saram{\"a}ki}(2019)}]{holme2019temporal}%
  \BibitemOpen
  \bibfield  {author} {\bibinfo {author} {\bibfnamefont {P.}~\bibnamefont
  {Holme}}\ and\ \bibinfo {author} {\bibfnamefont {J.}~\bibnamefont
  {Saram{\"a}ki}},\ }\href@noop {} {\emph {\bibinfo {title} {Temporal Network
  Theory}}}\ (\bibinfo  {publisher} {Springer},\ \bibinfo {year}
  {2019})\BibitemShut {NoStop}%
\bibitem [{\citenamefont {Starnini}\ \emph {et~al.}(2013)\citenamefont
  {Starnini}, \citenamefont {Baronchelli},\ and\ \citenamefont
  {Pastor-Satorras}}]{Starnini:2013}%
  \BibitemOpen
  \bibfield  {author} {\bibinfo {author} {\bibfnamefont {M.}~\bibnamefont
  {Starnini}}, \bibinfo {author} {\bibfnamefont {A.}~\bibnamefont
  {Baronchelli}}, \ and\ \bibinfo {author} {\bibfnamefont {R.}~\bibnamefont
  {Pastor-Satorras}},\ }\href {\doibase 10.1103/PhysRevLett.110.168701}
  {\bibfield  {journal} {\bibinfo  {journal} {Phys. Rev. Lett.}\ }\textbf
  {\bibinfo {volume} {110}},\ \bibinfo {pages} {168701} (\bibinfo {year}
  {2013})}\BibitemShut {NoStop}%
\bibitem [{\citenamefont {Szell}\ \emph {et~al.}(2012)\citenamefont {Szell},
  \citenamefont {Sinatra}, \citenamefont {Petri}, \citenamefont {Thurner},\
  and\ \citenamefont {Latora}}]{Szell:2012aa}%
  \BibitemOpen
  \bibfield  {author} {\bibinfo {author} {\bibfnamefont {M.}~\bibnamefont
  {Szell}}, \bibinfo {author} {\bibfnamefont {R.}~\bibnamefont {Sinatra}},
  \bibinfo {author} {\bibfnamefont {G.}~\bibnamefont {Petri}}, \bibinfo
  {author} {\bibfnamefont {S.}~\bibnamefont {Thurner}}, \ and\ \bibinfo
  {author} {\bibfnamefont {V.}~\bibnamefont {Latora}},\ }\href
  {http://dx.doi.org/10.1038/srep00457} {\bibfield  {journal} {\bibinfo
  {journal} {Scientific Reports}\ }\textbf {\bibinfo {volume} {2}},\ \bibinfo
  {pages} {457 EP } (\bibinfo {year} {2012})}\BibitemShut {NoStop}%
\bibitem [{\citenamefont {Yoneki}\ \emph {et~al.}(2009)\citenamefont {Yoneki},
  \citenamefont {Greenfield},\ and\ \citenamefont {Crowcroft}}]{Yoneki:2009}%
  \BibitemOpen
  \bibfield  {author} {\bibinfo {author} {\bibfnamefont {E.}~\bibnamefont
  {Yoneki}}, \bibinfo {author} {\bibfnamefont {D.}~\bibnamefont {Greenfield}},
  \ and\ \bibinfo {author} {\bibfnamefont {J.}~\bibnamefont {Crowcroft}},\ }in\
  \href {\doibase 10.1109/ASONAM.2009.42} {\emph {\bibinfo {booktitle} {2009
  International Conference on Advances in Social Network Analysis and
  Mining}}}\ (\bibinfo {year} {2009})\ pp.\ \bibinfo {pages}
  {356--361}\BibitemShut {NoStop}%
\bibitem [{\citenamefont {Corsi}\ \emph {et~al.}(2018)\citenamefont {Corsi},
  \citenamefont {Lillo}, \citenamefont {Pirino},\ and\ \citenamefont
  {Trapin}}]{corsi2018measuring}%
  \BibitemOpen
  \bibfield  {author} {\bibinfo {author} {\bibfnamefont {F.}~\bibnamefont
  {Corsi}}, \bibinfo {author} {\bibfnamefont {F.}~\bibnamefont {Lillo}},
  \bibinfo {author} {\bibfnamefont {D.}~\bibnamefont {Pirino}}, \ and\ \bibinfo
  {author} {\bibfnamefont {L.}~\bibnamefont {Trapin}},\ }\href@noop {}
  {\bibfield  {journal} {\bibinfo  {journal} {Journal of Financial Stability}\
  }\textbf {\bibinfo {volume} {38}},\ \bibinfo {pages} {18} (\bibinfo {year}
  {2018})}\BibitemShut {NoStop}%
\bibitem [{\citenamefont {Mazzarisi}\ \emph {et~al.}(2019)\citenamefont
  {Mazzarisi}, \citenamefont {Barucca}, \citenamefont {Lillo},\ and\
  \citenamefont {Tantari}}]{mazzarisi2019dynamic}%
  \BibitemOpen
  \bibfield  {author} {\bibinfo {author} {\bibfnamefont {P.}~\bibnamefont
  {Mazzarisi}}, \bibinfo {author} {\bibfnamefont {P.}~\bibnamefont {Barucca}},
  \bibinfo {author} {\bibfnamefont {F.}~\bibnamefont {Lillo}}, \ and\ \bibinfo
  {author} {\bibfnamefont {D.}~\bibnamefont {Tantari}},\ }\href@noop {}
  {\bibfield  {journal} {\bibinfo  {journal} {European Journal of Operational
  Research}\ } (\bibinfo {year} {2019})}\BibitemShut {NoStop}%
\bibitem [{\citenamefont {Mill{\'a}n}\ \emph {et~al.}(2018)\citenamefont
  {Mill{\'a}n}, \citenamefont {Torres}, \citenamefont {Johnson},\ and\
  \citenamefont {Marro}}]{millan2018concurrence}%
  \BibitemOpen
  \bibfield  {author} {\bibinfo {author} {\bibfnamefont {A.~P.}\ \bibnamefont
  {Mill{\'a}n}}, \bibinfo {author} {\bibfnamefont {J.}~\bibnamefont {Torres}},
  \bibinfo {author} {\bibfnamefont {S.}~\bibnamefont {Johnson}}, \ and\
  \bibinfo {author} {\bibfnamefont {J.}~\bibnamefont {Marro}},\ }\href@noop {}
  {\bibfield  {journal} {\bibinfo  {journal} {Nature communications}\ }\textbf
  {\bibinfo {volume} {9}},\ \bibinfo {pages} {2236} (\bibinfo {year}
  {2018})}\BibitemShut {NoStop}%
\bibitem [{\citenamefont {Valencia}\ \emph {et~al.}(2008)\citenamefont
  {Valencia}, \citenamefont {Martinerie}, \citenamefont {Dupont},\ and\
  \citenamefont {Chavez}}]{Valencia08}%
  \BibitemOpen
  \bibfield  {author} {\bibinfo {author} {\bibfnamefont {M.}~\bibnamefont
  {Valencia}}, \bibinfo {author} {\bibfnamefont {J.}~\bibnamefont
  {Martinerie}}, \bibinfo {author} {\bibfnamefont {S.}~\bibnamefont {Dupont}},
  \ and\ \bibinfo {author} {\bibfnamefont {M.}~\bibnamefont {Chavez}},\ }\href
  {\doibase 10.1103/PhysRevE.77.050905} {\bibfield  {journal} {\bibinfo
  {journal} {Phys. Rev. E}\ }\textbf {\bibinfo {volume} {77}},\ \bibinfo
  {pages} {050905} (\bibinfo {year} {2008})}\BibitemShut {NoStop}%
\bibitem [{\citenamefont {Zanin}\ \emph {et~al.}(2009)\citenamefont {Zanin},
  \citenamefont {Lacasa},\ and\ \citenamefont {Cea}}]{zanin2009dynamics}%
  \BibitemOpen
  \bibfield  {author} {\bibinfo {author} {\bibfnamefont {M.}~\bibnamefont
  {Zanin}}, \bibinfo {author} {\bibfnamefont {L.}~\bibnamefont {Lacasa}}, \
  and\ \bibinfo {author} {\bibfnamefont {M.}~\bibnamefont {Cea}},\ }\href@noop
  {} {\bibfield  {journal} {\bibinfo  {journal} {Chaos: An Interdisciplinary
  Journal of Nonlinear Science}\ }\textbf {\bibinfo {volume} {19}},\ \bibinfo
  {pages} {023111} (\bibinfo {year} {2009})}\BibitemShut {NoStop}%
\bibitem [{\citenamefont {Tang}\ \emph {et~al.}(2010)\citenamefont {Tang},
  \citenamefont {Scellato}, \citenamefont {Musolesi}, \citenamefont {Mascolo},\
  and\ \citenamefont {Latora}}]{tang2010sw}%
  \BibitemOpen
  \bibfield  {author} {\bibinfo {author} {\bibfnamefont {J.}~\bibnamefont
  {Tang}}, \bibinfo {author} {\bibfnamefont {S.}~\bibnamefont {Scellato}},
  \bibinfo {author} {\bibfnamefont {M.}~\bibnamefont {Musolesi}}, \bibinfo
  {author} {\bibfnamefont {C.}~\bibnamefont {Mascolo}}, \ and\ \bibinfo
  {author} {\bibfnamefont {V.}~\bibnamefont {Latora}},\ }\href@noop {}
  {\bibfield  {journal} {\bibinfo  {journal} {Phys. Rev. E}\ }\textbf {\bibinfo
  {volume} {81}},\ \bibinfo {pages} {055101} (\bibinfo {year}
  {2010})}\BibitemShut {NoStop}%
\bibitem [{\citenamefont {Lambiotte}\ \emph {et~al.}(2019)\citenamefont
  {Lambiotte}, \citenamefont {Rosvall},\ and\ \citenamefont
  {Scholtes}}]{lambiotte2019networks}%
  \BibitemOpen
  \bibfield  {author} {\bibinfo {author} {\bibfnamefont {R.}~\bibnamefont
  {Lambiotte}}, \bibinfo {author} {\bibfnamefont {M.}~\bibnamefont {Rosvall}},
  \ and\ \bibinfo {author} {\bibfnamefont {I.}~\bibnamefont {Scholtes}},\
  }\href@noop {} {\bibfield  {journal} {\bibinfo  {journal} {Nature physics}\
  ,\ \bibinfo {pages} {1}} (\bibinfo {year} {2019})}\BibitemShut {NoStop}%
\bibitem [{\citenamefont {Delvenne}\ \emph {et~al.}(2015)\citenamefont
  {Delvenne}, \citenamefont {Lambiotte},\ and\ \citenamefont
  {Rocha}}]{delvenne2015diffusion}%
  \BibitemOpen
  \bibfield  {author} {\bibinfo {author} {\bibfnamefont {J.-C.}\ \bibnamefont
  {Delvenne}}, \bibinfo {author} {\bibfnamefont {R.}~\bibnamefont {Lambiotte}},
  \ and\ \bibinfo {author} {\bibfnamefont {L.~E.}\ \bibnamefont {Rocha}},\
  }\href@noop {} {\bibfield  {journal} {\bibinfo  {journal} {Nature
  communications}\ }\textbf {\bibinfo {volume} {6}},\ \bibinfo {pages} {7366}
  (\bibinfo {year} {2015})}\BibitemShut {NoStop}%
\bibitem [{\citenamefont {Lambiotte}\ \emph {et~al.}(2015)\citenamefont
  {Lambiotte}, \citenamefont {Salnikov},\ and\ \citenamefont
  {Rosvall}}]{Lambiotte_jcn15}%
  \BibitemOpen
  \bibfield  {author} {\bibinfo {author} {\bibfnamefont {R.}~\bibnamefont
  {Lambiotte}}, \bibinfo {author} {\bibfnamefont {V.}~\bibnamefont {Salnikov}},
  \ and\ \bibinfo {author} {\bibfnamefont {M.}~\bibnamefont {Rosvall}},\ }\href
  {\doibase 10.1093/comnet/cnu017} {\bibfield  {journal} {\bibinfo  {journal}
  {Journal of Complex Networks}\ }\textbf {\bibinfo {volume} {3}},\ \bibinfo
  {pages} {177} (\bibinfo {year} {2015})}\BibitemShut {NoStop}%
\bibitem [{\citenamefont {Masuda}\ \emph {et~al.}(2013)\citenamefont {Masuda},
  \citenamefont {Klemm},\ and\ \citenamefont
  {Egu{\'\i}luz}}]{masuda2013temporal}%
  \BibitemOpen
  \bibfield  {author} {\bibinfo {author} {\bibfnamefont {N.}~\bibnamefont
  {Masuda}}, \bibinfo {author} {\bibfnamefont {K.}~\bibnamefont {Klemm}}, \
  and\ \bibinfo {author} {\bibfnamefont {V.~M.}\ \bibnamefont {Egu{\'\i}luz}},\
  }\href@noop {} {\bibfield  {journal} {\bibinfo  {journal} {Physical Review
  Letters}\ }\textbf {\bibinfo {volume} {111}},\ \bibinfo {pages} {188701}
  (\bibinfo {year} {2013})}\BibitemShut {NoStop}%
\bibitem [{\citenamefont {Scholtes}\ \emph {et~al.}(2014)\citenamefont
  {Scholtes}, \citenamefont {Wider}, \citenamefont {Pfitzner}, \citenamefont
  {Garas}, \citenamefont {Tessone},\ and\ \citenamefont
  {Schweitzer}}]{Scholtes_natcomm14}%
  \BibitemOpen
  \bibfield  {author} {\bibinfo {author} {\bibfnamefont {I.}~\bibnamefont
  {Scholtes}}, \bibinfo {author} {\bibfnamefont {N.}~\bibnamefont {Wider}},
  \bibinfo {author} {\bibfnamefont {R.}~\bibnamefont {Pfitzner}}, \bibinfo
  {author} {\bibfnamefont {A.}~\bibnamefont {Garas}}, \bibinfo {author}
  {\bibfnamefont {C.~J.}\ \bibnamefont {Tessone}}, \ and\ \bibinfo {author}
  {\bibfnamefont {F.}~\bibnamefont {Schweitzer}},\ }\href
  {http://dx.doi.org/10.1038/ncomms6024} {\bibfield  {journal} {\bibinfo
  {journal} {Nat. Commun.}\ }\textbf {\bibinfo {volume} {5}},\ \bibinfo {pages}
  {5024} (\bibinfo {year} {2014})}\BibitemShut {NoStop}%
\bibitem [{\citenamefont {Hiraoka}\ and\ \citenamefont
  {Jo}(2018)}]{Hiraoka:2018aa}%
  \BibitemOpen
  \bibfield  {author} {\bibinfo {author} {\bibfnamefont {T.}~\bibnamefont
  {Hiraoka}}\ and\ \bibinfo {author} {\bibfnamefont {H.-H.}\ \bibnamefont
  {Jo}},\ }\href {\doibase 10.1038/s41598-018-33700-8} {\bibfield  {journal}
  {\bibinfo  {journal} {Scientific Reports}\ }\textbf {\bibinfo {volume} {8}},\
  \bibinfo {pages} {15321} (\bibinfo {year} {2018})}\BibitemShut {NoStop}%
\bibitem [{\citenamefont {Takaguchi}\ \emph {et~al.}(2013)\citenamefont
  {Takaguchi}, \citenamefont {Masuda},\ and\ \citenamefont
  {Holme}}]{takaguchi2013bursty}%
  \BibitemOpen
  \bibfield  {author} {\bibinfo {author} {\bibfnamefont {T.}~\bibnamefont
  {Takaguchi}}, \bibinfo {author} {\bibfnamefont {N.}~\bibnamefont {Masuda}}, \
  and\ \bibinfo {author} {\bibfnamefont {P.}~\bibnamefont {Holme}},\
  }\href@noop {} {\bibfield  {journal} {\bibinfo  {journal} {PloS one}\
  }\textbf {\bibinfo {volume} {8}},\ \bibinfo {pages} {e68629} (\bibinfo {year}
  {2013})}\BibitemShut {NoStop}%
\bibitem [{\citenamefont {Lambiotte}\ \emph {et~al.}(2013)\citenamefont
  {Lambiotte}, \citenamefont {Tabourier},\ and\ \citenamefont
  {Delvenne}}]{lambiotte2013burstiness}%
  \BibitemOpen
  \bibfield  {author} {\bibinfo {author} {\bibfnamefont {R.}~\bibnamefont
  {Lambiotte}}, \bibinfo {author} {\bibfnamefont {L.}~\bibnamefont
  {Tabourier}}, \ and\ \bibinfo {author} {\bibfnamefont {J.-C.}\ \bibnamefont
  {Delvenne}},\ }\href@noop {} {\bibfield  {journal} {\bibinfo  {journal} {The
  European Physical Journal B}\ }\textbf {\bibinfo {volume} {86}},\ \bibinfo
  {pages} {320} (\bibinfo {year} {2013})}\BibitemShut {NoStop}%
\bibitem [{\citenamefont {Karsai}\ \emph {et~al.}(2011)\citenamefont {Karsai},
  \citenamefont {Kivel{\"a}}, \citenamefont {Pan}, \citenamefont {Kaski},
  \citenamefont {Kert{\'e}sz}, \citenamefont {Barab{\'a}si},\ and\
  \citenamefont {Saram{\"a}ki}}]{karsai2011small}%
  \BibitemOpen
  \bibfield  {author} {\bibinfo {author} {\bibfnamefont {M.}~\bibnamefont
  {Karsai}}, \bibinfo {author} {\bibfnamefont {M.}~\bibnamefont {Kivel{\"a}}},
  \bibinfo {author} {\bibfnamefont {R.~K.}\ \bibnamefont {Pan}}, \bibinfo
  {author} {\bibfnamefont {K.}~\bibnamefont {Kaski}}, \bibinfo {author}
  {\bibfnamefont {J.}~\bibnamefont {Kert{\'e}sz}}, \bibinfo {author}
  {\bibfnamefont {A.-L.}\ \bibnamefont {Barab{\'a}si}}, \ and\ \bibinfo
  {author} {\bibfnamefont {J.}~\bibnamefont {Saram{\"a}ki}},\ }\href@noop {}
  {\bibfield  {journal} {\bibinfo  {journal} {Physical Review E}\ }\textbf
  {\bibinfo {volume} {83}},\ \bibinfo {pages} {025102} (\bibinfo {year}
  {2011})}\BibitemShut {NoStop}%
\bibitem [{\citenamefont {Williams}\ \emph
  {et~al.}(2019{\natexlab{a}})\citenamefont {Williams}, \citenamefont {Lillo},\
  and\ \citenamefont {Latora}}]{Williams_2019}%
  \BibitemOpen
  \bibfield  {author} {\bibinfo {author} {\bibfnamefont {O.~E.}\ \bibnamefont
  {Williams}}, \bibinfo {author} {\bibfnamefont {F.}~\bibnamefont {Lillo}}, \
  and\ \bibinfo {author} {\bibfnamefont {V.}~\bibnamefont {Latora}},\ }\href
  {\doibase 10.1088/1367-2630/ab13fb} {\bibfield  {journal} {\bibinfo
  {journal} {New Journal of Physics}\ }\textbf {\bibinfo {volume} {21}},\
  \bibinfo {pages} {043028} (\bibinfo {year} {2019}{\natexlab{a}})}\BibitemShut
  {NoStop}%
\bibitem [{\citenamefont {Van~Mieghem}\ and\ \citenamefont {Van~de
  Bovenkamp}(2013)}]{van2013non}%
  \BibitemOpen
  \bibfield  {author} {\bibinfo {author} {\bibfnamefont {P.}~\bibnamefont
  {Van~Mieghem}}\ and\ \bibinfo {author} {\bibfnamefont {R.}~\bibnamefont
  {Van~de Bovenkamp}},\ }\href@noop {} {\bibfield  {journal} {\bibinfo
  {journal} {Physical review letters}\ }\textbf {\bibinfo {volume} {110}},\
  \bibinfo {pages} {108701} (\bibinfo {year} {2013})}\BibitemShut {NoStop}%
\bibitem [{\citenamefont {Fallani}\ \emph {et~al.}(2008)\citenamefont
  {Fallani}, \citenamefont {Latora}, \citenamefont {Astolfi}, \citenamefont
  {Cincotti}, \citenamefont {Mattia}, \citenamefont {Marciani}, \citenamefont
  {Salinari}, \citenamefont {Colosimo},\ and\ \citenamefont
  {Babiloni}}]{Fallani08}%
  \BibitemOpen
  \bibfield  {author} {\bibinfo {author} {\bibfnamefont {F.~D.~V.}\
  \bibnamefont {Fallani}}, \bibinfo {author} {\bibfnamefont {V.}~\bibnamefont
  {Latora}}, \bibinfo {author} {\bibfnamefont {L.}~\bibnamefont {Astolfi}},
  \bibinfo {author} {\bibfnamefont {F.}~\bibnamefont {Cincotti}}, \bibinfo
  {author} {\bibfnamefont {D.}~\bibnamefont {Mattia}}, \bibinfo {author}
  {\bibfnamefont {M.~G.}\ \bibnamefont {Marciani}}, \bibinfo {author}
  {\bibfnamefont {S.}~\bibnamefont {Salinari}}, \bibinfo {author}
  {\bibfnamefont {A.}~\bibnamefont {Colosimo}}, \ and\ \bibinfo {author}
  {\bibfnamefont {F.}~\bibnamefont {Babiloni}},\ }\href
  {http://stacks.iop.org/1751-8121/41/i=22/a=224014} {\bibfield  {journal}
  {\bibinfo  {journal} {Journal of Physics A: Mathematical and Theoretical}\
  }\textbf {\bibinfo {volume} {41}},\ \bibinfo {pages} {224014} (\bibinfo
  {year} {2008})}\BibitemShut {NoStop}%
\bibitem [{\citenamefont {Singer}\ \emph {et~al.}(2014)\citenamefont {Singer},
  \citenamefont {Helic}, \citenamefont {Taraghi},\ and\ \citenamefont
  {Strohmaier}}]{Singer_PLOSONE14}%
  \BibitemOpen
  \bibfield  {author} {\bibinfo {author} {\bibfnamefont {P.}~\bibnamefont
  {Singer}}, \bibinfo {author} {\bibfnamefont {D.}~\bibnamefont {Helic}},
  \bibinfo {author} {\bibfnamefont {B.}~\bibnamefont {Taraghi}}, \ and\
  \bibinfo {author} {\bibfnamefont {M.}~\bibnamefont {Strohmaier}},\ }\href
  {\doibase 10.1371/journal.pone.0102070} {\bibfield  {journal} {\bibinfo
  {journal} {PLOS ONE}\ }\textbf {\bibinfo {volume} {9}},\ \bibinfo {pages} {1}
  (\bibinfo {year} {2014})}\BibitemShut {NoStop}%
\bibitem [{\citenamefont {Peixoto}\ and\ \citenamefont
  {Rosvall}(2017)}]{peixoto2017modelling}%
  \BibitemOpen
  \bibfield  {author} {\bibinfo {author} {\bibfnamefont {T.~P.}\ \bibnamefont
  {Peixoto}}\ and\ \bibinfo {author} {\bibfnamefont {M.}~\bibnamefont
  {Rosvall}},\ }\href@noop {} {\bibfield  {journal} {\bibinfo  {journal}
  {Nature communications}\ }\textbf {\bibinfo {volume} {8}},\ \bibinfo {pages}
  {582} (\bibinfo {year} {2017})}\BibitemShut {NoStop}%
\bibitem [{\citenamefont {Rosvall}\ \emph {et~al.}(2014)\citenamefont
  {Rosvall}, \citenamefont {Esquivel}, \citenamefont {Lancichinetti},
  \citenamefont {West},\ and\ \citenamefont {Lambiotte}}]{Rosvall_natcomm14}%
  \BibitemOpen
  \bibfield  {author} {\bibinfo {author} {\bibfnamefont {M.}~\bibnamefont
  {Rosvall}}, \bibinfo {author} {\bibfnamefont {A.~V.}\ \bibnamefont
  {Esquivel}}, \bibinfo {author} {\bibfnamefont {A.}~\bibnamefont
  {Lancichinetti}}, \bibinfo {author} {\bibfnamefont {J.~D.}\ \bibnamefont
  {West}}, \ and\ \bibinfo {author} {\bibfnamefont {R.}~\bibnamefont
  {Lambiotte}},\ }\href {http://dx.doi.org/10.1038/ncomms5630} {\bibfield
  {journal} {\bibinfo  {journal} {Nat. Commun.}\ }\textbf {\bibinfo {volume}
  {5}},\ \bibinfo {pages} {4630} (\bibinfo {year} {2014})}\BibitemShut
  {NoStop}%
\bibitem [{\citenamefont {Scholtes}(2017)}]{scholtes2017network}%
  \BibitemOpen
  \bibfield  {author} {\bibinfo {author} {\bibfnamefont {I.}~\bibnamefont
  {Scholtes}},\ }in\ \href@noop {} {\emph {\bibinfo {booktitle} {Proceedings of
  the 23rd ACM SIGKDD International Conference on Knowledge Discovery and Data
  Mining}}}\ (\bibinfo {organization} {ACM},\ \bibinfo {year} {2017})\ pp.\
  \bibinfo {pages} {1037--1046}\BibitemShut {NoStop}%
\bibitem [{\citenamefont {Peixoto}\ and\ \citenamefont
  {Gauvin}(2018)}]{peixoto2018change}%
  \BibitemOpen
  \bibfield  {author} {\bibinfo {author} {\bibfnamefont {T.~P.}\ \bibnamefont
  {Peixoto}}\ and\ \bibinfo {author} {\bibfnamefont {L.}~\bibnamefont
  {Gauvin}},\ }\href@noop {} {\bibfield  {journal} {\bibinfo  {journal}
  {Scientific reports}\ }\textbf {\bibinfo {volume} {8}},\ \bibinfo {pages}
  {15511} (\bibinfo {year} {2018})}\BibitemShut {NoStop}%
\bibitem [{\citenamefont {Zhao}\ \emph {et~al.}(2001)\citenamefont {Zhao},
  \citenamefont {Dorea},\ and\ \citenamefont
  {Gon{\c{c}}alves}}]{zhao2001determination}%
  \BibitemOpen
  \bibfield  {author} {\bibinfo {author} {\bibfnamefont {L.}~\bibnamefont
  {Zhao}}, \bibinfo {author} {\bibfnamefont {C.}~\bibnamefont {Dorea}}, \ and\
  \bibinfo {author} {\bibfnamefont {C.}~\bibnamefont {Gon{\c{c}}alves}},\
  }\href@noop {} {\bibfield  {journal} {\bibinfo  {journal} {Statistical
  inference for stochastic processes}\ }\textbf {\bibinfo {volume} {4}},\
  \bibinfo {pages} {273} (\bibinfo {year} {2001})}\BibitemShut {NoStop}%
\bibitem [{\citenamefont {Dorea}\ \emph {et~al.}(2014)\citenamefont {Dorea},
  \citenamefont {Goncalves},\ and\ \citenamefont
  {Resende}}]{dorea2014simulation}%
  \BibitemOpen
  \bibfield  {author} {\bibinfo {author} {\bibfnamefont {C.~C.}\ \bibnamefont
  {Dorea}}, \bibinfo {author} {\bibfnamefont {C.~R.}\ \bibnamefont
  {Goncalves}}, \ and\ \bibinfo {author} {\bibfnamefont {P.}~\bibnamefont
  {Resende}},\ }in\ \href@noop {} {\emph {\bibinfo {booktitle} {Proc. World
  Congress on Engineering and Computer Science}}},\ Vol.~\bibinfo {volume} {2}\
  (\bibinfo {year} {2014})\ pp.\ \bibinfo {pages} {899--901}\BibitemShut
  {NoStop}%
\bibitem [{\citenamefont {Michalski}\ \emph {et~al.}(2011)\citenamefont
  {Michalski}, \citenamefont {Palus},\ and\ \citenamefont
  {Kazienko}}]{email_bb}%
  \BibitemOpen
  \bibfield  {author} {\bibinfo {author} {\bibfnamefont {R.}~\bibnamefont
  {Michalski}}, \bibinfo {author} {\bibfnamefont {S.}~\bibnamefont {Palus}}, \
  and\ \bibinfo {author} {\bibfnamefont {P.}~\bibnamefont {Kazienko}},\ }in\
  \href@noop {} {\emph {\bibinfo {booktitle} {Lecture Notes in Business
  Information Processing}}},\ Vol.~\bibinfo {volume} {87}\ (\bibinfo
  {publisher} {Springer Berlin Heidelberg},\ \bibinfo {year} {2011})\ pp.\
  \bibinfo {pages} {197--206}\BibitemShut {NoStop}%
\bibitem [{\citenamefont {Panzarasa}\ \emph {et~al.}(2009)\citenamefont
  {Panzarasa}, \citenamefont {Opsahl},\ and\ \citenamefont
  {Carley}}]{panzarasa2009patterns}%
  \BibitemOpen
  \bibfield  {author} {\bibinfo {author} {\bibfnamefont {P.}~\bibnamefont
  {Panzarasa}}, \bibinfo {author} {\bibfnamefont {T.}~\bibnamefont {Opsahl}}, \
  and\ \bibinfo {author} {\bibfnamefont {K.~M.}\ \bibnamefont {Carley}},\
  }\href@noop {} {\bibfield  {journal} {\bibinfo  {journal} {Journal of the
  American Society for Information Science and Technology}\ }\textbf {\bibinfo
  {volume} {60}},\ \bibinfo {pages} {911} (\bibinfo {year} {2009})}\BibitemShut
  {NoStop}%
\bibitem [{\citenamefont {Eagle}\ and\ \citenamefont
  {Pentland}(2006)}]{eagle2006reality}%
  \BibitemOpen
  \bibfield  {author} {\bibinfo {author} {\bibfnamefont {N.}~\bibnamefont
  {Eagle}}\ and\ \bibinfo {author} {\bibfnamefont {A.~S.}\ \bibnamefont
  {Pentland}},\ }\href@noop {} {\bibfield  {journal} {\bibinfo  {journal}
  {Personal and ubiquitous computing}\ }\textbf {\bibinfo {volume} {10}},\
  \bibinfo {pages} {255} (\bibinfo {year} {2006})}\BibitemShut {NoStop}%
\bibitem [{\citenamefont {Kujala}\ \emph {et~al.}(2018)\citenamefont {Kujala},
  \citenamefont {Weckstr{\"o}m}, \citenamefont {Darst}, \citenamefont
  {Mladenovi{\'c}},\ and\ \citenamefont {Saram{\"a}ki}}]{kujala2018collection}%
  \BibitemOpen
  \bibfield  {author} {\bibinfo {author} {\bibfnamefont {R.}~\bibnamefont
  {Kujala}}, \bibinfo {author} {\bibfnamefont {C.}~\bibnamefont
  {Weckstr{\"o}m}}, \bibinfo {author} {\bibfnamefont {R.~K.}\ \bibnamefont
  {Darst}}, \bibinfo {author} {\bibfnamefont {M.~N.}\ \bibnamefont
  {Mladenovi{\'c}}}, \ and\ \bibinfo {author} {\bibfnamefont {J.}~\bibnamefont
  {Saram{\"a}ki}},\ }\href@noop {} {\bibfield  {journal} {\bibinfo  {journal}
  {Scientific data}\ }\textbf {\bibinfo {volume} {5}},\ \bibinfo {pages}
  {180089} (\bibinfo {year} {2018})}\BibitemShut {NoStop}%
\bibitem [{\citenamefont {Williams}\ \emph
  {et~al.}(2019{\natexlab{b}})\citenamefont {Williams}, \citenamefont {Lillo},\
  and\ \citenamefont {Latora}}]{williams2019_diff}%
  \BibitemOpen
  \bibfield  {author} {\bibinfo {author} {\bibfnamefont {O.~E.}\ \bibnamefont
  {Williams}}, \bibinfo {author} {\bibfnamefont {F.}~\bibnamefont {Lillo}}, \
  and\ \bibinfo {author} {\bibfnamefont {V.}~\bibnamefont {Latora}},\
  }\href@noop {} {\  (\bibinfo {year} {2019}{\natexlab{b}})},\ \Eprint
  {http://arxiv.org/abs/1909.08134} {arXiv:1909.08134 [cond-mat.stat-mech]}
  \BibitemShut {NoStop}%
\bibitem [{\citenamefont {Gagniuc}(2017)}]{gagniuc2017markov}%
  \BibitemOpen
  \bibfield  {author} {\bibinfo {author} {\bibfnamefont {P.~A.}\ \bibnamefont
  {Gagniuc}},\ }\href@noop {} {\emph {\bibinfo {title} {Markov chains: from
  theory to implementation and experimentation}}}\ (\bibinfo  {publisher} {John
  Wiley \& Sons},\ \bibinfo {year} {2017})\BibitemShut {NoStop}%
\bibitem [{\citenamefont {Tong}(1975)}]{tong1975determination}%
  \BibitemOpen
  \bibfield  {author} {\bibinfo {author} {\bibfnamefont {H.}~\bibnamefont
  {Tong}},\ }\href@noop {} {\bibfield  {journal} {\bibinfo  {journal} {Journal
  of applied probability}\ }\textbf {\bibinfo {volume} {12}},\ \bibinfo {pages}
  {488} (\bibinfo {year} {1975})}\BibitemShut {NoStop}%
\bibitem [{\citenamefont {Schwarz}\ \emph {et~al.}(1978)\citenamefont {Schwarz}
  \emph {et~al.}}]{schwarz1978estimating}%
  \BibitemOpen
  \bibfield  {author} {\bibinfo {author} {\bibfnamefont {G.}~\bibnamefont
  {Schwarz}} \emph {et~al.},\ }\href@noop {} {\bibfield  {journal} {\bibinfo
  {journal} {The annals of statistics}\ }\textbf {\bibinfo {volume} {6}},\
  \bibinfo {pages} {461} (\bibinfo {year} {1978})}\BibitemShut {NoStop}%
\bibitem [{\citenamefont {Van~der Heyden}\ \emph {et~al.}(1998)\citenamefont
  {Van~der Heyden}, \citenamefont {Diks}, \citenamefont {Hoekstra},\ and\
  \citenamefont {DeGoede}}]{van1998testing}%
  \BibitemOpen
  \bibfield  {author} {\bibinfo {author} {\bibfnamefont {M.~J.}\ \bibnamefont
  {Van~der Heyden}}, \bibinfo {author} {\bibfnamefont {C.~G.}\ \bibnamefont
  {Diks}}, \bibinfo {author} {\bibfnamefont {B.~P.}\ \bibnamefont {Hoekstra}},
  \ and\ \bibinfo {author} {\bibfnamefont {J.}~\bibnamefont {DeGoede}},\
  }\href@noop {} {\bibfield  {journal} {\bibinfo  {journal} {Physica D:
  Nonlinear Phenomena}\ }\textbf {\bibinfo {volume} {117}},\ \bibinfo {pages}
  {299} (\bibinfo {year} {1998})}\BibitemShut {NoStop}%
\bibitem [{\citenamefont {Katz}(1981)}]{katz1981some}%
  \BibitemOpen
  \bibfield  {author} {\bibinfo {author} {\bibfnamefont {R.~W.}\ \bibnamefont
  {Katz}},\ }\href@noop {} {\bibfield  {journal} {\bibinfo  {journal}
  {Technometrics}\ }\textbf {\bibinfo {volume} {23}},\ \bibinfo {pages} {243}
  (\bibinfo {year} {1981})}\BibitemShut {NoStop}%
\bibitem [{\citenamefont {Papapetrou}\ and\ \citenamefont
  {Kugiumtzis}(2016)}]{papapetrou2016markov}%
  \BibitemOpen
  \bibfield  {author} {\bibinfo {author} {\bibfnamefont {M.}~\bibnamefont
  {Papapetrou}}\ and\ \bibinfo {author} {\bibfnamefont {D.}~\bibnamefont
  {Kugiumtzis}},\ }\href@noop {} {\bibfield  {journal} {\bibinfo  {journal}
  {Simulation Modelling Practice and Theory}\ }\textbf {\bibinfo {volume}
  {61}},\ \bibinfo {pages} {1} (\bibinfo {year} {2016})}\BibitemShut {NoStop}%
\bibitem [{\citenamefont {Granger}(1969)}]{granger1969investigating}%
  \BibitemOpen
  \bibfield  {author} {\bibinfo {author} {\bibfnamefont {C.~W.}\ \bibnamefont
  {Granger}},\ }\href@noop {} {\bibfield  {journal} {\bibinfo  {journal}
  {Econometrica: Journal of the Econometric Society}\ ,\ \bibinfo {pages}
  {424}} (\bibinfo {year} {1969})}\BibitemShut {NoStop}%
\bibitem [{\citenamefont {Jacobs}\ and\ \citenamefont
  {Lewis}(1978)}]{jacobs1978discrete}%
  \BibitemOpen
  \bibfield  {author} {\bibinfo {author} {\bibfnamefont {P.~A.}\ \bibnamefont
  {Jacobs}}\ and\ \bibinfo {author} {\bibfnamefont {P.~A.}\ \bibnamefont
  {Lewis}},\ }\href@noop {} {\emph {\bibinfo {title} {Discrete Time Series
  Generated by Mixtures. III. Autoregressive Processes (DAR (p)).}}},\ \bibinfo
  {type} {Tech. Rep.}\ (\bibinfo  {institution} {NAVAL POSTGRADUATE SCHOOL
  MONTEREY CALIF},\ \bibinfo {year} {1978})\BibitemShut {NoStop}%
\bibitem [{\citenamefont {Pearl}(1982)}]{pearl1982reverend}%
  \BibitemOpen
  \bibfield  {author} {\bibinfo {author} {\bibfnamefont {J.}~\bibnamefont
  {Pearl}},\ }\href@noop {} {\emph {\bibinfo {title} {Reverend Bayes on
  inference engines: A distributed hierarchical approach}}}\ (\bibinfo
  {publisher} {Cognitive Systems Laboratory, School of Engineering and Applied
  Science, University of California, Los Angeles},\ \bibinfo {year}
  {1982})\BibitemShut {NoStop}%
\bibitem [{\citenamefont {Yedidia}\ \emph {et~al.}(2003)\citenamefont
  {Yedidia}, \citenamefont {Freeman},\ and\ \citenamefont
  {Weiss}}]{yedidia2003understanding}%
  \BibitemOpen
  \bibfield  {author} {\bibinfo {author} {\bibfnamefont {J.~S.}\ \bibnamefont
  {Yedidia}}, \bibinfo {author} {\bibfnamefont {W.~T.}\ \bibnamefont
  {Freeman}}, \ and\ \bibinfo {author} {\bibfnamefont {Y.}~\bibnamefont
  {Weiss}},\ }\href@noop {} {\bibfield  {journal} {\bibinfo  {journal}
  {Exploring artificial intelligence in the new millennium}\ }\textbf {\bibinfo
  {volume} {8}},\ \bibinfo {pages} {236} (\bibinfo {year} {2003})}\BibitemShut
  {NoStop}%
\bibitem [{\citenamefont {Weiss}\ and\ \citenamefont
  {Freeman}(2000)}]{weiss2000correctness}%
  \BibitemOpen
  \bibfield  {author} {\bibinfo {author} {\bibfnamefont {Y.}~\bibnamefont
  {Weiss}}\ and\ \bibinfo {author} {\bibfnamefont {W.~T.}\ \bibnamefont
  {Freeman}},\ }in\ \href@noop {} {\emph {\bibinfo {booktitle} {Advances in
  neural information processing systems}}}\ (\bibinfo {year} {2000})\ pp.\
  \bibinfo {pages} {673--679}\BibitemShut {NoStop}%
\bibitem [{\citenamefont {Baron}\ \emph {et~al.}(2009)\citenamefont {Baron},
  \citenamefont {Sarvotham},\ and\ \citenamefont
  {Baraniuk}}]{baron2009bayesian}%
  \BibitemOpen
  \bibfield  {author} {\bibinfo {author} {\bibfnamefont {D.}~\bibnamefont
  {Baron}}, \bibinfo {author} {\bibfnamefont {S.}~\bibnamefont {Sarvotham}}, \
  and\ \bibinfo {author} {\bibfnamefont {R.~G.}\ \bibnamefont {Baraniuk}},\
  }\href@noop {} {\bibfield  {journal} {\bibinfo  {journal} {IEEE Transactions
  on Signal Processing}\ }\textbf {\bibinfo {volume} {58}},\ \bibinfo {pages}
  {269} (\bibinfo {year} {2009})}\BibitemShut {NoStop}%
\bibitem [{\citenamefont {Ihler}\ and\ \citenamefont
  {McAllester}(2009)}]{ihler2009particle}%
  \BibitemOpen
  \bibfield  {author} {\bibinfo {author} {\bibfnamefont {A.}~\bibnamefont
  {Ihler}}\ and\ \bibinfo {author} {\bibfnamefont {D.}~\bibnamefont
  {McAllester}},\ }in\ \href@noop {} {\emph {\bibinfo {booktitle} {Artificial
  Intelligence and Statistics}}}\ (\bibinfo {year} {2009})\ pp.\ \bibinfo
  {pages} {256--263}\BibitemShut {NoStop}%
\bibitem [{\citenamefont {Felzenszwalb}\ and\ \citenamefont
  {Huttenlocher}(2006)}]{felzenszwalb2006efficient}%
  \BibitemOpen
  \bibfield  {author} {\bibinfo {author} {\bibfnamefont {P.~F.}\ \bibnamefont
  {Felzenszwalb}}\ and\ \bibinfo {author} {\bibfnamefont {D.~P.}\ \bibnamefont
  {Huttenlocher}},\ }\href@noop {} {\bibfield  {journal} {\bibinfo  {journal}
  {International journal of computer vision}\ }\textbf {\bibinfo {volume}
  {70}},\ \bibinfo {pages} {41} (\bibinfo {year} {2006})}\BibitemShut {NoStop}%
\bibitem [{\citenamefont {Lokhov}\ \emph {et~al.}(2014)\citenamefont {Lokhov},
  \citenamefont {M{\'e}zard}, \citenamefont {Ohta},\ and\ \citenamefont
  {Zdeborov{\'a}}}]{lokhov2014inferring}%
  \BibitemOpen
  \bibfield  {author} {\bibinfo {author} {\bibfnamefont {A.~Y.}\ \bibnamefont
  {Lokhov}}, \bibinfo {author} {\bibfnamefont {M.}~\bibnamefont {M{\'e}zard}},
  \bibinfo {author} {\bibfnamefont {H.}~\bibnamefont {Ohta}}, \ and\ \bibinfo
  {author} {\bibfnamefont {L.}~\bibnamefont {Zdeborov{\'a}}},\ }\href@noop {}
  {\bibfield  {journal} {\bibinfo  {journal} {Physical Review E}\ }\textbf
  {\bibinfo {volume} {90}},\ \bibinfo {pages} {012801} (\bibinfo {year}
  {2014})}\BibitemShut {NoStop}%
\bibitem [{\citenamefont {Yedidia}\ \emph {et~al.}(2005)\citenamefont
  {Yedidia}, \citenamefont {Freeman},\ and\ \citenamefont
  {Weiss}}]{yedidia2005constructing}%
  \BibitemOpen
  \bibfield  {author} {\bibinfo {author} {\bibfnamefont {J.~S.}\ \bibnamefont
  {Yedidia}}, \bibinfo {author} {\bibfnamefont {W.~T.}\ \bibnamefont
  {Freeman}}, \ and\ \bibinfo {author} {\bibfnamefont {Y.}~\bibnamefont
  {Weiss}},\ }\href@noop {} {\bibfield  {journal} {\bibinfo  {journal} {IEEE
  Transactions on information theory}\ }\textbf {\bibinfo {volume} {51}},\
  \bibinfo {pages} {2282} (\bibinfo {year} {2005})}\BibitemShut {NoStop}%
\bibitem [{\citenamefont {Opper}\ \emph {et~al.}(2001)\citenamefont {Opper},
  \citenamefont {Winther} \emph {et~al.}}]{opper2001naive}%
  \BibitemOpen
  \bibfield  {author} {\bibinfo {author} {\bibfnamefont {M.}~\bibnamefont
  {Opper}}, \bibinfo {author} {\bibfnamefont {O.}~\bibnamefont {Winther}},
  \emph {et~al.},\ }\href@noop {} {\bibfield  {journal} {\bibinfo  {journal}
  {Advanced mean field methods: theory and practice}\ ,\ \bibinfo {pages} {7}}
  (\bibinfo {year} {2001})}\BibitemShut {NoStop}%
\bibitem [{\citenamefont {Kabashima}(2003)}]{kabashima2003cdma}%
  \BibitemOpen
  \bibfield  {author} {\bibinfo {author} {\bibfnamefont {Y.}~\bibnamefont
  {Kabashima}},\ }\href@noop {} {\bibfield  {journal} {\bibinfo  {journal}
  {Journal of Physics A: Mathematical and General}\ }\textbf {\bibinfo {volume}
  {36}},\ \bibinfo {pages} {11111} (\bibinfo {year} {2003})}\BibitemShut
  {NoStop}%
\bibitem [{\citenamefont {Neirotti}\ and\ \citenamefont
  {Saad}(2005)}]{neirotti2005improved}%
  \BibitemOpen
  \bibfield  {author} {\bibinfo {author} {\bibfnamefont {J.~P.}\ \bibnamefont
  {Neirotti}}\ and\ \bibinfo {author} {\bibfnamefont {D.}~\bibnamefont
  {Saad}},\ }\href@noop {} {\bibfield  {journal} {\bibinfo  {journal} {EPL
  (Europhysics Letters)}\ }\textbf {\bibinfo {volume} {71}},\ \bibinfo {pages}
  {866} (\bibinfo {year} {2005})}\BibitemShut {NoStop}%
\bibitem [{\citenamefont {Murphy}\ \emph {et~al.}(1999)\citenamefont {Murphy},
  \citenamefont {Weiss},\ and\ \citenamefont {Jordan}}]{murphy1999loopy}%
  \BibitemOpen
  \bibfield  {author} {\bibinfo {author} {\bibfnamefont {K.~P.}\ \bibnamefont
  {Murphy}}, \bibinfo {author} {\bibfnamefont {Y.}~\bibnamefont {Weiss}}, \
  and\ \bibinfo {author} {\bibfnamefont {M.~I.}\ \bibnamefont {Jordan}},\ }in\
  \href@noop {} {\emph {\bibinfo {booktitle} {Proceedings of the Fifteenth
  conference on Uncertainty in artificial intelligence}}}\ (\bibinfo
  {organization} {Morgan Kaufmann Publishers Inc.},\ \bibinfo {year} {1999})\
  pp.\ \bibinfo {pages} {467--475}\BibitemShut {NoStop}%
\bibitem [{\citenamefont {Yedidia}\ \emph {et~al.}(2001)\citenamefont
  {Yedidia}, \citenamefont {Freeman},\ and\ \citenamefont
  {Weiss}}]{yedidia2001generalized}%
  \BibitemOpen
  \bibfield  {author} {\bibinfo {author} {\bibfnamefont {J.~S.}\ \bibnamefont
  {Yedidia}}, \bibinfo {author} {\bibfnamefont {W.~T.}\ \bibnamefont
  {Freeman}}, \ and\ \bibinfo {author} {\bibfnamefont {Y.}~\bibnamefont
  {Weiss}},\ }in\ \href@noop {} {\emph {\bibinfo {booktitle} {Advances in
  neural information processing systems}}}\ (\bibinfo {year} {2001})\ pp.\
  \bibinfo {pages} {689--695}\BibitemShut {NoStop}%
\bibitem [{\citenamefont {Ihler}\ \emph {et~al.}(2005)\citenamefont {Ihler},
  \citenamefont {John~III},\ and\ \citenamefont {Willsky}}]{ihler2005loopy}%
  \BibitemOpen
  \bibfield  {author} {\bibinfo {author} {\bibfnamefont {A.~T.}\ \bibnamefont
  {Ihler}}, \bibinfo {author} {\bibfnamefont {W.~F.}\ \bibnamefont {John~III}},
  \ and\ \bibinfo {author} {\bibfnamefont {A.~S.}\ \bibnamefont {Willsky}},\
  }\href@noop {} {\bibfield  {journal} {\bibinfo  {journal} {Journal of Machine
  Learning Research}\ }\textbf {\bibinfo {volume} {6}},\ \bibinfo {pages} {905}
  (\bibinfo {year} {2005})}\BibitemShut {NoStop}%
\bibitem [{\citenamefont {Cantwell}\ and\ \citenamefont
  {Newman}(2019)}]{cantwell2019message}%
  \BibitemOpen
  \bibfield  {author} {\bibinfo {author} {\bibfnamefont {G.~T.}\ \bibnamefont
  {Cantwell}}\ and\ \bibinfo {author} {\bibfnamefont {M.~E.~J.}\ \bibnamefont
  {Newman}},\ }\href@noop {} {\enquote {\bibinfo {title} {Message passing on
  networks with loops},}\ } (\bibinfo {year} {2019}),\ \Eprint
  {http://arxiv.org/abs/1907.08252} {arXiv:1907.08252 [cs.SI]} \BibitemShut
  {NoStop}%
\bibitem [{\citenamefont {Passarino}\ and\ \citenamefont
  {Veltman}(1979)}]{passarino1979one}%
  \BibitemOpen
  \bibfield  {author} {\bibinfo {author} {\bibfnamefont {G.}~\bibnamefont
  {Passarino}}\ and\ \bibinfo {author} {\bibfnamefont {M.}~\bibnamefont
  {Veltman}},\ }\href@noop {} {\bibfield  {journal} {\bibinfo  {journal}
  {Nuclear Physics B}\ }\textbf {\bibinfo {volume} {160}},\ \bibinfo {pages}
  {151} (\bibinfo {year} {1979})}\BibitemShut {NoStop}%
\bibitem [{\citenamefont {t~Hooft}\ and\ \citenamefont
  {Veltman}(1974)}]{t1974one}%
  \BibitemOpen
  \bibfield  {author} {\bibinfo {author} {\bibfnamefont {G.}~\bibnamefont
  {t~Hooft}}\ and\ \bibinfo {author} {\bibfnamefont {M.}~\bibnamefont
  {Veltman}},\ }in\ \href@noop {} {\emph {\bibinfo {booktitle} {Annales de
  l'IHP Physique th{\'e}orique}}},\ Vol.~\bibinfo {volume} {20}\ (\bibinfo
  {year} {1974})\ pp.\ \bibinfo {pages} {69--94}\BibitemShut {NoStop}%
\bibitem [{\citenamefont {Mostepanenko}\ and\ \citenamefont
  {Trunov}(1997)}]{mostepanenko1997casimir}%
  \BibitemOpen
  \bibfield  {author} {\bibinfo {author} {\bibfnamefont {V.~M.}\ \bibnamefont
  {Mostepanenko}}\ and\ \bibinfo {author} {\bibfnamefont {N.}~\bibnamefont
  {Trunov}},\ }\href@noop {} {\emph {\bibinfo {title} {The Casimir effect and
  its applications}}}\ (\bibinfo  {publisher} {Oxford University Press},\
  \bibinfo {year} {1997})\BibitemShut {NoStop}%
\bibitem [{\citenamefont {Jaffe}(2005)}]{jaffe2005casimir}%
  \BibitemOpen
  \bibfield  {author} {\bibinfo {author} {\bibfnamefont {R.}~\bibnamefont
  {Jaffe}},\ }\href@noop {} {\bibfield  {journal} {\bibinfo  {journal}
  {Physical Review D}\ }\textbf {\bibinfo {volume} {72}},\ \bibinfo {pages}
  {021301} (\bibinfo {year} {2005})}\BibitemShut {NoStop}%
\bibitem [{\citenamefont {Russo}\ \emph {et~al.}(1992)\citenamefont {Russo},
  \citenamefont {Susskind},\ and\ \citenamefont {Thorlacius}}]{russo1992end}%
  \BibitemOpen
  \bibfield  {author} {\bibinfo {author} {\bibfnamefont {J.~G.}\ \bibnamefont
  {Russo}}, \bibinfo {author} {\bibfnamefont {L.}~\bibnamefont {Susskind}}, \
  and\ \bibinfo {author} {\bibfnamefont {L.}~\bibnamefont {Thorlacius}},\
  }\href@noop {} {\bibfield  {journal} {\bibinfo  {journal} {Physical Review
  D}\ }\textbf {\bibinfo {volume} {46}},\ \bibinfo {pages} {3444} (\bibinfo
  {year} {1992})}\BibitemShut {NoStop}%
\bibitem [{\citenamefont {Czarnecki}\ \emph {et~al.}(2005)\citenamefont
  {Czarnecki}, \citenamefont {Jentschura},\ and\ \citenamefont
  {Pachucki}}]{czarnecki2005calculation}%
  \BibitemOpen
  \bibfield  {author} {\bibinfo {author} {\bibfnamefont {A.}~\bibnamefont
  {Czarnecki}}, \bibinfo {author} {\bibfnamefont {U.~D.}\ \bibnamefont
  {Jentschura}}, \ and\ \bibinfo {author} {\bibfnamefont {K.}~\bibnamefont
  {Pachucki}},\ }\href@noop {} {\bibfield  {journal} {\bibinfo  {journal}
  {Physical review letters}\ }\textbf {\bibinfo {volume} {95}},\ \bibinfo
  {pages} {180404} (\bibinfo {year} {2005})}\BibitemShut {NoStop}%
\bibitem [{\citenamefont {Goh}\ and\ \citenamefont
  {Barab{\'a}si}(2008)}]{goh2008burstiness}%
  \BibitemOpen
  \bibfield  {author} {\bibinfo {author} {\bibfnamefont {K.-I.}\ \bibnamefont
  {Goh}}\ and\ \bibinfo {author} {\bibfnamefont {A.-L.}\ \bibnamefont
  {Barab{\'a}si}},\ }\href@noop {} {\bibfield  {journal} {\bibinfo  {journal}
  {EPL (Europhysics Letters)}\ }\textbf {\bibinfo {volume} {81}},\ \bibinfo
  {pages} {48002} (\bibinfo {year} {2008})}\BibitemShut {NoStop}%
\end{thebibliography}
%

\end{document}